\documentclass[12pt]{article}

\usepackage{geometry}
\usepackage{amsmath}    
\usepackage{enumitem}
\usepackage{amssymb}    
\usepackage{amsthm}     
\usepackage{bm}         
\usepackage{mathtools}  
\usepackage{subcaption}
\usepackage{booktabs}
\usepackage{threeparttable}
\usepackage{array}
\usepackage{caption}
\usepackage[T1]{fontenc}
\usepackage[utf8]{inputenc}
\usepackage{lmodern}
\usepackage{xcolor}
\usepackage{geometry}   
\usepackage{hyperref}   
\usepackage{enumitem}   
\usepackage{natbib}     
\usepackage{setspace}
\usepackage{todonotes}
\usepackage{comment}
\newtheorem{theorem}{Theorem}[section]
\newtheorem{lemma}[theorem]{Lemma}
\newtheorem{remark}{Remark}[section]
\newtheorem{assumption}{Assumption}[section]
\newtheorem{corollary}{Corollary}[section]

\newtheorem{example}{Example}
\counterwithin{equation}{section}

\newcommand{\yuya}[1]{\todo[inline,color=orange!50]{\textbf{Yuya:  }#1}}

\title{\setlength{\baselineskip}{10mm}Inference for High-Dimensional Network Data\thanks{\setlength{\baselineskip}{4.4mm}We benefited from useful comments by participants in 2026 New York Camp Econometrics. All remaining errors are ours. Sasaki gratefully acknowledges the generous research support of Brian and Charlotte Grove.}}
\author{Yuya Sasaki\thanks{Yuya Sasaki. Department of Economics, Vanderbilt University. \texttt{yuya.sasaki@vanderbilt.edu}} \qquad Baoning Zheng\thanks{Baoning Zheng. Department of Economics, Vanderbilt University. \texttt{baoning.zheng@vanderbilt.edu}}}
\date{}

\begin{document}
\maketitle
\begin{abstract}
\begin{spacing}{1.15}
We develop a novel method of inference for network-dependent high-dimensional random vectors. Dependence is characterized via a functional dependence measure based on graph distance, allowing the approximation theory to capture the interaction between the decay of dependence and the growth of network neighborhoods. We establish Gaussian approximation results for the maximum norm under finite-moment and sub-Weibull conditions, providing explicit conditions under which the dimension may increase with the network size. We also propose a high-dimensional network HAC covariance estimator and establish its convergence properties, yielding a feasible procedure for simultaneous inference. Simulation studies demonstrate favorable finite-sample performance of the proposed method. We apply the procedure to study how spillover effects vary with an index of network homophily by constructing confidence bands for the conditional spillover-effect function. The application reveals heterogeneity and local significance that would be obscured by conventional low-dimensional inference.


\end{spacing}
\end{abstract}

\section{Introduction}
Network-dependent data arise in a wide range of empirical settings in which individuals, firms, regions, markets, or institutions are connected. Prominent examples include social interactions and peer effects, production and supply-chain networks, financial networks, trade networks, and interference in experiments conducted on social or geographic networks. In these settings, statistical dependence is governed by graph distance. Units that are close in the network may exhibit strong dependence, while the influence between units may decay as their graph distance increases. An important feature of network dependence is that its aggregate strength is determined jointly by the rate at which dependence decays with graph distance and the rate at which network neighborhoods grow.

While a growing literature has developed asymptotic theory for network-dependent data, much of the existing theory concerns low-dimensional statistics. At the same time, many empirical problems require simultaneous inference for high-dimensional parameters. Examples include inference for functional treatment-effect parameters (e.g., conditional average treatment effects or CATEs, continuous treatment effects, etc.) and inference for structural parameters that are (partially) identified by conditional moment (in)equalities. High-dimensional Gaussian approximation provides a powerful foundation for such inference, but existing results are primarily designed for independent observations, time series, spatial random fields, local dependence structures, or exchangeable arrays. These results do not directly accommodate dependence that propagates over a general observed network, where the accumulation of dependence is governed by both graph distance and network topology.

This paper develops high-dimensional Gaussian approximation and a method of simultaneous inference for sums of network-dependent random vectors. We consider observations indexed by the nodes of an observed network and characterize their dependence through a functional dependence measure based on perturbations of primitive shocks. The effect of a shock on an observation decays with their graph distance, while the network itself determines how many potentially dependent observations occur at each distance. This formulation provides a direct way to combine the strength and spatial propagation of dependence with the topology of the network.

Our main contribution is to establish Gaussian approximation for the maximum of a high-dimensional sum of network-dependent observations. We allow the dimension of the random vectors to increase with the network size and provide results under both finite-moment and sub-Weibull conditions. The approximation theory explicitly accommodates the interaction among the dimension of the statistic, moment conditions, decay of functional dependence, and growth of network neighborhoods. The proofs combine functional-dependence coupling and localization with a partition of the network into approximately independent interior clusters separated by buffer regions. This construction reduces the network-dependent problem to a high-dimensional Gaussian approximation for independent cluster sums while controlling the errors introduced by localization, buffers, and cross-cluster dependence.

Our second contribution is to make the Gaussian approximation feasible for simultaneous inference. We develop a high-dimensional network HAC covariance estimator in which observations are weighted according to graph distance and establish its consistency in high dimensions. Combining the covariance estimator with the Gaussian approximation yields feasible critical values for the maximum statistic, and hence simultaneous confidence intervals and hypothesis testing for high-dimensional network statistics.

Our framework is related to several strands of the literature. A large literature studies cross-sectional dependence through spatial random fields. When dependence is driven by geographic distance or a similar spatial metric, random fields indexed by Euclidean locations provide a natural framework. Important contributions include \cite{conley1999gmm}, \cite{KelejianPrucha2007}, \cite{kim2011spatial}, \cite{BesterConleyHansen2011}, \cite{JenishPrucha2009}, and \cite{jenish2012spatial}. Our setting instead takes the observed network and its graph distance as the primitive structure governing statistical dependence.

More closely related is the recent literature on network dependence. \cite{KMS2021} introduce a general framework based on $\psi$-dependence and graph distance, and establish LLNs, CLTs, and network HAC variance estimation. Building on this framework, \cite{Leung2022} studies causal inference under interference, while \cite{Leung2023} characterizes the validity of cluster-robust inference. See also \cite{kojevnikov2021bootstrap}, \cite{johnsson2021estimation}, and \cite{LeeSong2019} for related inference methods. This literature, however, has not investigated Gaussian approximation or inference for high-dimensional vectors. Our paper complements this literature
by developing such methods under network dependence.

Recently, \cite{gao2026coupling} study graph-dependent data and establish novel maximal inequalities for empirical processes. Our paper provides a complementary contribution with a different objective. While \cite{gao2026coupling} develop maximal inequalities for uniformly controlling empirical processes, our focus is on high-dimensional Gaussian approximation for processes that need not be stochastically equicontinuous.

Our work is also related to limit theory for cross-sectionally dependent data based on martingale and mixingale methods. \cite{KuersteinerPrucha2013} establish a CLT for martingale difference sequences. \cite{kuersteiner2020dynamic} incorporates network dependence into a martingale framework, while \cite{kuersteiner2019limit} develops conditional spatial mixingale limit theory accommodating stochastic network formation. Related results based on dependency graphs and local dependence include \cite{stein1972bound}, \cite{janson1988normal}, \cite{BaldiRinott1989}, \cite{rinott1996multivariate}, and \cite{chen2004normal}; applications to network data include \cite{AronowSamii2017}, \cite{Leung2020}, and \cite{song2018measuring}.

Our theoretical results build on the literature on high-dimensional Gaussian approximation. For independent random vectors and empirical processes, important contributions include \cite{CCK2013,CCK2014,chernozhukov2015comparison,CCK2016,CCK2017}, \cite{CCK2019}, \cite{deng2020beyond}, \cite{KuchibhotlaMukherjeeBanerjee2021}, and \cite{chernozhukov2023nearly}. High-dimensional Gaussian approximation has subsequently been extended to several dependent-data settings. For time series, \cite{ZW2017} establish Gaussian approximation under the physical dependence measure \citep{wu2005nonlinear}, while \cite{ZhangCheng2018} derive related results under functional dependence and develop inference based on the non-overlapping block bootstrap. For local dependence, \cite{fang2021high} establish high-dimensional CLTs under dependency graph structures. \cite{chiang2023inference} establish a high-dimensional Gaussian approximation for exchangeable arrays. For spatial dependence, \cite{kurisu2024gaussian} extend high-dimensional Gaussian approximation to spatial random fields. Our results complement this literature by allowing dependence to propagate over an observed network and by making explicit how network neighborhood growth interacts with dependence decay in determining the validity of high-dimensional Gaussian approximation.

We investigate the finite-sample performance of the proposed methods through Monte Carlo simulations. The empirical application examines heterogeneous spillover effects and demonstrates how the proposed procedure can be used to construct simultaneous confidence bands for network-based treatment- and spillover-effect parameters.

\medskip\noindent
{\bf Organization:}
Section \ref{sec:setup} introduces the setting, the statistic of interest, and analysis tools. Section \ref{sec:GA} presents the main Gaussian approximation results. Section \ref{sec:estimation-cov} develops a method of inference. Sections \ref{sec:simGA} and \ref{sec:HAC} study the finite sample performance through Monte Carlo simulations. Section \ref{sec:application} applies the method to heterogeneous spillover analysis in a network experiment. Proofs and additional technical details are collected in the appendix.

\medskip\noindent
{\bf Notation:}
For a random variable $X$ and $q>0$, 
we write $X \in \mathcal{L}^q$ if
$
\|X\|_q := \big( \mathbb{E}|X|^q \big)^{1/q} < \infty,
$
and for a vector $v = (v_1,\dots,v_p)^\top$, let the norm-$s$ length be
$|v|_s = \left( \sum_{j=1}^p |v_j|^s \right)^{1/s}$ for $s \ge 1.$
Write the $p \times p$ identity matrix as $\mathrm{Id}_p$. 
For two sequences of positive numbers $(a_n)$ and $(b_n)$, 
we write $a_n \asymp b_n$ (resp., $a_n \lesssim b_n$ or 
$a_n \ll b_n$) if there exists some constant $C>0$ such that
$C^{-1} \le \frac{a_n}{b_n} \le C
\quad
(\text{resp., } \frac{a_n}{b_n} \le C 
\text{ or } \frac{a_n}{b_n} \to 0)$
for all large $n$. 
We use $C, C_1, C_2, \dots$ to denote positive constants whose values may differ from place to place. 
A constant with a symbolic subscript is used to emphasize the dependence of the value on the subscript. 
We allow $p = p_n$ to increase with $n$,
and all asymptotic statements are taken as $n\to\infty$.

\section{Network Dependence Framework}{\label{sec:setup}}
\subsection{Network Topology}
Suppose that a researcher observes a set of cross-sectional units, denoted by $V_n=\{1,\dots,n\}$.
Modeling statistical dependence across these units typically requires a notion of distance.
In time-series applications, dependence is indexed by temporal distance, whereas in spatial applications, it is typically indexed by the Euclidean distance.
In the current paper, we consider dependence induced by the topology of an observed network.

Specifically, suppose that a researcher observes an undirected network $\mathcal{G}_n=(V_n,E_n)$ on $V_n$, where
$
E_n\subseteq \{\{i,u\}: i,u\in V_n,\ i\neq u\}
$
denotes the set of links.
For any $i,u\in V_n$, let $d_n(i,u)$ denote the distance between units $i$ and $u$ in $\mathcal{G}_n$, defined as the length of the shortest path connecting them.
The resulting function $d_n$ defines a metric on $V_n$ and serves as the fundamental notion of proximity throughout the paper.

Define the volume neighborhood $\mathcal N_n(i,r)$ as the set of nodes within distance $r$ of node $i$, and the shell neighborhood $\mathcal N_n^\partial(i,r)$ as the set of nodes at exactly distance $r$ from node $i$:
\[
\mathcal N_n(i,r)
=
\{u\in V_n:d_n(i,u)\le r\}
\qquad\text{and}\qquad
\mathcal N_n^\partial(i,r)
=
\{u\in V_n:d_n(i,u)=r\}.
\]
Their cardinality will be denoted by
\begin{equation}\label{eq:neighbor size}
N_n(i,r)
=
|\mathcal N_n(i,r)|
\qquad\text{and}\qquad
N_n^\partial(i,r)
=
|\mathcal N_n^\partial(i,r)|.
\end{equation}
The sizes of these local neighborhoods play a central role in our analysis.
Because dependence propagates through the network, the cumulative dependence surrounding a node depends not only on the strength of dependence but also on the rate at which neighborhoods expand.

\subsection{ Functional Dependence}

Network dependence can be characterized through a variety of weak dependence notions, including mixing conditions, dependency graphs, and functional dependence.\footnote{These alternative approaches are also encompassed by the $\psi$-dependence framework of \citet{kojevnikov2021bootstrap}.}
Among these alternatives, we focus on the functional dependence approach, originally introduced by \citet{wu2005nonlinear}, as it provides a convenient framework for deriving localization bounds, nonasymptotic probability inequalities, and high-dimensional Gaussian approximations required for our analysis.
Related functional-dependence ideas have been widely used in the high-dimensional time-series literature. See \citet{ZW2017} and \citet{ZhangCheng2018} for example.

The functional dependence framework  represents each observation as a function of primitive shocks and measures dependence through the effects of shock perturbations.
In the network setting, this formulation is particularly natural, as the impact of a shock can be allowed to decay with its graph distance from the observation.
Specifically, let
\[
\varepsilon_n:=\{\varepsilon_{n,u}:u\in V_n\}
\]
be a collection of mutually independent primitive shocks indexed by the nodes of the network.
Each observation is generated as a measurable function of the entire shock field,
\[
X_{n,i}=H_{n,i}(\varepsilon_n)\in\mathbb R^p,
\qquad
\mathbb E(X_{n,i})=0
\quad
\text{for all } i\in V_n,
\]
where $H_{n,i}$ is allowed to vary across units, thereby accommodating heterogeneity.
Write
\[
X_{n,i}
=
(X_{n,i1},\ldots,X_{n,ip})^\top
\]
for the coordinate representation of the $p$-dimensional vector $X_{n,i}$.
This framework encompasses a broad class of linear and nonlinear network processes. 

Throughout the paper, the network \(\mathcal{G}_n=(V_n,E_n)\) is treated as observed, and all probabilistic statements are understood conditionally on the realized network structure.
For any node \(u\in V_n\), let \(\varepsilon_n^{*(u)}\) denote the coupled shock field obtained by replacing \(\varepsilon_{n,u}\) with an independent copy \(\varepsilon'_{n,u}\) while leaving all remaining shocks unchanged. 
Define
\[
X_{n,i}^{*(u)}
:=
H_{n,i}(\varepsilon_n^{*(u)}).
\]

Following the functional dependence approach of \citet{ZW2017}, we measure the influence of the shock at node \(u\) on the \(j\)-th coordinate of observation \(i\) by the \(L_q\) distance
\[
\delta_{n,i,u,q,j}
:=
\left\|
X_{n,ij}
-
X^{*(u)}_{n,ij}
\right\|_q
\]
between \(X_{n,ij}\) and its coupled counterpart $X_{n,ij}^{\ast(u)}$.
This quantity \(\delta_{n,i,u,q,j}\) measures the sensitivity of \(X_{n,ij}\) to a perturbation of the shock at node \(u\).
Large values indicate that the shock at \(u\) has a substantial effect on observation \(i\).

Although the dependence measure may vary across units and shocks, network dependence is fundamentally governed by network distance. We therefore summarize the maximal influence of shocks from distance $s$ through the distance-wise functional dependence measure
\[
\delta_{n,s,q,j}
:=
\sup_{i\in V_n}
\sup_{u:d_n(i,u)=s}
\delta_{n,i,u,q,j}.
\]

To quantify the cumulative influence of distant shocks,
we  define the cumulative tail dependence envelop
\(\Delta_{n,m,q,j}\) and its high-dimensional envelope
\(\Psi_{n,q}(m)\) by
\begin{equation}\label{eq:psi}
\Delta_{n,m,q,j}
=
\sup_{i\in V_n}\sum_{s\ge m}
N_n^\partial(i,s)\,
\delta_{n,s,q,j},
\qquad
\Psi_{n,q}(m)
=
\max_{1\le j\le p}
\Delta_{n,m,q,j},
\end{equation}

The special case \(m=0\) plays a particularly important role. 
Define
\[
\Theta_{n,q,j}
:=
\Delta_{n,0,q,j}.
\]
Following \citet{ZW2017}, we refer to \(\Theta_{n,q,j}\) as the dependence-adjusted norm of the \(j\)-th coordinate, as it combines both moment magnitude and cumulative network dependence.

The quantities \(\delta_{n,s,q,j}\), \(\Delta_{n,m,q,j}\), and \(\Psi_{n,q}(m)\) will serve as the primary measures of network dependence throughout the paper.


\section{Gaussian Approximations}{\label{sec:GA}}
This section presents the main assumptions and results on Gaussian approximations.

To facilitate this objective, we introduce the notation
$$
T_X=\sum_{i \in V_n}X_{n,i},
\qquad\text{and}\qquad
\Sigma_n=\mathrm {Var}(T_X/\sqrt{n})=n^{-1}\sum_{i,i'\in V_n}\mathrm {Cov}(X_{n,i},X_{n,i'})
$$
for the network sum and scaled variance, respectively.
Let $\Sigma_0=\mathrm{diag}(\Sigma_n)=\mathrm{diag}(\sigma_{11},\cdots,\sigma_{pp})$ denote the diagonal matrix of $\Sigma_n$, and let $D_0=\Sigma_0^{1/2}$ be its square root. 
The location is normalized to $\mathbb E(X_{n,i})=0$ for all $i\in V_n$ throughout.

We aim to establish the Gaussian approximation:
\begin{equation}{\label{eq:GA}}
\rho_n 
:= \sup_{u \ge 0} 
\left| 
\mathbb{P}\!\left( 
\,\bigl| D_0^{-1} T_X/\sqrt{n} \bigr|_\infty \le u 
\right)
-
\mathbb{P}\!\left(
\bigl| D_0^{-1} Z \bigr|_\infty \le u
\right)
\right|
\to 0
\end{equation}
where $Z \sim N(0,\Sigma_n)$.

\subsection{Cluster Partition}
The key step in establishing the Gaussian approximation  \eqref{eq:GA} is to reduce the network-dependent statistic
to a sum of approximately independent cluster-level contributions.
Once such a representation is obtained, the Gaussian approximation theory for maxima of sums of independent random vectors \citep{CCK2013} becomes applicable. 

To achieve this goal, we partition the network into many balanced clusters and show that dependence across clusters are asymptotically negligible after localization.
Specifically, we suppose that there exists a squence of cluster sets satisfying the following assumption.

\begin{assumption}[Regular Cluster Partition]
\label{ass:clusters}
The network $V_n$ can be partitioned\footnote{
That is, $\{\mathcal C_g\}_{g=1}^{G_n}$ satisfies $\bigcup_{g=1}^{G_n}\mathcal C_g=V_n$ and $\mathcal C_{g_1}\cap\mathcal C_{g_2}=\varnothing$ for $g_1\neq g_2.
$} 
into clusters
$
\{\mathcal C_g\}_{g=1}^{G_n}
$
satisfying the following properties:
\begin{itemize}
    \item[(i)] $G_n\to\infty$ and $G_n=o(n)$.
    \item[(ii)] There exist a sequence $l_n\to\infty$ and constants $0<\underline c<\bar c<\infty$
such that $\underline c\,l_n \le l_g \le \bar c\,l_n$ for each $g=1,\ldots,G_n$, where $l_g := |\mathcal C_g|$.
\end{itemize}
\end{assumption}

Assumption \ref{ass:clusters} paves the way for a blocking structure that will be used for our analysis.
The required property $G_n\to\infty$ in part (i) of the assumption guarantees that the statistic of interest can be represented as a sum of many cluster-level contributions, while the balance condition (ii) prevents any single cluster from dominating the statistic.

Given $\{\mathcal C_g\}_{g=1}^{G_n}$, we define the boundary
\[
\partial\mathcal C_g
:=
\Bigl\{
i\in\mathcal C_g:
\exists\,k\notin\mathcal C_g
\text{ such that }
(i,k)\in E_n
\Bigr\}
\]
of a cluster $\mathcal C_g$.
This object will be used to quantify interactions across clusters.
The aggregate boundary size is denoted by
\[
\eta_n(G_n)
:=
\sum_{g=1}^{G_n}
|\partial\mathcal C_g|.
\]

The following examples illustrate how balanced cluster partitions can be constructed for commonly used network models.

\begin{example}[Ring Network]
\label{ex:ring}
Consider the ring graph
$\mathcal G_n=(V_n,E_n)$
where each node is connected to its two nearest index neighbors.
That is, node $1$ is connected to nodes $n$ and $2$, node $i$ is connected to nodes $i-1$ and $i+1$ for $i \in \{2,...,n-1\}$, and node $n$ is connected to nodes $n-1$ and $1$.
The graph distance can be written as
$d_n(i,i') = \min\{|i-i'|,\,n-|i-i'|\}.$

Partition the nodes into contiguous segments of equal
length $a_n$, i.e.,
\[
\mathcal C_g
=
\{(g-1)a_n+1,\ldots,ga_n\}
\qquad\text{and}\qquad
G_n={n}/{a_n}.
\]
By construction, all clusters have the same size
$
l_g=a_n.
$
Therefore, setting
$a_n\to\infty$
and
$a_n=o(n)$
satisfies Assumption~\ref{ass:clusters}.

Since each cluster has exactly two boundary nodes, we have
$|\partial \mathcal C_g|=2$
and
$\eta_n(G_n)=2G_n.$
\end{example}

\begin{example}[Regular Random Geometric Graph]
\label{ex:rgg}
Consider a random geometric graph
\(\mathcal G_n=(V_n,E_n)\),
where each node \(i\in V_n\) is associated with a random spatial location
$
U_{n,i}\in[0,1]^2.
$
Two nodes are connected whenever their Euclidean distance is no larger than a connection radius
$
r_n
=
c_r n^{-1/2}
$
for some fixed constant \(c_r>0\), that is,
\[
(i,i')\in E_n
\quad\Longleftrightarrow\quad
\lVert U_{n,i}-U_{n,i'}\rVert
\le r_n.
\]
The choice \(r_n\asymp n^{-1/2}\) makes the connection radius of the same order as the typical spacing between \(n\) nodes in a two-dimensional region.

Suppose further that there exist constants
\(c_{\mathrm{sep}},C_{\mathrm{Cov}}>0\), independent of \(n\), such that
\[
\min_{i\neq i'}
\lVert U_{n,i}-U_{n,i'}\rVert
\ge
c_{\mathrm{sep}}r_n
\qquad\text{and}\qquad
\sup_{x\in[0,1]^2}
\min_{i\in V_n}
\lVert x-U_{n,i}\rVert
\le
C_{\mathrm{Cov}}r_n.
\]
The first condition prevents excessive local crowding, while the second rules out spatial holes larger than the connection scale. Together, these two conditions allow irregular random configurations while ensuring that the node locations remain sufficiently regular at scale \(r_n\).

Partition the unit square into disjoint square regions
\(\{\mathcal R_g\}_{g=1}^{G_n}\)
with side length \(a_n\), and define
$
\mathcal C_g
=
\left\{
i\in V_n:
U_{n,i}\in\mathcal R_g
\right\}
$
and
$
G_n=a_n^{-2}.
$
Let
$
a_n\to0
$
and
$
\frac{a_n}{r_n}\to\infty.
$
Under the separation and covering conditions, the number of nodes contained in each square is of the same order as the square's area relative to the local spacing scale:
\[
l_g
=
|\mathcal C_g|
\asymp
\left(
\frac{a_n}{r_n}
\right)^2
\asymp
na_n^2
=
\frac{n}{G_n}
\]
uniformly over \(g\), with probability approaching one. Hence the resulting partition satisfies the balanced-cluster requirement in Assumption~\ref{ass:clusters}.

A node belongs to the cluster boundary if it lies within graph distance one of another cluster. Since every edge has Euclidean length at most \(r_n\), such nodes are contained in a spatial strip of width of order \(r_n\) along the boundary of the corresponding square region. The area of this strip is of order
$
a_n r_n.
$
By the same local-density regularity,
\[
|\partial\mathcal C_g|
\asymp
\frac{a_n r_n}{r_n^2}
=
\frac{a_n}{r_n}
\asymp
\left(
\frac{n}{G_n}
\right)^{1/2}
\]
uniformly over \(g\), with probability approaching one. Summing over all clusters gives
\[
\eta_n(G_n)
=
\sum_{g=1}^{G_n}
|\partial\mathcal C_g|
\asymp
n^{1/2}G_n^{1/2}.
\]
\end{example}

\subsection{ Localization and Approximate Independence}
Even if the network is partitioned into many clusters, observations belonging to different clusters generally remain dependent since shocks may propagate through arbitrarily long network paths.
In this section, we are going to localize such dependence inside the clusters and construct independent sums by imposing some conditions.

The first condition requires that the influence of a shock decays sufficiently fast with network distance, as formally stated below.
\begin{assumption}[Exponential Decay]
\label{ass:decay rate}
For some $q>4$, there exists $0<\rho<1$ such that
\[
\delta_{n,s,q,j}
\le
C_q\rho^s
\]
for all $s \ge 0$ and $j \in \{1,...,p\}$ for all $n$.
\end{assumption}

The second condition restricts the growth of network neighborhoods. 

\begin{assumption}[Polynomial Shell and Volume Growth]
\label{ass:volume}
There exist constants $C_V>0$ and $d\ge1$, such that
\[
N_n(r)
=
\sup_{i\in V_n}
N_n(i,r)
\le
C_V(1+r)^{d}
\]
for all $r \ge 0$ and $n$.
\end{assumption}

Assumptions \ref{ass:decay rate} and \ref{ass:volume} are motivated by the fact that cumulative dependence depends on both the decay of shock propagator and the number of nodes affected, respectively.
Assumption \ref{ass:volume} is satisfied by bounded-growth networks such as rings and regular random geometric graphs. 
However, it excludes graph sequences whose neighborhoods expand too rapidly relative to  the dependence decay.

\begin{remark}\label{rem:local}
Under Assumption \ref{ass:decay rate}  and \ref{ass:volume},
\begin{align*}
&\Theta_{n,q,j}
=
\Delta_{n,0,q,j}
\le
C_VC_q
\sum_{s\ge0}
(1+s)^d\rho^s
=
O(1)
\qquad\text{for all $j$,}\\
&\Psi_{n,q}(m)
:=
\max_{1\le j\le p}
\Delta_{n,m,q,j}
\lesssim
(1+m)^d\rho^m.
\end{align*}
\end{remark}
This bound motivates us to introduce the localized approximation
\begin{equation}{\label{eq:localize}}
X_{n,i}^{(m)}
=
\mathbb E
\bigl(
X_{n,i}
\mid
\mathcal F_{i,m}
\bigr),
\end{equation}
where
\[
\mathcal F_{i,m}
=
\sigma
\{
\varepsilon_{n,u}
:
d_n(i,u)\le m
\}.
\]
The localized approximation $X_{n,i}^{(m)}$ depends only on the shocks contained in the $m$-neighborhood of node $i$. 
When $m$ is chosen of logarithmic order in $np$, the localized approximation error
\(
X_{n,i}-X_{n,i}^{(m)}
\)
is asymptotically negligible. 
A formal statement is provided in Appendix \ref{sec:appendix}.

After this localization procedure, the dependence is effectively confined to local neighborhoods. 
Consequently, cluster sums constructed from observations whose $m$-neighborhoods are entirely contained within different clusters become independent.
This observation forms the main bases for our blocking argument used in our proofs of Theorems \ref{th:GA_q} and \ref{th:GA_sub} to be presented in the following subsection.

\subsection{Gaussian Approximation Results}\label{sec:gaussian_approximation}

The  assumptions in previous section ensure that localized observations can be organized into approximately independent cluster sums.
With this, we are now ready to apply the Gaussian approximation theory for independent observations \citep{CCK2013}, under the nondegeneracy condition on both the overall statistic and the cluster-level contributions.

\begin{assumption}[Variance Regularity]
\label{ass:variance lb}
There exist constants $c>0$ and $c_1>0$ such that
\begin{align*}
&\min_{1\le j\le p}\sigma_{jj}\ge c
\qquad\text{and}
\\
&\inf_{1\le g\le G_n}\inf_{1\le j\le p}
\mathrm{Var}\left(
l_g^{-1/2}\sum_{i\in \mathcal{C}_g}X_{n,ij}
\right)
\ge c_1.
\end{align*}
\end{assumption}

\begin{theorem}[Finite-Moment Gaussian Approximation]\label{th:GA_q}
Suppose that Assumption \ref{ass:clusters}, \ref{ass:decay rate}, \ref{ass:volume}, and \ref{ass:variance lb} hold.
If  the partition satisfies
\begin{align*}
\mathrm{(i)} \ & \ 
G_n
\gg
\max\{
{\log(np)}^{7},
p^{2/(q-2)}
{\log(np)}^{3q/(q-2)}\},
\\
\mathrm{(ii)} \ & \ \max_{1\le g\le G_n}
\frac{|\partial \mathcal C_g|
(\log(np))^{d}
G_n}{n}
\to0,
\qquad\text{and}\\
\mathrm{(iii)} \ & \ \eta_n(G_n)
(\log(np))^{d}
\ll
n p^{-2/q}(\log p)^{-1},
\end{align*}
then the Gaussian approximation \eqref{eq:GA} holds.
\end{theorem}

Restriction (i) in the statement of this theorem requires that the partition contains sufficiently many clusters, whereas
restrictions (ii) and (iii) require that the contribution of boundary observations is asymptotically negligible both at the cluster level and in aggregate, respectively.
These general conditions are stated at a high-level.
Thus, let us revisit Examples~\ref{ex:ring} and~\ref{ex:rgg}, and discuss sufficient conditions in the contexts of these examples.

\bigskip\noindent
{\bf Example~\ref{ex:ring}} (Ring Network){\bf, Revisited.}
Consider the ring network introduced in Example~\ref{ex:ring}. Since
$|\partial\mathcal C_g|=2$ and $\eta_n(G_n)=2G_n$,
the lower and upper bounds on \(G_n\)  given by conditions (i) and (iii) are
\[
G_n
\gg
\max\left\{
(\log np)^7,\,
p^{2/(q-2)}(\log np)^{3q/(q-2)}
\right\},
\]
and
\[
G_n
\ll
np^{-2/q}(\log np)^{-(d+1)}.
\]
For $p\ge n$,  a sufficient 
rate condition ensuring that a feasible choice of \(G_n\) exists is
\[
p^{\frac{2}{q-2}+\frac{2}{q}}
(\log p)^{d+1+3q/(q-2)}
=o(n).
\]

\bigskip\noindent
{\bf Example~\ref{ex:rgg}} (Regular Random Geometric Graph){\bf, Revisited.}
Consider the regular random geometric graph introduced in Example~\ref{ex:rgg}. Since
$
|\partial\mathcal C_g|
\asymp
\left({n}/{G_n}\right)^{1/2}
$ 
and
$
\eta_n(G_n)
\asymp
n^{1/2}G_n^{1/2},
$
the lower and upper bounds on \(G_n\)  given by conditions (i) and (iii) are
\[
G_n
\gg
\max\left\{
(\log np)^7,\,
p^{2/(q-2)}(\log np)^{3q/(q-2)}
\right\},
\]
and
\[
G_n
\ll
np^{-4/q}(\log np)^{-(2d+2)}.
\]
For $p\ge n$, a sufficient 
rate condition ensuring that a feasible choice of \(G_n\) exists is
\[
p^{\frac{2}{q-2}+\frac{4}{q}}
(\log p)^{2d+2+3q/(q-2)}
=o(n).
\]
\bigskip

Theorem \ref{th:GA_q}, and thus the discussions of the two examples above, focus on the case where the dependence-adjusted norm $\Theta_{n,q,j}$ exists for some $q>4$. This framework allows the dimension $p$ to grow with $n$ only at a polynomial rate. To have ultra high dimensionality for $p$, we can consider a stronger moment condition under which the dependence-adjusted norm $\Theta_{n,q,j}$ exists for all orders.
We extend Assumption \ref{ass:decay rate} to hold for all orders $q \ge 4$.

\begin{assumption}[Sub-Weibull Dependence Decay]
\label{ass:subweibull}
Assumption \ref{ass:decay rate} holds for all \(q\ge4\). 
Moreover, there exist constants \(C>0\), \(\nu>0\), and \(\rho \in (0,1)\) such that
\[
\delta_{n,s,q,j}
\le
C q^\nu \rho^s
\]
holds for all $s\ge0$, $q\ge4$, and $j \in \{1,...,p\}$ for all $n$.
\end{assumption}

The parameter \(\nu\) governs the growth rate of higher-order moments and characterizes the tail behavior. 
The next theorem states a counterpart of Theorem \ref{th:GA_q} for the case of sub-Weibull tails.

\begin{theorem}[Sub-Weibull Gaussian Approximation]{\label{th:GA_sub}}
Suppose that Assumption \ref{ass:clusters}, \ref{ass:volume}, \ref{ass:variance lb}, and \ref{ass:subweibull} hold. 
If the partition satisfies
\begin{align*}
\mathrm{(i)} \ & \ G_n\gg (\log(np))^{\max\{7,\,4+2\nu\}},
\\
\mathrm{(ii)} \ & \ \max_{1\le g\le G_n}
\frac{|\partial \mathcal C_g|
(\log(np))^{d}
G_n}{n}
\to0,
\qquad\text{and}\\
\mathrm{(iii)} \ & \ \eta_n(G_n)\left(\log(np)\right)^{d}
\ll
n(\log p)^{-2-2\nu}.
\end{align*}
then the Gaussian Approximation in \eqref{eq:GA} holds.
\end{theorem}

\bigskip\noindent
{\bf Example~\ref{ex:ring}} (Ring Network){\bf, Revisited.}
Consider again the ring network introduced in Example~\ref{ex:ring}.
Now, suppose the conditions of Theorem~\ref{th:GA_sub} requiring sub-Weibull tails.
For $p\ge n$, all conditions reduce to
\[
(\log p)^{\max\{9+d+2\nu,\;6+d+4\nu\}}
=o(n).
\]
Equivalently, the admitted dimension can be ultra-high-dimensional:
\begin{align*}
\log p=o(n^c),
\qquad\text{where}\qquad
c=
\begin{cases}
1/(9+d+2\nu)
& \text{if }\frac{3}{2}\ge \nu\ge 0,\\[6pt]
1/(6+d+4\nu)
& \text{if }\nu\ge \frac{3}{2}.
\end{cases}
\end{align*}

\bigskip\noindent
{\bf Example~\ref{ex:rgg}} (Regular Random Geometric Graph){\bf, Revisited.}
Consider again the regular random geometric graph in Example~\ref{ex:rgg}. 
Now, suppose the conditions of Theorem~\ref{th:GA_sub} requiring sub-Weibull tails.
For $p\ge n$, all conditions reduce to
\[
(\log p)^{\max\{11+2d+4\nu,\;8+2d+6\nu\}}
=o(n).
\]
Equivalently, the admitted dimension can be ultra-high-dimensional:
\begin{align*}
\log p=o(n^c),
\qquad\text{where}\qquad
c=
\begin{cases}
1/(11+2d+4\nu)
& \text{if }\dfrac{3}{2}\ge \nu\ge 0,\\[6pt]
1/(8+2d+6\nu)
& \text{if }\nu\ge \dfrac{3}{2}.
\end{cases}
\end{align*}

\section{Inference}\label{sec:estimation-cov}
To implement statistical inference based on the Gaussian approximation developed in Section \ref{sec:GA}, we require a consistent estimator of the network covariance matrix \(\Sigma_n\). Unlike the low-dimensional setting, however, consistency under conventional matrix norms is generally insufficient for high-dimensional inference. 
Instead, we need uniform consistency under the matrix infinity norm.

Several covariance estimation procedures have been proposed for network-dependent data, including the HAC estimator \citep{KMS2021}, cluster-robust methods \citep{Leung2023}, and bootstrap procedures \citep{kojevnikov2021bootstrap}. 
We focus on the network HAC estimation approach, and establish its convergence rate under the matrix infinity norm.
  
\subsection{Network HAC Estimator}
Recall the network covariance 
\[
\Sigma_n = \mathrm{Var}(T_X/\sqrt{n})=\frac1n\sum_{i,i'\in V_n}\mathrm{Cov}(X_{n,i},X_{n,i'})\in \mathbb{R}^{p\times p}.
\]
For simplicity of exposition, we first assume
\(\mathbb E(X_{n,i})=0\) for all \(i\in V_n\).
Under Assumption \ref{ass:decay rate}, covariance contributions from pairs of nodes decay with network distance. Hence, nearby pairs carry the dominant contribution to \(\Sigma_n\), while distant pairs contribute primarily to a small tail component. This motivates a network HAC estimator that keeps local covariance terms, while down-weighting and eventually truncating covariance terms at larger network distances.

For \(s\ge0\), define the distance-specific covariance matrix
\[
\Gamma_n(s)=
\frac{1}{n}
\sum_{i\in V_n}
\sum_{i'\in \mathcal{N}_n^\partial(i,s)}
\mathbb{E}[X_{n,i}X_{n,i'}^\top] ,
\]
By construction,\[
\Sigma_n=\sum_{s\ge0}
\Gamma_n(s).
\]
This representation is analogous to the autocovariance decomposition in time-series analysis, where the graph distance \(s\) plays the role of a lag.

To implement the down-weighting of distant covariance terms, we use a kernel function $w$ satisfying the following conditions.
\begin{assumption}[Kernel]{\label{ass:kernel}}
A kernel function \(w:\mathbb R^+\to[-1,1]\) satisfies \(w(0)=1\) and \(w(x)=0\) for \(x>1\). Moreover, there exist
constants \(C_w<\infty\) and \(\kappa>0\) such that
\[
|w(x)-1|\le C_w x^\kappa,
\qquad 0\le x\le1.
\]
\end{assumption}

Let $b_n$ denote the bandwidth and write 
\[
w_n(s)
=
w(s/b_n).
\]
With this notation, the network HAC estimator can be defined by
\begin{align}
\label{eq:hac}
\widehat\Sigma_n(b_n)
=&
\sum_{s\ge0}
w_n(s)\widehat\Gamma_n(s),
\qquad\text{where}
\\
\widehat\Gamma_n(s)
=&
\frac1n
\sum_{i\in V_n}
\sum_{i'\in\mathcal N_n^\partial(i,s)}
X_{n,i}X_{n,i'}^\top.
\notag
\end{align}
The bandwidth \(b_n\) controls the range of network distances incorporated into the estimator, while the kernel function $w$ controls how covariance contributions are attenuated as the distance increases.

The next theorem characterizes the deterministic approximation error induced by kernel weighting and bandwidth truncation.
\begin{theorem}{\label{lem:bias}} 
If Assumptions \ref{ass:decay rate}, \ref{ass:volume}, \ref{ass:kernel} hold, then
\[
\left|
\mathbb{E}\widehat\Sigma_n(b_n)-\Sigma_n
\right|_\infty
=
O\left(b_n^{-\kappa}+b_n^{2d+1}\rho^{b_n}\right).
\]
In particular, if \(b_n\to\infty\), then
\[
\left|
\mathbb{E}\widehat\Sigma_n(b_n)-\Sigma_n
\right|_\infty=o(1).
\]
\end{theorem}

To establish consistency of the HAC estimator, it remains to control its stochastic fluctuations around its expectation. The following theorem establishes nonasymptotic concentration inequalities for the HAC estimator, both entrywise and uniformly over all matrix elements, under finite-moment conditions.

\begin{theorem}\label{lem:hac-doob-variance_q}
Suppose that Assumptions \ref{ass:decay rate}, \ref{ass:volume}, and \ref{ass:kernel} hold.
For each $(a,b)\in [p]^2$, there exists a constant $c>0$ such that for all $x>0$,

\begin{equation}\label{eq:hac-entry-tail}
\mathbb P\!\left(
n\bigl|\,\widehat\Sigma_{n,ab}(b_n)-\mathbb E\widehat\Sigma_{n,ab}(b_n)\,\bigr|\ge x
\right)
\lesssim
\frac{
n^{q/4} (N_n(b_n))^{q/2} \Psi_{n,q}^q(0)
}{x^{q/2}}.
\end{equation}
Consequently,

\begin{equation}\label{eq:hac-max-tail}
\mathbb P\!\left(
n\bigl|\widehat\Sigma_n(b_n)-\mathbb E\widehat\Sigma_n(b_n)\bigr|_{\infty}
\ge x
\right)
\lesssim\,\frac{p^2\,n^{q/4}(N_n(b_n))^{q/2}\Psi_{n,q}^q(0)}{x^{q/2}}.
\end{equation}
\end{theorem}

Under the stronger sub-Weibull moment condition, the concentration inequality in Theorem \ref{lem:hac-doob-variance_q} can be strengthened to an exponential tail bound, as formally stated below.
For convenience, we define the dependence-adjusted sub-Weibull norm
\[
\Theta_{n,\psi_\nu,j}
:=
\sup_{q\ge2}
\frac{\Theta_{n,q,j}}{q^\nu}
\qquad\text{and}\qquad
\Phi_{n,\psi_\nu}
:=
\max_{1\le j\le p}
\Theta_{n,\psi_\nu,j}.
\]

\begin{theorem}\label{lem:hac-doob-variance_subwei}
Suppose that Assumptions \ref{ass:volume}, \ref{ass:subweibull}, and \ref{ass:kernel} hold.
For each $(a,b)\in [p]^2$, there exists a constant $c>0$ such that for all $x>0$,
\begin{equation}\label{eq:hac-entry-tail}
\mathbb P\!\left(
n\bigl|\,\widehat\Sigma_{n,ab}(b_n)-\mathbb E\widehat\Sigma_{n,ab}(b_n)\,\bigr|\ge x
\right)
\lesssim
 \exp\!\left(
- c \,
\frac{x^{\gamma}}
{ \big(\sqrt{n}\,N_n(b_n)\,\Phi_{n,\psi_\nu}^2\big)^{\gamma}}
\right).
\end{equation}
Consequently,
\begin{equation}\label{eq:hac-max-tail}
\mathbb P\!\left(
n\bigl|\widehat\Sigma_n(b_n)-\mathbb E\widehat\Sigma_n(b_n)\bigr|_{\infty}
\ge x
\right)
\lesssim\,
p^2 \exp\!\left(
- c \,
\frac{x^{\gamma}}
{ \big(\sqrt{n}\,N_n(b_n)\,\Phi_{n,\psi_\nu}^2\big)^{\gamma}}
\right),
\end{equation}
where $\gamma=1/(1+2\nu)$ for $\nu$ given in Assumption \ref{ass:subweibull}.
\end{theorem}

Combining the bias characterization with the concentration inequalities yields uniform consistency rates for the HAC estimator under the matrix infinity norm, as formally stated in the following corollary.

\begin{corollary}[Consistency Rate of the HAC Estimator]{\label{cor:sigmadis}}${}$\\
      (i) Under the conditions of Theorems \ref{lem:bias} and \ref{lem:hac-doob-variance_q}, $|\widehat\Sigma_n(b_n)-\Sigma_n|_\infty=O_p(R_n)$ holds, where
    \[R_n=p^{4/q}n^{-1/2}N_n(b_n)\Psi^2_{n,q}(0)+b_n^{-\kappa}+b_n^{2d+1}\rho^{b_n}.\]
    (ii) Under the conditions of Theorems \ref{lem:bias} and \ref{lem:hac-doob-variance_subwei}, 
    $|\widehat\Sigma_n(b_n)-\Sigma_n|_\infty=O_p(R^*_n)$ holds, where
    \[R^*_n=n^{-1/2}\,N_n(b_n)\,\Phi_{n,\psi_\nu}^2(\log{p})^{1/\gamma}+b_n^{-\kappa}+b_n^{2d+1}\rho^{b_n}.\]
\end{corollary}

Thus far, we have focused on the HAC estimator \(\widehat\Sigma_n(b_n)\), which assumes that $\mathbb E(X_{n,i})=0$ for all $i\in V_n$.
Now, consider the more general case where $\mathbb E (X_{n,i}) = \mu_0$ (not necessarily 0) for all $i \in V_n$, where $\mu_0$ is unknown to a researcher. 
Let the sample mean be denoted by $\bar X_n= T_X/n$.
The feasible HAC estimator in this general case is
\begin{align*}
\widetilde\Sigma_n(b_n)=&\sum_{s\ge 0}w_n(s)\widetilde\Gamma_n(s),
\qquad\text{where}
\\
\widetilde\Gamma_n(s)
=&
\frac{1}{n}
\sum_{i\in V_n}
\sum_{i'\in \mathcal{N}_n^\partial(i,s)}
(X_{n,i}-\bar X_n)(X_{n,i'}-\bar X_n)^\top.
\end{align*}

Under the conditions of Theorem \ref{lem:hac-doob-variance_q}, we have 
\[|\widetilde \Sigma_n(b_n)-\widehat \Sigma_n(b_n)|_\infty= O_p\left(
N_n(b_n)\Psi^2_{n,q}(0) p^{2/q}n^{-1}
\right).\]
Likewise, under the conditions of Theorem \ref{lem:hac-doob-variance_subwei}, we have  
\[
|\widetilde\Sigma_n(b_n)-\widehat\Sigma_n(b_n)|_\infty
=
O_p\left(
N_n(b_n)\Phi_{n,\psi_\nu}^2 (\log p)^{1/\gamma}n^{-1}
\right).
\]
Therefore, the residual errors $R_n$ or $R^*_n$ in Corollary \ref{cor:sigmadis} are of the same order between $|\widetilde\Sigma_n(b_n)-\Sigma_n|_\infty$ and $|\widehat\Sigma_n(b_n)-\Sigma_n|_\infty$.
We summarize the implications of these results formally as the following corollary.

\begin{corollary}[Consistency Rate of the Feasible HAC Estimator]\label{cor:sigmafea}${}$\\
     (i) Under the conditions of Theorems \ref{lem:bias} and \ref{lem:hac-doob-variance_q}, $|\widetilde\Sigma_n(b_n)-\Sigma_n|_\infty=O_p(R_n)$ holds, where
    \[R_n=p^{4/q}n^{-1/2}N_n(b_n)\Psi^2_{n,q}(0)+b_n^{-\kappa}+b_n^{2d+1}\rho^{b_n}.\]
    (ii) Under the conditions of Theorems \ref{lem:bias} and \ref{lem:hac-doob-variance_subwei}, 
    $|\widetilde\Sigma_n(b_n)-\Sigma_n|_\infty=O_p(R^*_n)$ holds, where
    \[R^*_n=n^{-1/2}\,N_n(b_n)\,\Phi_{n,\psi_\nu}^2(\log{p})^{1/\gamma}+b_n^{-\kappa}+b_n^{2d+1}\rho^{b_n}.\]
\end{corollary}

\begin{remark}[Positive Semidefiniteness]
The network HAC estimator $\widetilde\Sigma_n(b_n)$ is not necessarily positive semidefinite in finite samples for a general network-distance kernel. This feature is shared by network HAC estimators more generally, including that of \citet{KMS2021}. When positive semidefiniteness is required for implementation, such as for Gaussian simulation, one may replace $\widetilde\Sigma_n(b_n)$ by a positive-semidefinite regularization, for example by replacing its negative eigenvalues with zero. 
\end{remark}

\subsection{Simultaneous Inference}{\label{sec:inference}}
To conduct hypothesis testing and construct simultaneous confidence intervals, we need to approximate the critical value of the limiting Gaussian distribution in Theorems \ref{th:GA_q} and \ref{th:GA_sub}. Let \(\chi_\theta\) denote the \(\theta\)-quantile of
\[
|D_0^{-1}Z|_\infty,\qquad
Z\sim N(0,\Sigma_n),
\]
where \(0<\theta<1\).
If \(\Sigma_n\) were known, \(\chi_\theta\) could be computed directly by simulation.

In practice, \(\Sigma_n\) is unknown and must be replaced by the network HAC estimator
\(\widetilde\Sigma_n(b_n)\). Let $\widetilde D_0=[\mathrm{diag}(\widetilde\Sigma_n(b_n))]^{1/2}$. We estimate $\theta$-quantile $\chi_\theta$ by the conditional $\theta$-quantile $\widetilde{\chi}_\theta$ of
$
\left| \widetilde{D}_0^{-1} (\widetilde{\Sigma}_n(b_n))^{1/2} \eta \right|_\infty
$ given $(X_{n,i})_{i\in V_n}$, where $\eta \sim N(0,Id_p)$ is independent of $(X_{n,i})_{i\in V_n}$. Note that $\widetilde{\chi}_\theta$ can be computed by simulations where we use a Gaussian multiplier resampling method with the estimated network covariance matrices. 

Given a significance level \(\alpha\in(0,1)\), we reject the null hypothesis
\(
H_0:\mu=\mu_0
\)
whenever

\[
\sqrt n
\left|
\widetilde D_0^{-1}
(\bar X_n-\mu_0)
\right|_\infty
>
\widetilde\chi_{1-\alpha}.
\]
The corresponding simultaneous \((1-\alpha)\)-confidence intervals for
\(\mu=(\mu_1,\ldots,\mu_p)^\top\) are

\[
\bar X_{n,j}
\pm
\frac{\widetilde\chi_{1-\alpha}}
{\sqrt n}
\widetilde\sigma_{jj}^{1/2},
\qquad
j=1,\ldots,p,
\]
where
\(
\widetilde\sigma_{jj}
=
(\widetilde\Sigma_n(b_n))_{jj}.
\)

The following corollary establishes the asymptotic validity of this inference procedure.
\begin{corollary}[Validity of Inference]{\label{cor:test}}
    (i) Let conditions in Theorem \ref{th:GA_q} and  Corollary \ref{cor:sigmafea}(i) hold. In addition, assume $R_n\log^2 p \to 0$ with $R_n$ in Corollary \ref{cor:sigmafea}(i). Then,
    \begin{equation}\label{eq:test}
\sup_{\theta \in (0,1)}
\left|
\mathbb{P}\!\left(
\sqrt{n}\,\big|\widetilde D_0^{-1}(\bar X_n-\mu_0)\big|_\infty 
\ge \widetilde\chi_{1-\theta}
\right)
- \theta
\right|
\;\to\; 0.
\end{equation}
(ii) Let conditions in Theorem \ref{th:GA_sub} and Corollary \ref{cor:sigmafea}(ii) hold. In addition, assume $R^*_n\log^2 p \to 0$ with $R^*_n$ in Corollary \ref{cor:sigmafea}(ii). Then, 
we have validity of the test as \eqref{eq:test}.
\end{corollary}

\section{Simulation Studies I: Gaussian Approximation}
\label{sec:simGA}
In this section, we examine the finite-sample performance of the high-dimensional Gaussian approximation property using Monte Carlo simulations.

\subsection{Design Overview}
\label{subsec:sim_overview}
In each Monte Carlo repetition, we first generate a network $\mathcal{G}_n$, and then simulate node-level observations $\{X_{n,i}\}_{i\in V_n}$ from a network-dependent data generating process (DGP). 
We compare the distribution of the standardized maximum statistic with its Gaussian benchmark. Performance of the approximation is assessed by using QQ plots.

\subsection{Data Generating Processes}
\label{subsec:sim_dgp}
We take the ring network and the regular random geometric graph introduced in Examples~\ref{ex:ring} and \ref{ex:rgg}, respectively. For the random geometric graph, we use the standard random geometric graph where node locations are generated independently from the uniform distribution on $[0,1]^2$, and choose the connection radius 
\(
r_n
=
\left(\frac{\bar k}{\pi n}\right)^{1/2},
\)
where $\bar k$ is the target expected degree. 
This construction preserves the scaling $r_n\asymp n^{-1/2}$ used in Example~\ref{ex:rgg}.
Conditional on each realized finite network, the network regularity conditions are satisfied for suitable realization-specific constants. In particular, neighborhood sizes satisfy a polynomial growth bound in graph distance, and the network admits a partition into approximately balanced clusters with small boundaries.

Conditional on the network realization, we generate $p$-dimensional  observations $\{X_{n,i}\}$ according to
\begin{equation}
X_{n,i}
=
\sum_{s\ge 0}
\frac{\gamma^s}{N_n^{\partial}(i,s)} M_{n,s}
\left(
\sum_{i'\in N_n^{\partial}(i,s)}
\varepsilon_{n,i'}
\right),
\qquad i\in V_n,
\label{eq:sim_dgp_kms}
\end{equation}
where $N_n^{\partial}(i,s)$ is the neighborhood size defined in \eqref{eq:neighbor size}. The shocks  $\{\varepsilon_{n,i}\}$ are independently generated from a standardized
Student-$t$ distribution, i.e.,
\(
\varepsilon_{n,i}\sim T_t/\sqrt{t/(t-2)},
\) where $T_t$ denotes a Student-$t$ random variable with $t$ degrees of freedom.
The matrices $\{M_{n,s}\}_{s\ge0}$ introduce cross-coordinate dependence in the
$p$-dimensional observations. They are generated once and then kept fixed across
Monte Carlo repetitions, with entries independently drawn from $N(0,1/p)$.   Parameter $\gamma\in\{0.6,0.8\}$
controls the strength of network dependence.

This DGP is a linear network functional of independent node-level shocks. The
weight attached to shocks at graph distance $s$ decays exponentially through
$\gamma^s$, while the normalization by $N_n^{\partial}(i,s)$ prevents distant
shells with many nodes from mechanically dominating the variance. Hence the
design directly captures the interaction between dependence decay and network
distance growth.  

Recall the notation $T_X = \sum_{i\in V_n} X_{n,i}$, and the covariance matrix of
$T_X/\sqrt n$ is given by
\begin{equation}
\Sigma_n
=
\mathrm {Var}(T_X/\sqrt n)
=
\frac{1}{n}\sum_{i,i'\in V_n}
\mathrm {Cov}(X_{n,i},X_{n,i'}).
\label{eq:Sigma-def}
\end{equation}
Equivalently, exchanging the order of summations yields the closed-form
representation
\begin{equation}
\Sigma_n
=
\frac{1}{n}
\sum_{i\in V_n}
v_{i} v_{i}^\top,
\qquad\text{where}\qquad
v_i
:=
\sum_{i'\in V_n}
\gamma^{d(i,i')} M_{n,d(i,i')}\mathbf 1\{d_n(i,i')<\infty\}.
\label{eq:Sigma-closed}
\end{equation}

The statistic of interest is the standardized maximum statistic
\begin{equation}
T_n
=
\sqrt{n}\left|
D_0^{-1} \sum_{i\in V_n} X_{n,i}
\right|_\infty,
\label{eq:sim_Tn}
\end{equation}
where $D_0=\left(\mathrm{diag}(\Sigma_n)\right)^{1/2}$. Let $Z\sim N(0,\Sigma_n)$ denote the Gaussian
benchmark, and let $T_Z=|D_0^{-1}Z|_\infty$ be its counterpart of the statistic $T_n$.

\subsection{Simulation results}
\label{subsec:sim_takeaways}
For each $(n,p,\gamma,t)$ configuration, we plot the empirical distributions of $T_n$ against the empirical distributions of $T_Z$ in 2000 Monte Carlo realizations. Simulation results are illustrated in Figures ~\ref{fig:ring_QQ_nu4_nu8} and ~\ref{fig:RGG_QQ_nu4_nu8} for the ring network and the random geometric graph, respectively. Figure \ref{fig:rgg_1x3} illustrates the random geometric graph networks.

\begin{figure}[p]
\centering

\begin{minipage}[t]{0.48\textwidth}
\centering

\begin{subfigure}[t]{0.48\linewidth}
  \includegraphics[width=\linewidth]{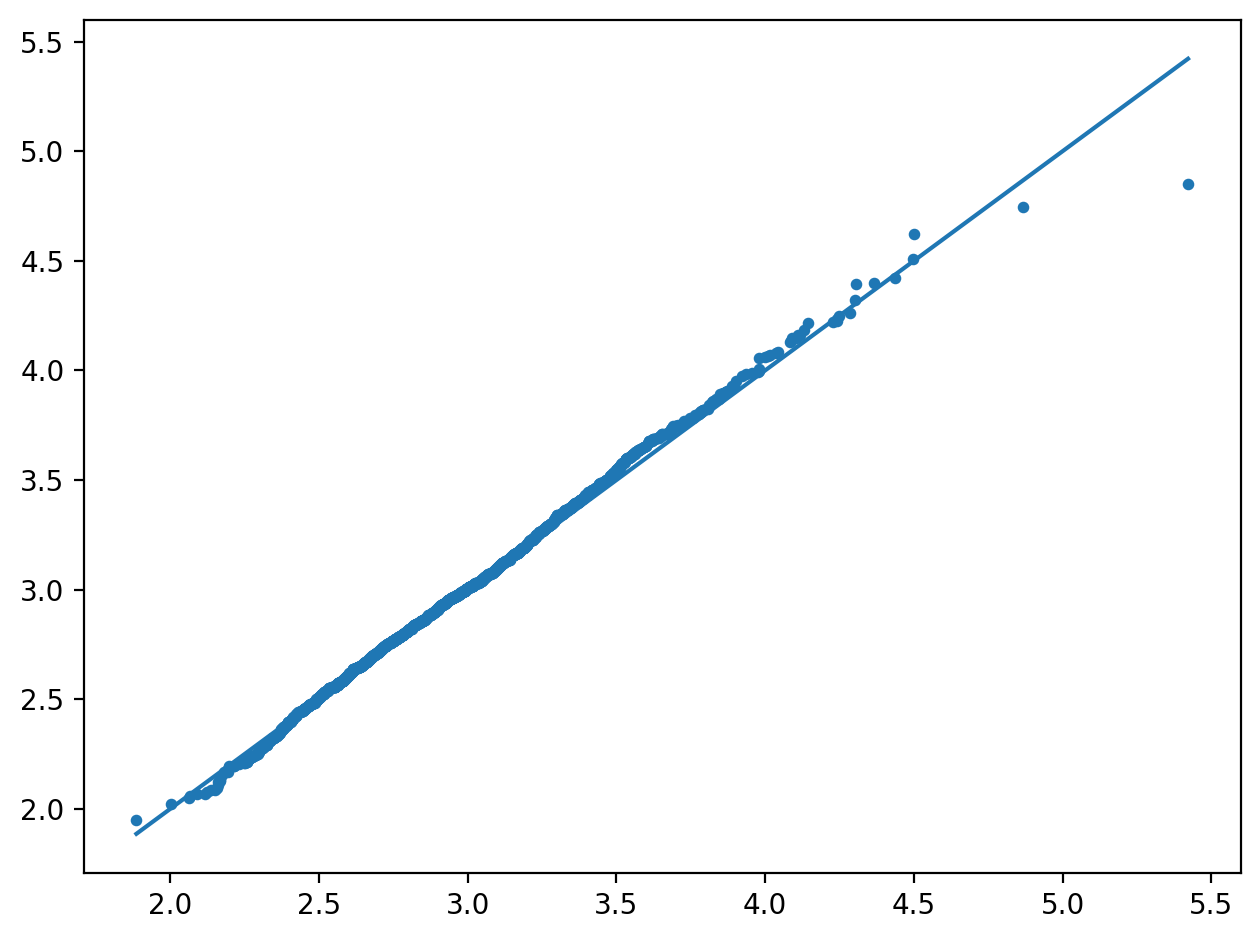}
\end{subfigure}\hfill
\begin{subfigure}[t]{0.48\linewidth}
  \includegraphics[width=\linewidth]{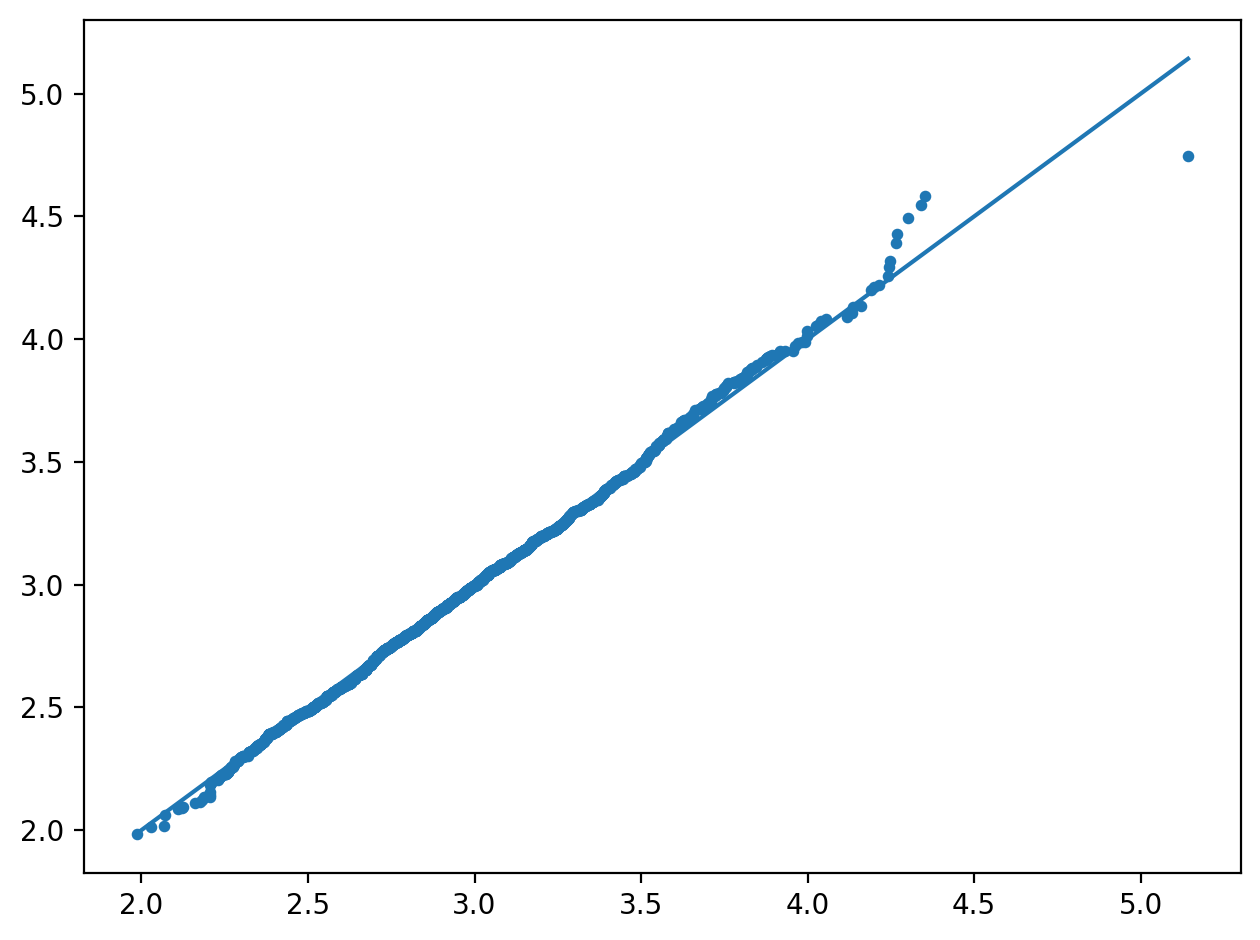}
\end{subfigure}

\vspace{2mm}

\begin{subfigure}[t]{0.48\linewidth}
  \includegraphics[width=\linewidth]{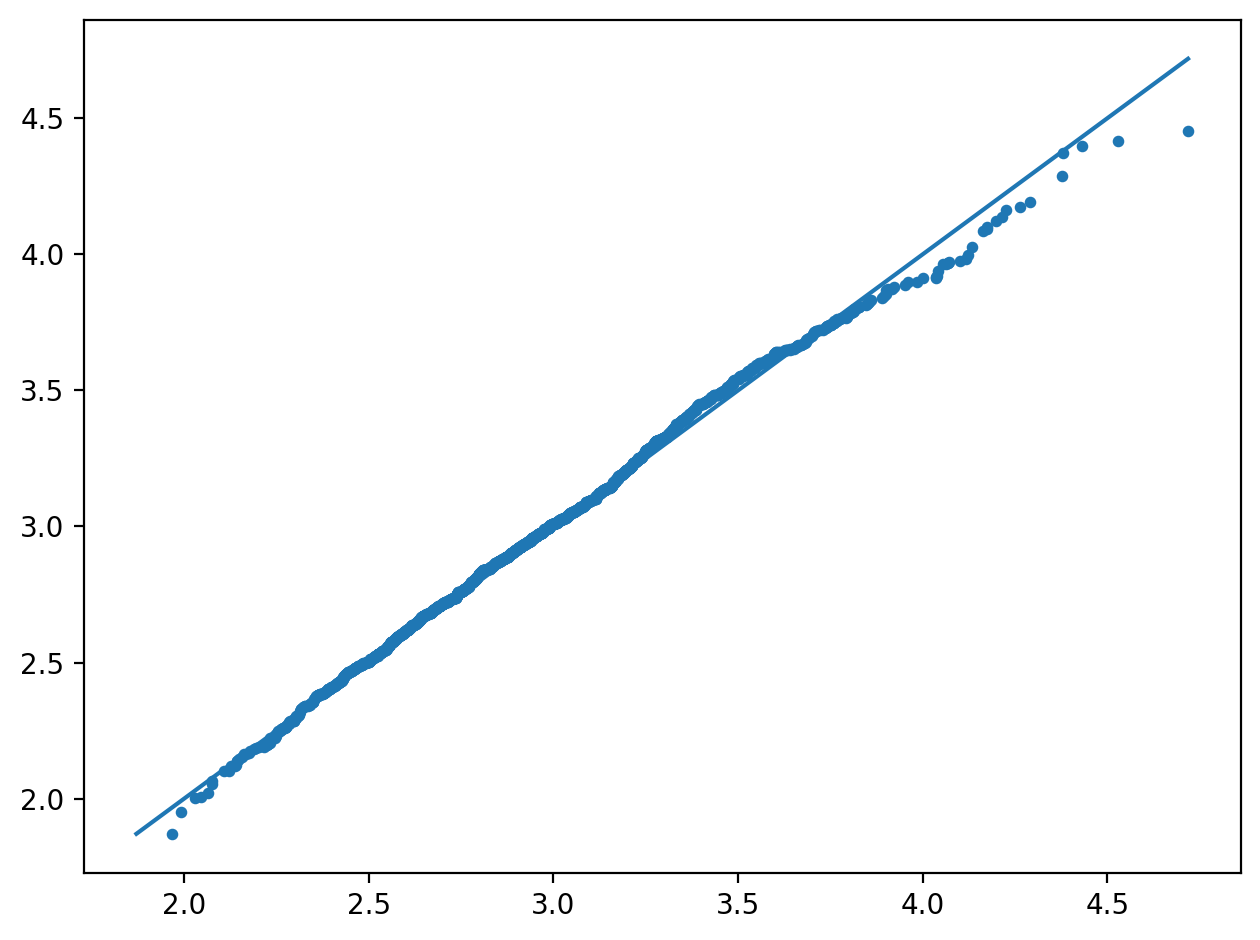}
\end{subfigure}\hfill
\begin{subfigure}[t]{0.48\linewidth}
  \includegraphics[width=\linewidth]{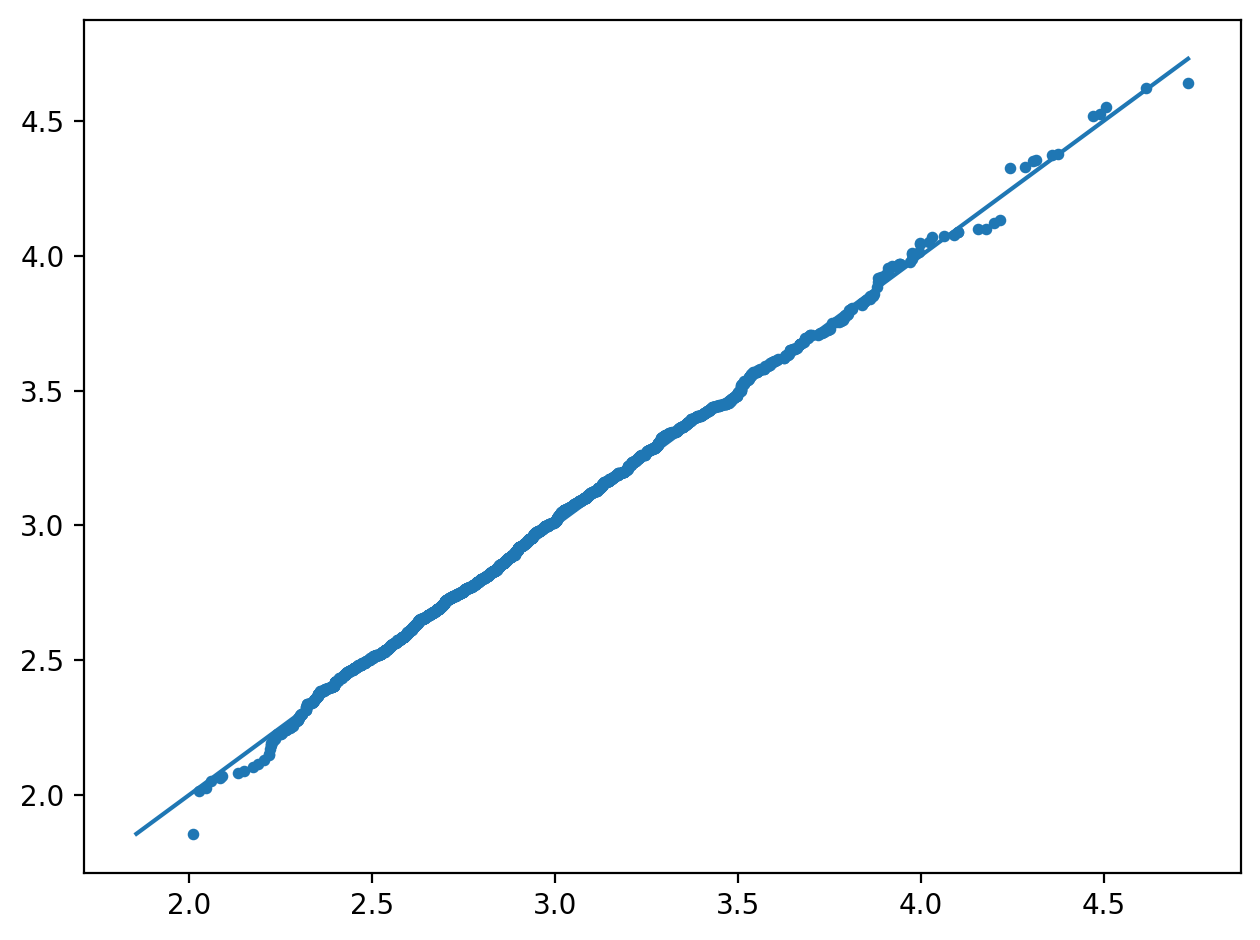}
\end{subfigure}

\vspace{2mm}

\begin{subfigure}[t]{0.48\linewidth}
  \includegraphics[width=\linewidth]{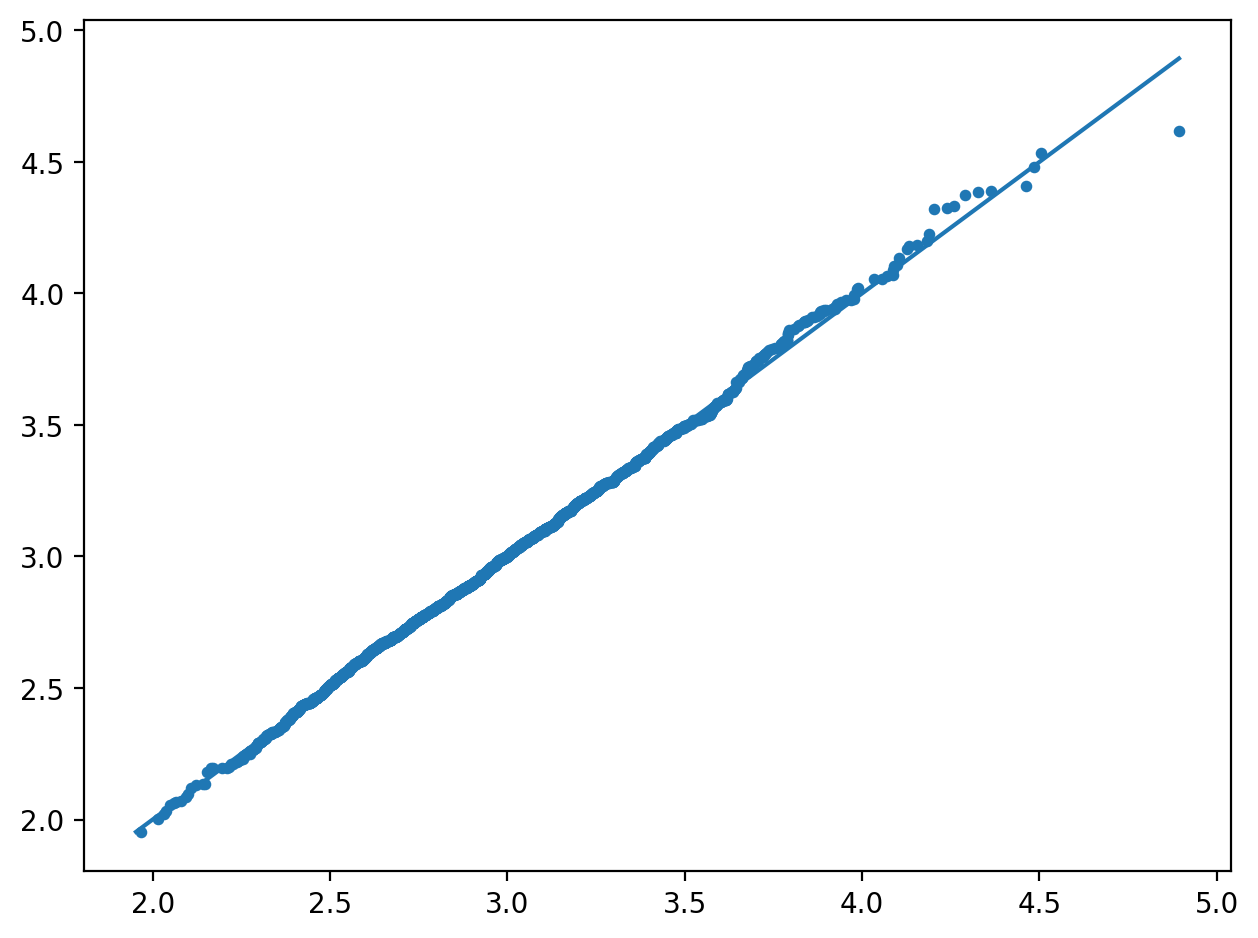}
\end{subfigure}\hfill
\begin{subfigure}[t]{0.48\linewidth}
  \includegraphics[width=\linewidth]{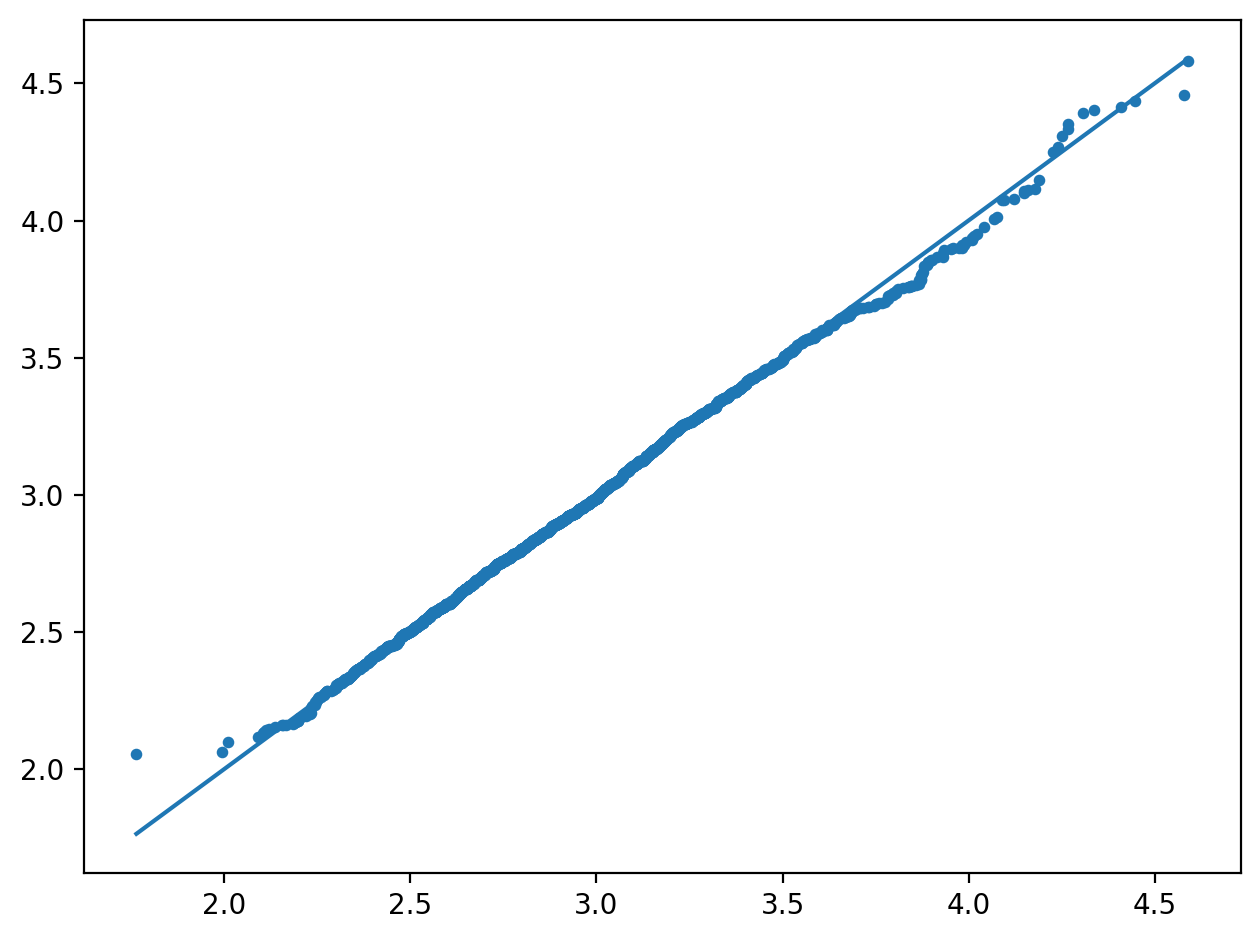}
\end{subfigure}

\caption*{\footnotesize $t = 8$}
\end{minipage}
\hfill
\begin{minipage}[t]{0.48\textwidth}
\centering

\begin{subfigure}[t]{0.48\linewidth}
  \includegraphics[width=\linewidth]{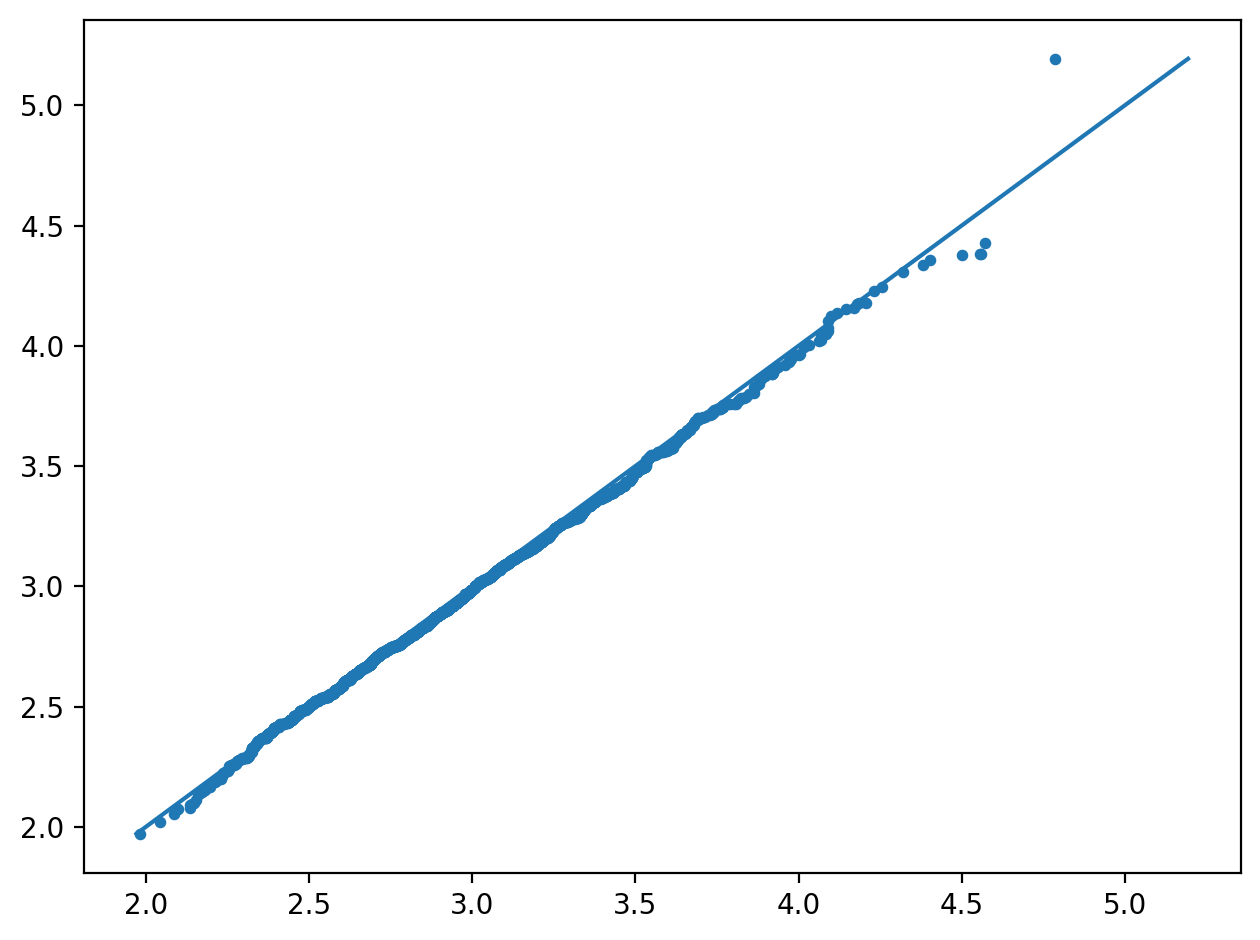}
\end{subfigure}\hfill
\begin{subfigure}[t]{0.48\linewidth}
  \includegraphics[width=\linewidth]{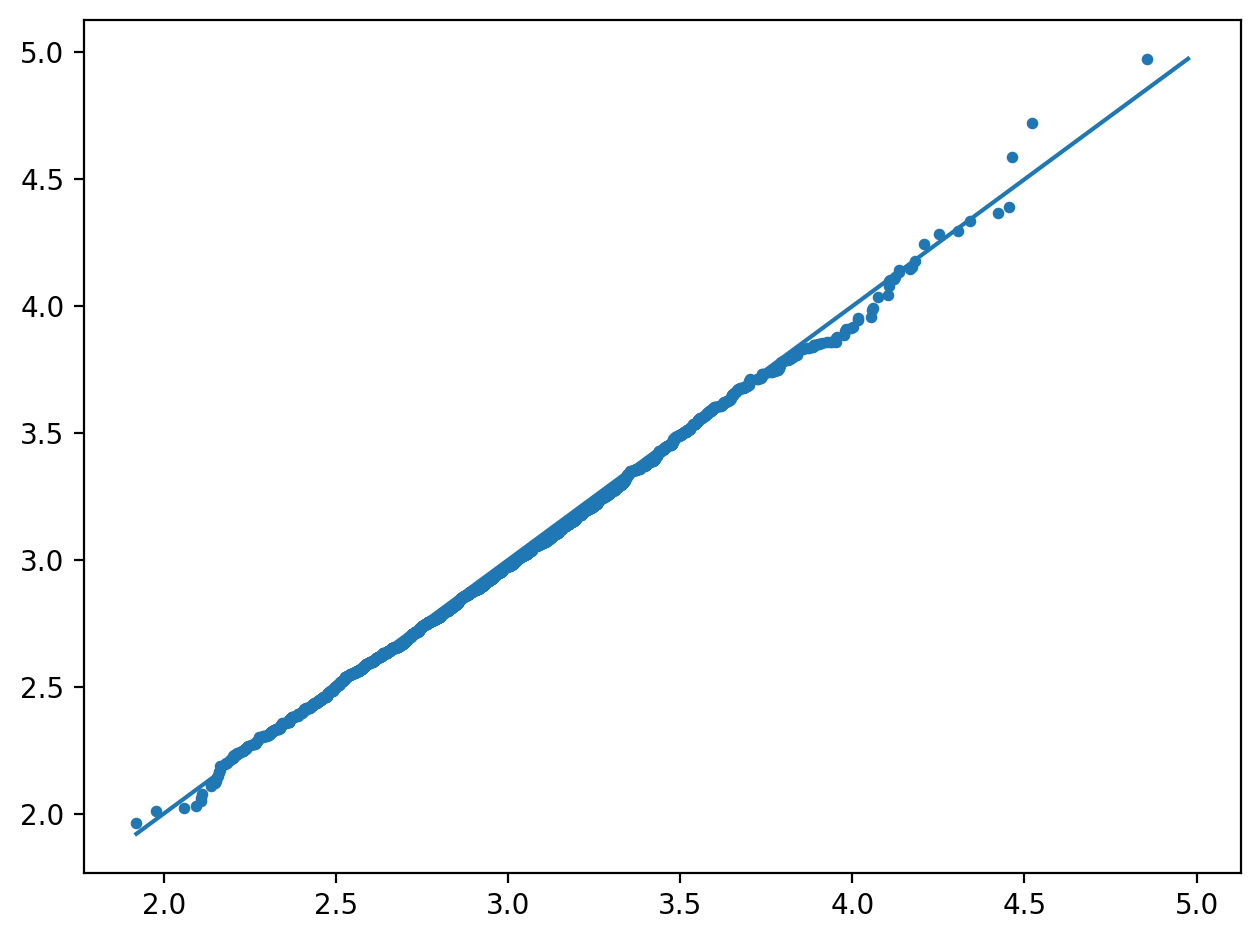}
\end{subfigure}

\vspace{2mm}

\begin{subfigure}[t]{0.48\linewidth}
  \includegraphics[width=\linewidth]{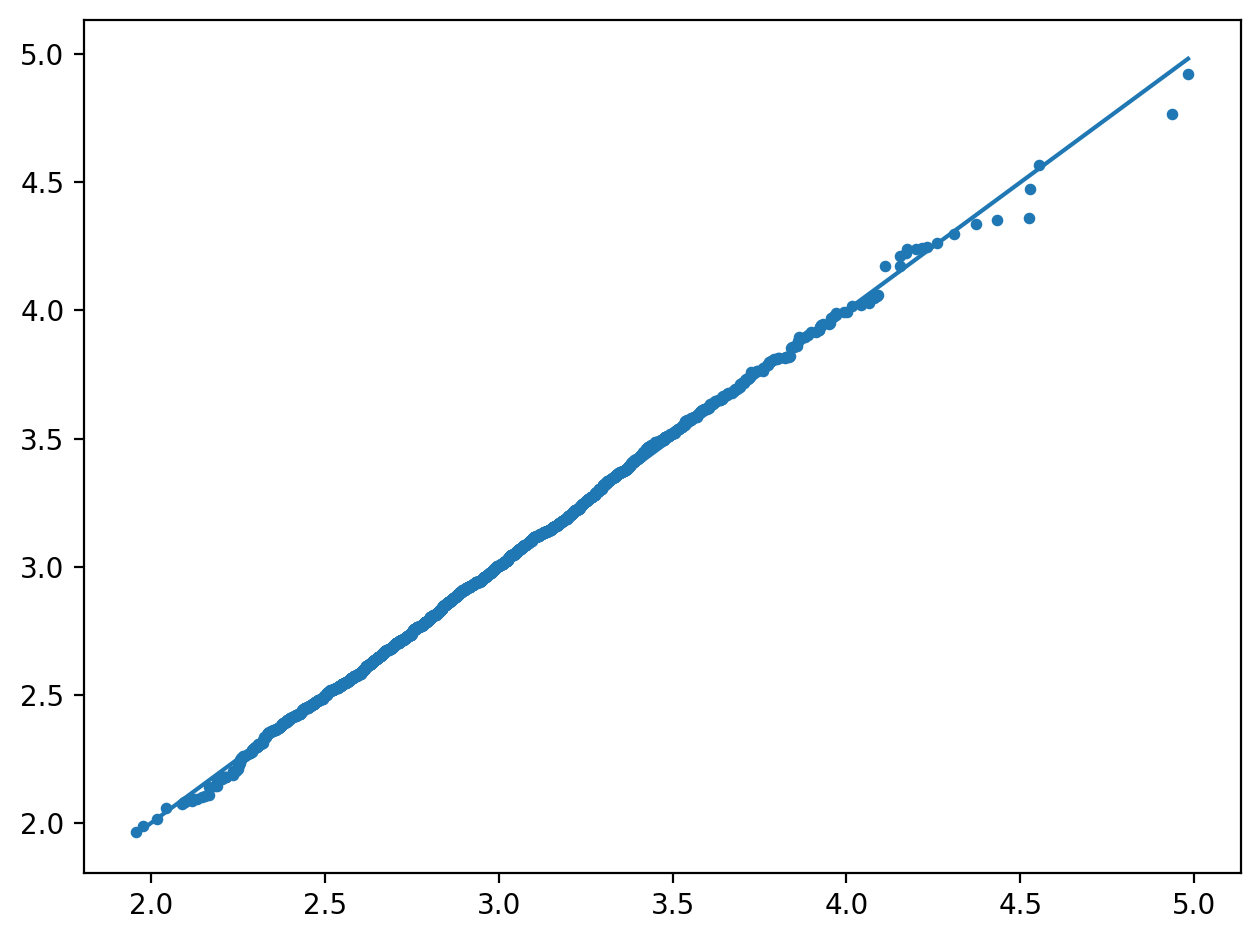}
\end{subfigure}\hfill
\begin{subfigure}[t]{0.48\linewidth}
  \includegraphics[width=\linewidth]{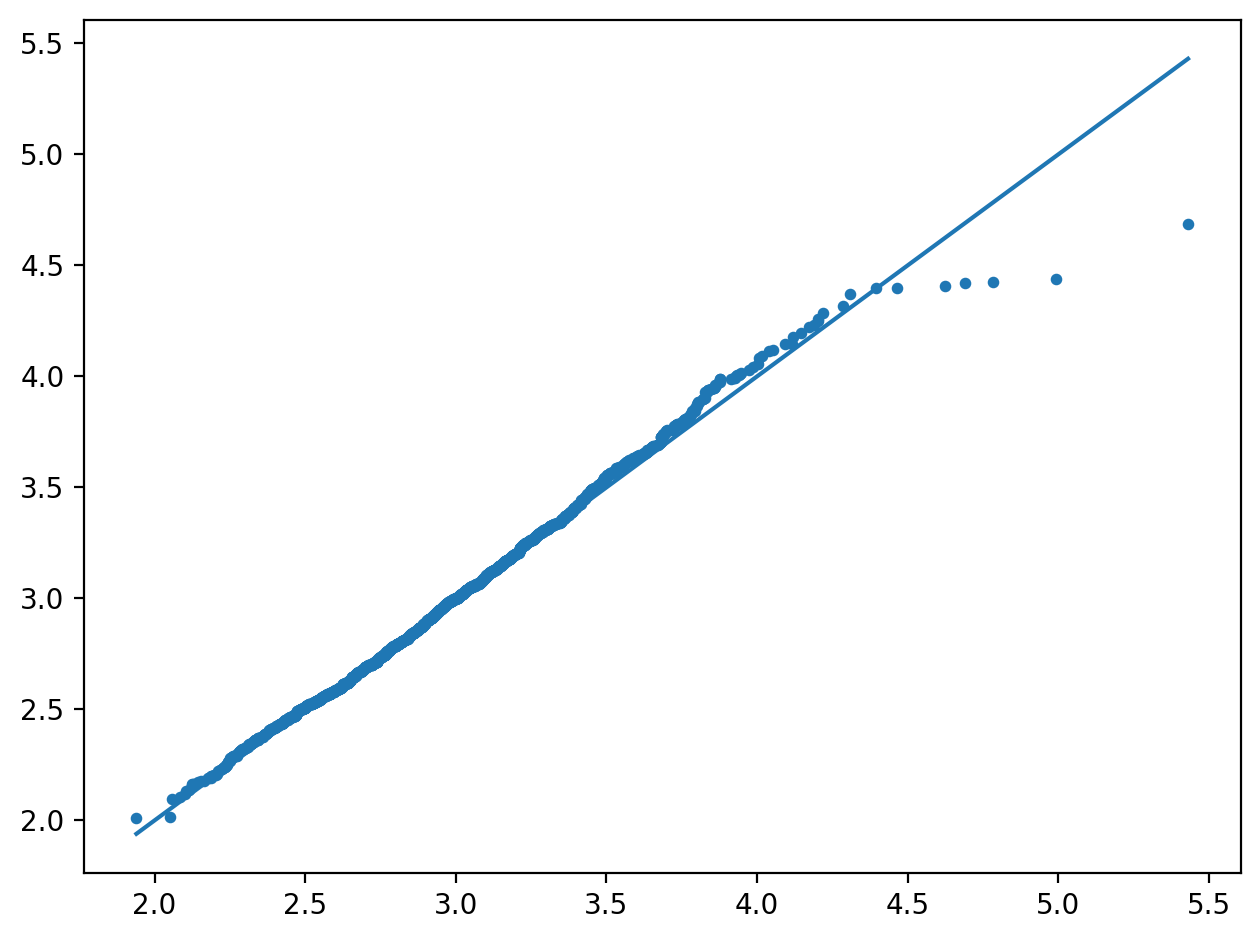}
\end{subfigure}

\vspace{2mm}

\begin{subfigure}[t]{0.48\linewidth}
  \includegraphics[width=\linewidth]{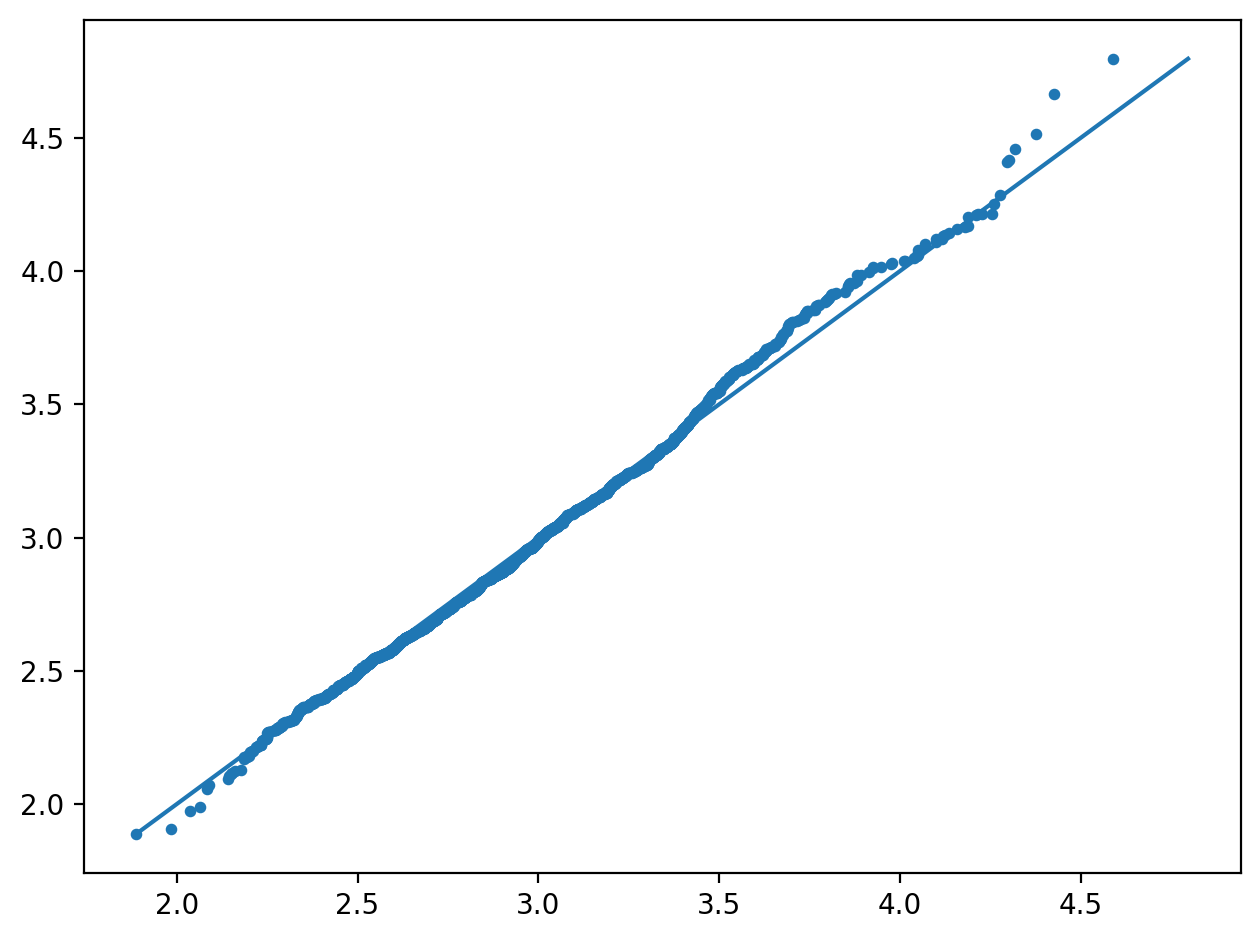}
\end{subfigure}\hfill
\begin{subfigure}[t]{0.48\linewidth}
  \includegraphics[width=\linewidth]{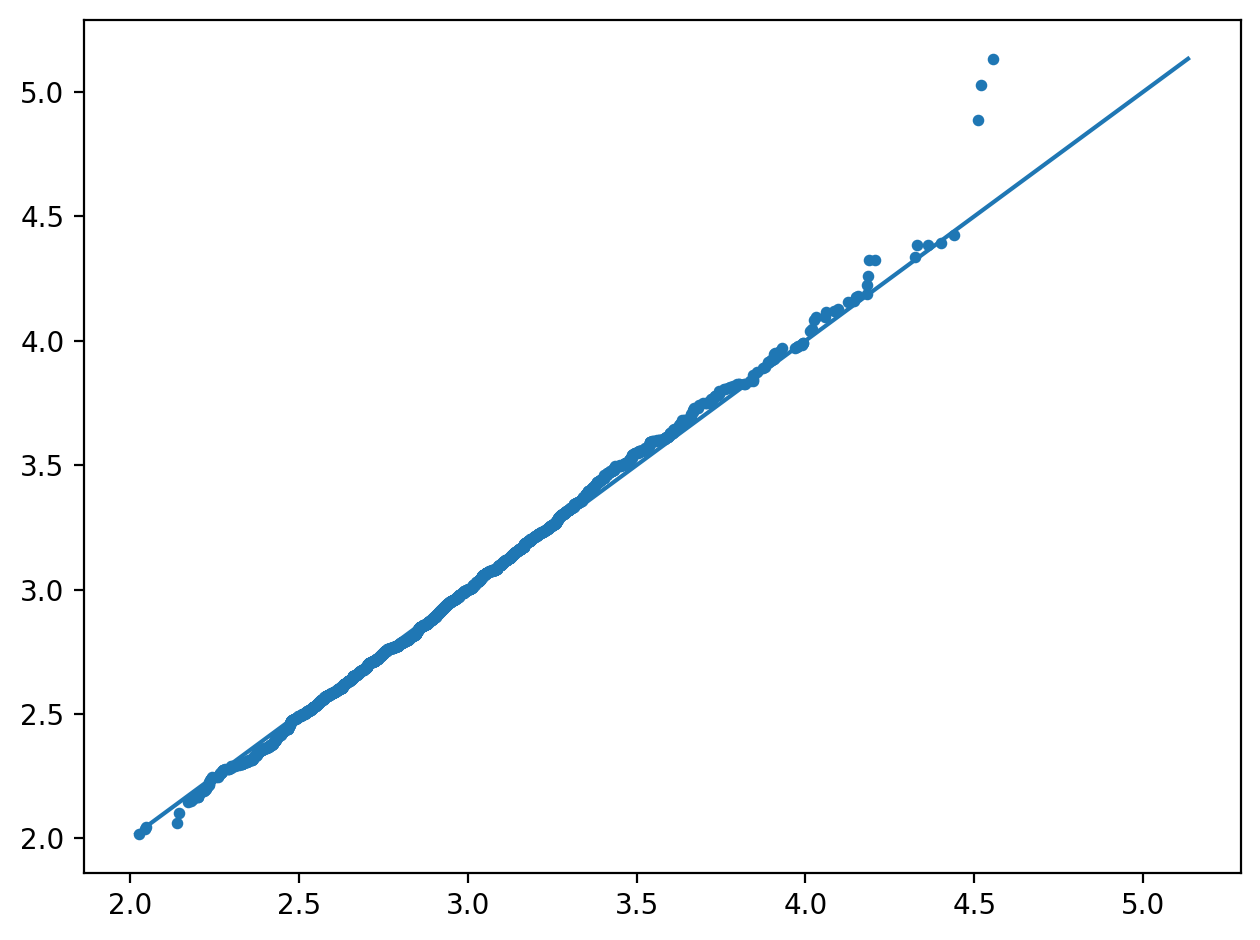}
\end{subfigure}

\caption*{\footnotesize $t = 12$}
\end{minipage}
\caption{
QQ plots for the ring network.
Left panels correspond to $t=8$ and right panels correspond to $t=12$.
Rows correspond to $n=\{100,200,400\}$ and columns correspond to
$\rho=\{0.6,0.8\}$ with $p=200$ and $R=2000$.
The horizontal axis reports the empirical distributions of $T_n$,
and the vertical axis reports the corresponding Gaussian distributions of $T_z$.
}
\label{fig:ring_QQ_nu4_nu8}
\end{figure}

\begin{figure}[p]
\centering

\begin{minipage}[t]{0.48\textwidth}
\centering

\begin{subfigure}[t]{0.48\linewidth}
  \includegraphics[width=\linewidth]{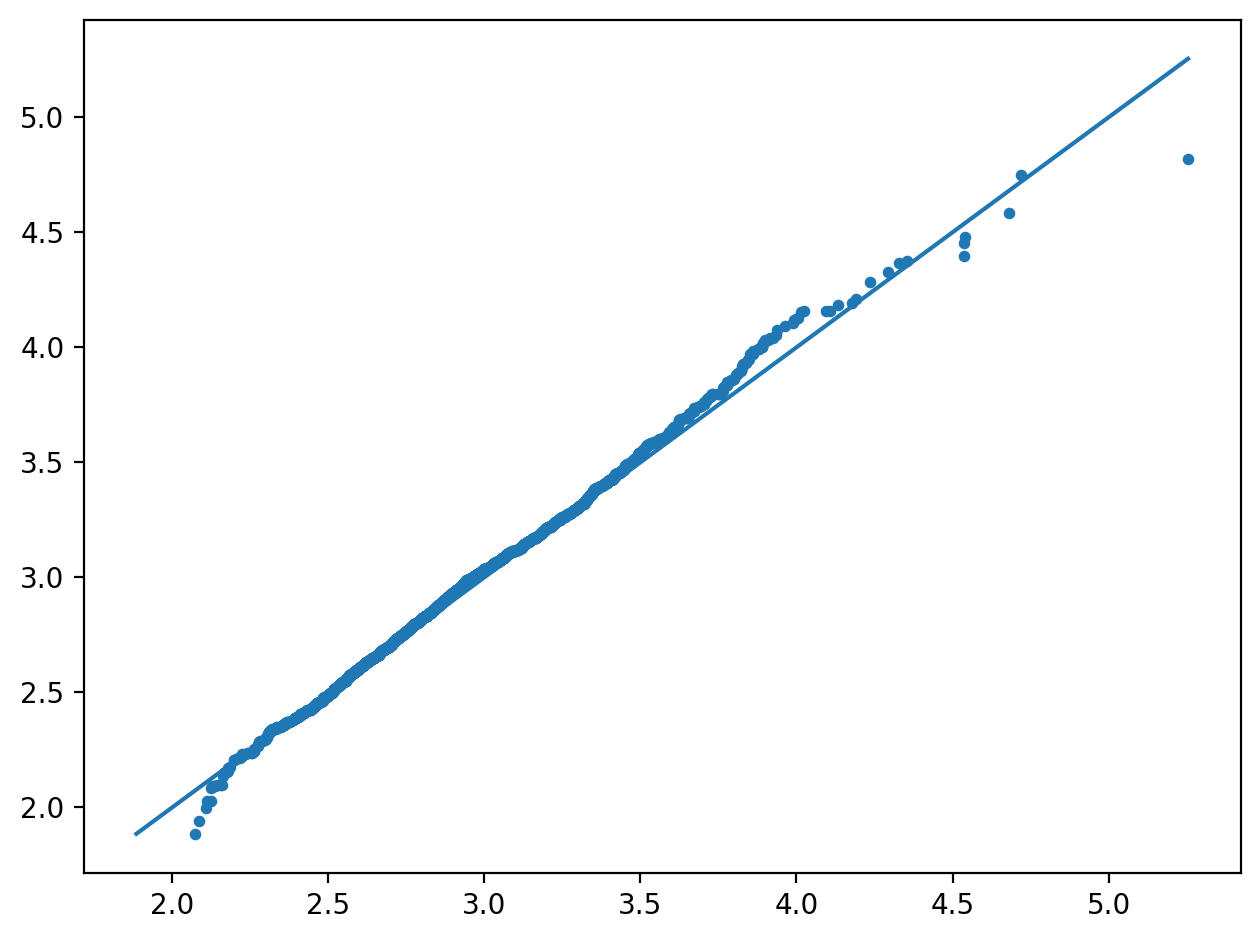}
\end{subfigure}\hfill
\begin{subfigure}[t]{0.48\linewidth}
  \includegraphics[width=\linewidth]{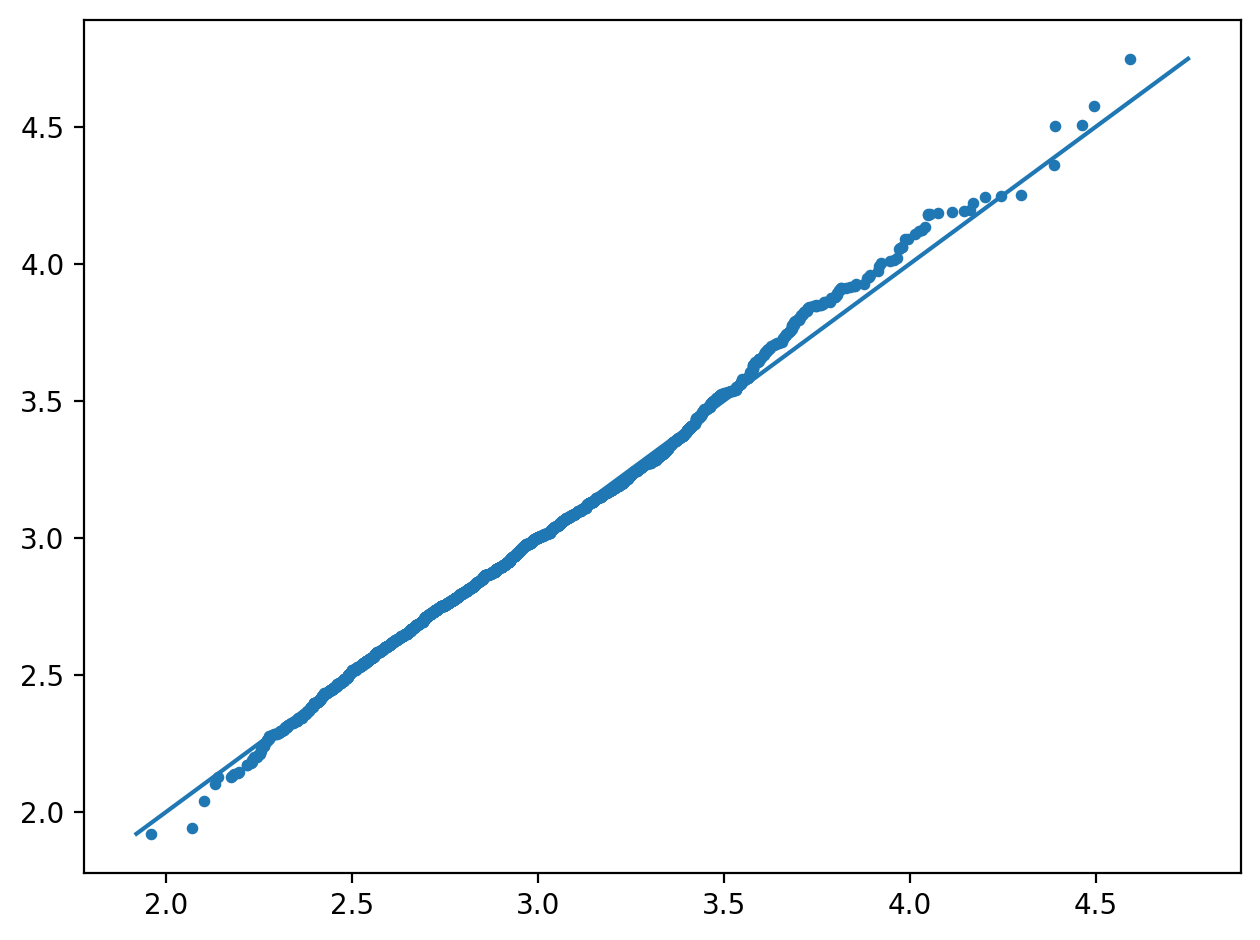}
\end{subfigure}

\vspace{2mm}

\begin{subfigure}[t]{0.48\linewidth}
  \includegraphics[width=\linewidth]{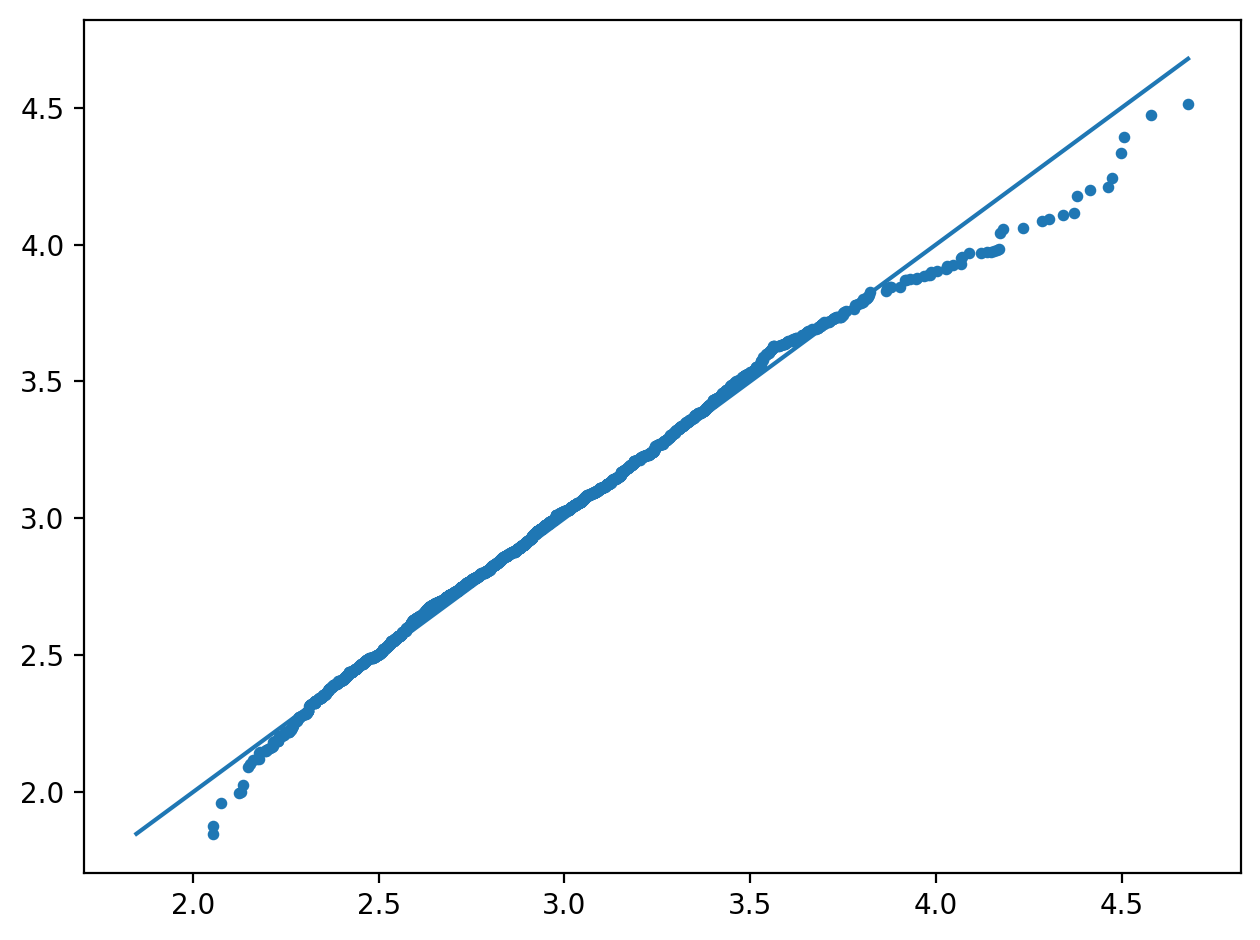}
\end{subfigure}\hfill
\begin{subfigure}[t]{0.48\linewidth}
  \includegraphics[width=\linewidth]{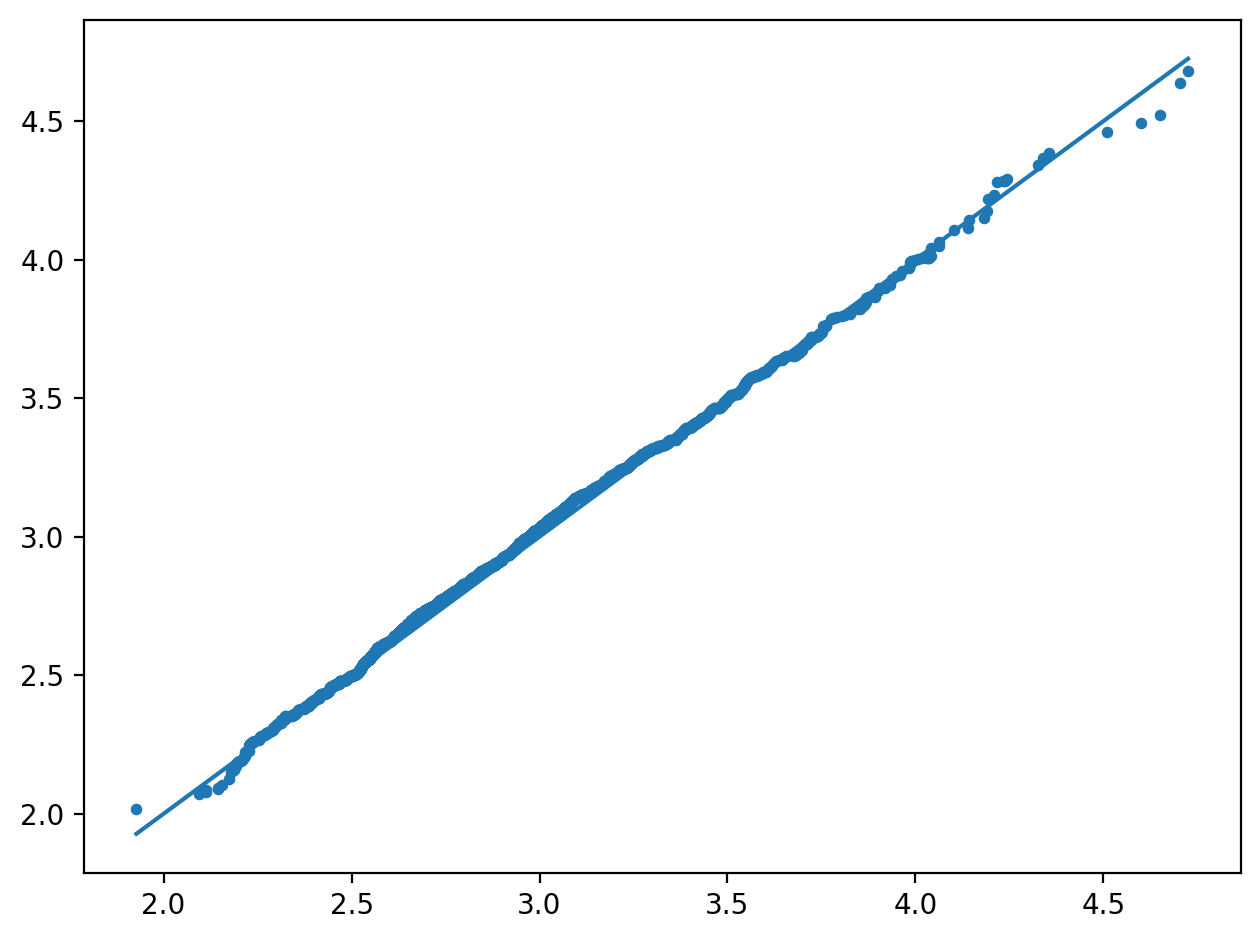}
\end{subfigure}

\vspace{2mm}

\begin{subfigure}[t]{0.48\linewidth}
  \includegraphics[width=\linewidth]{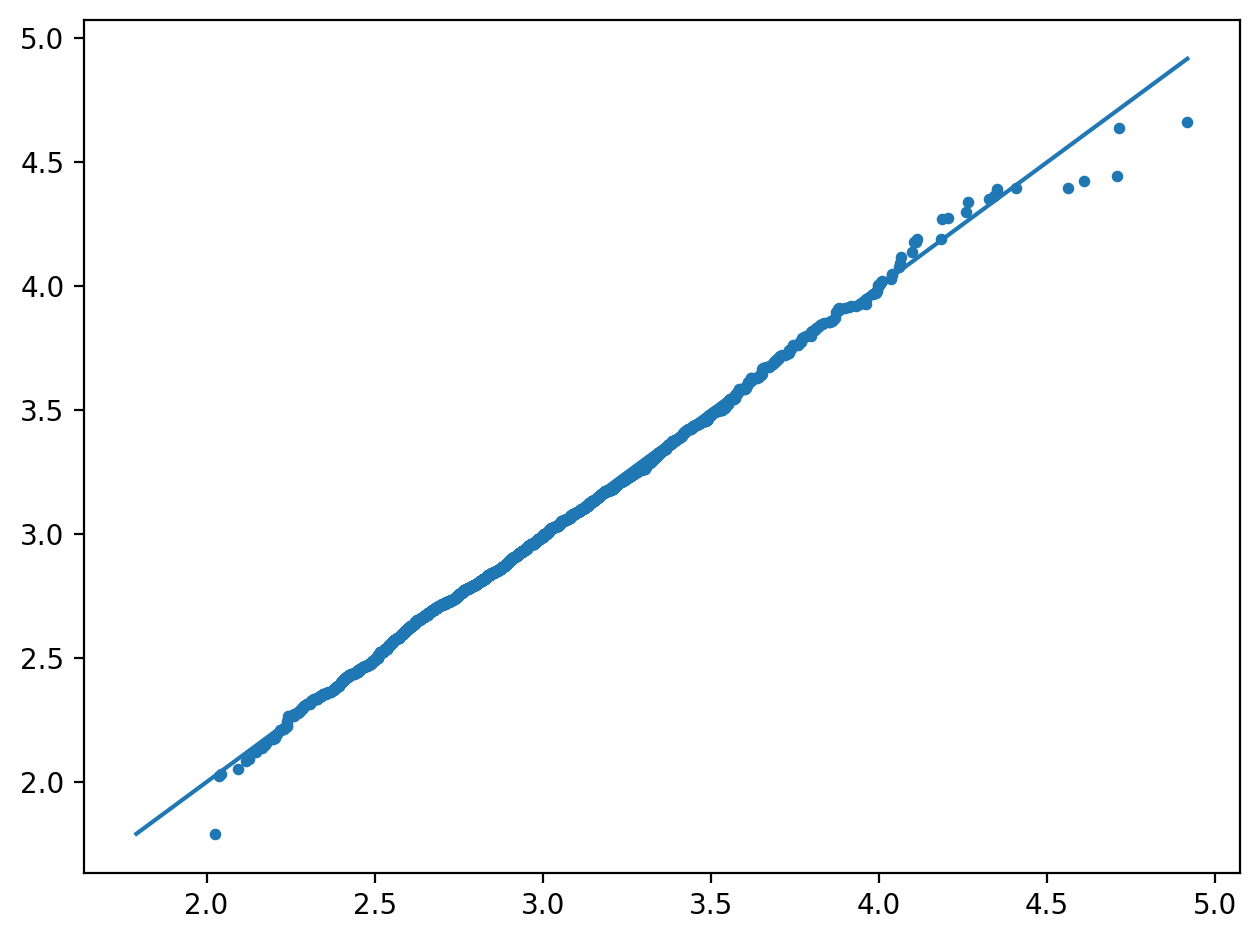}
\end{subfigure}\hfill
\begin{subfigure}[t]{0.48\linewidth}
  \includegraphics[width=\linewidth]{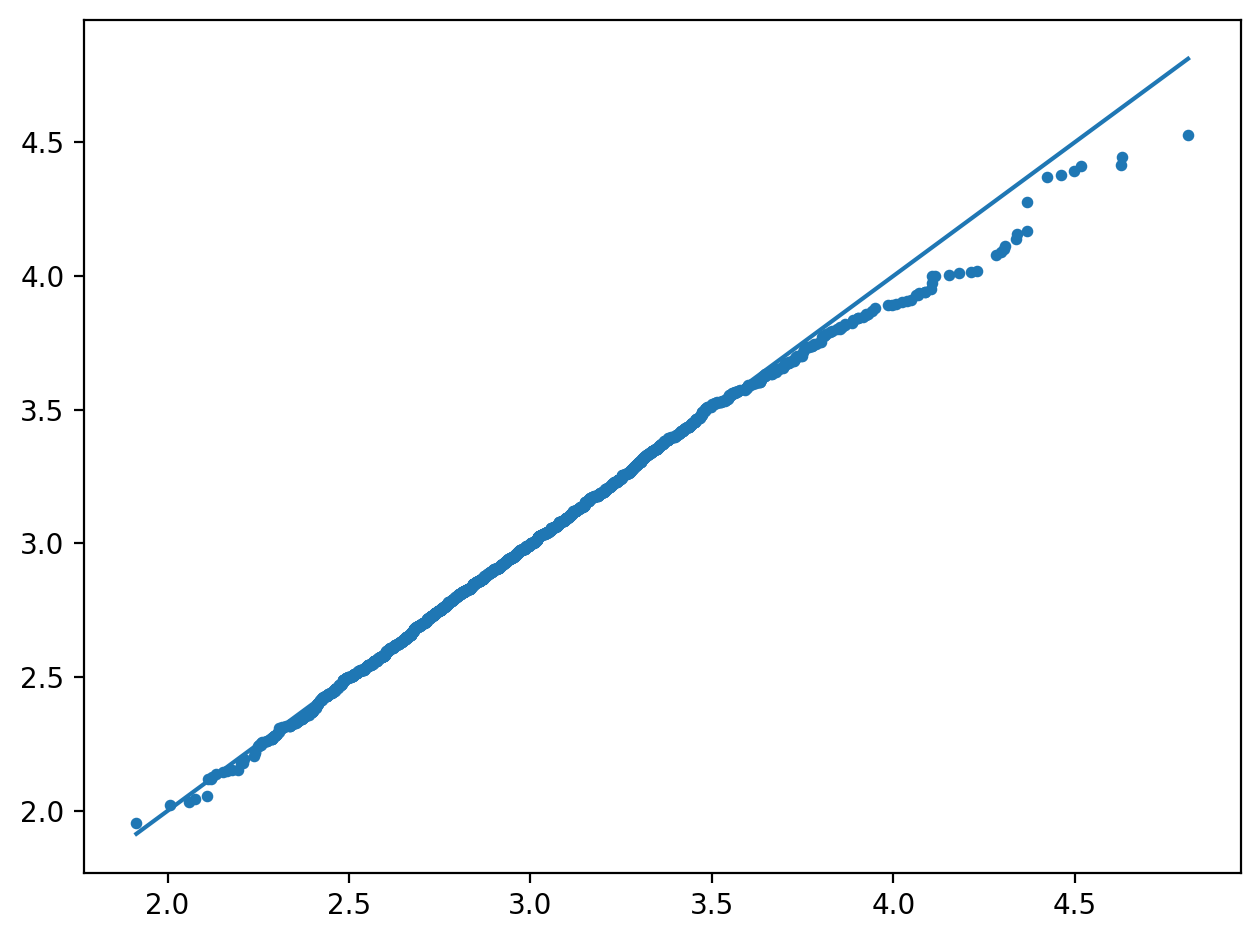}
\end{subfigure}

\caption*{\footnotesize $t = 8$}
\end{minipage}
\hfill
\begin{minipage}[t]{0.48\textwidth}
\centering

\begin{subfigure}[t]{0.48\linewidth}
  \includegraphics[width=\linewidth]{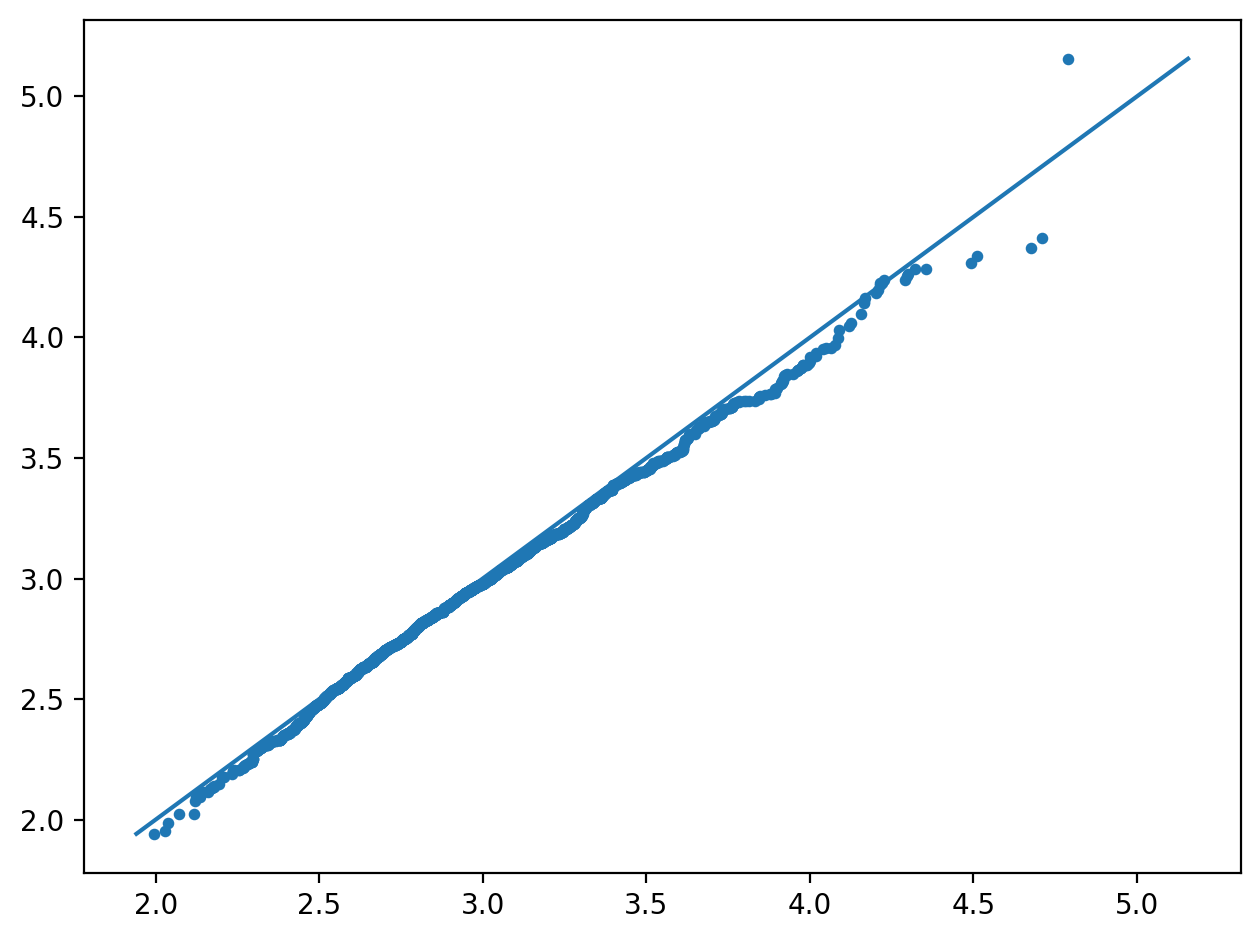}
\end{subfigure}\hfill
\begin{subfigure}[t]{0.48\linewidth}
  \includegraphics[width=\linewidth]{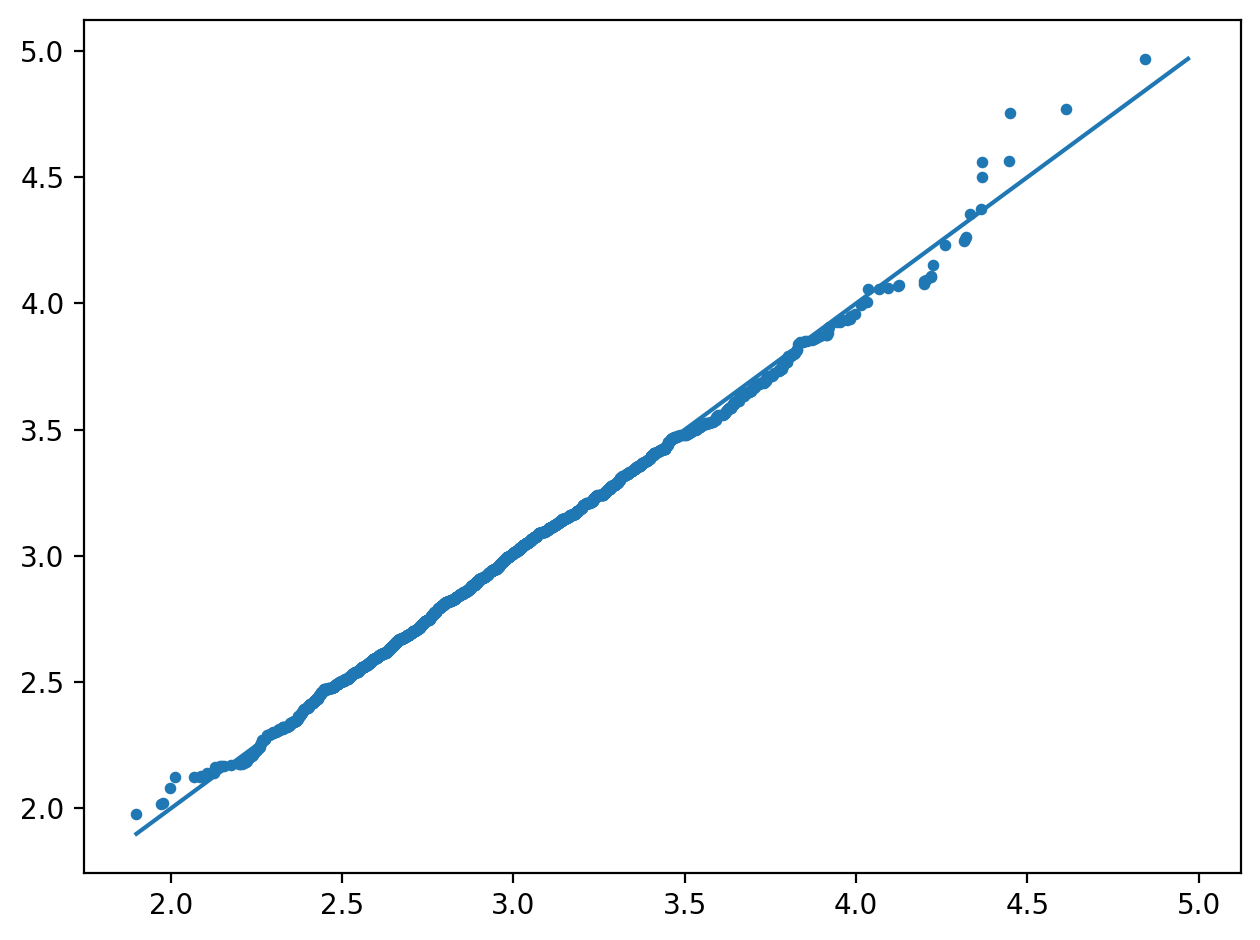}
\end{subfigure}

\vspace{2mm}

\begin{subfigure}[t]{0.48\linewidth}
  \includegraphics[width=\linewidth]{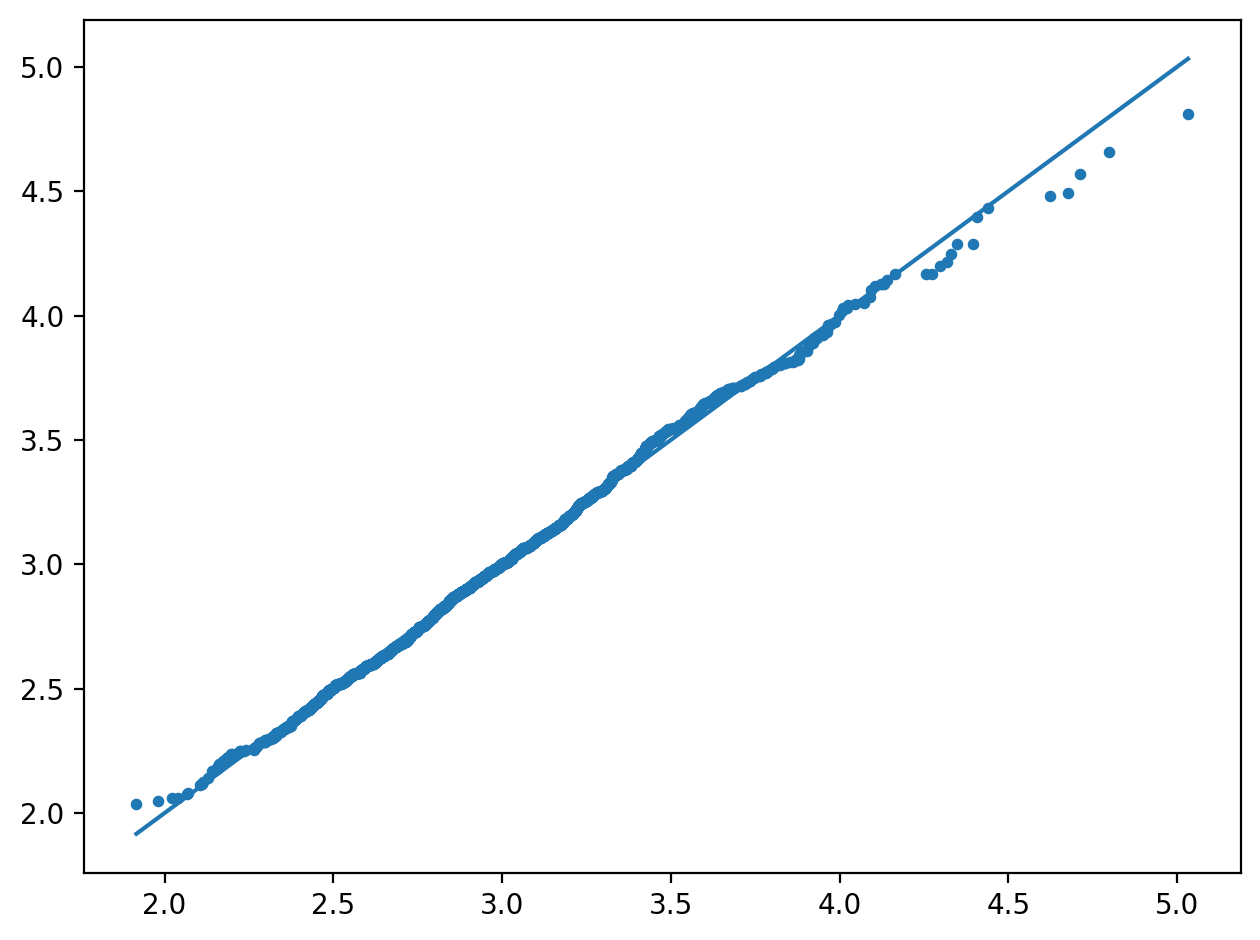}
\end{subfigure}\hfill
\begin{subfigure}[t]{0.48\linewidth}
  \includegraphics[width=\linewidth]{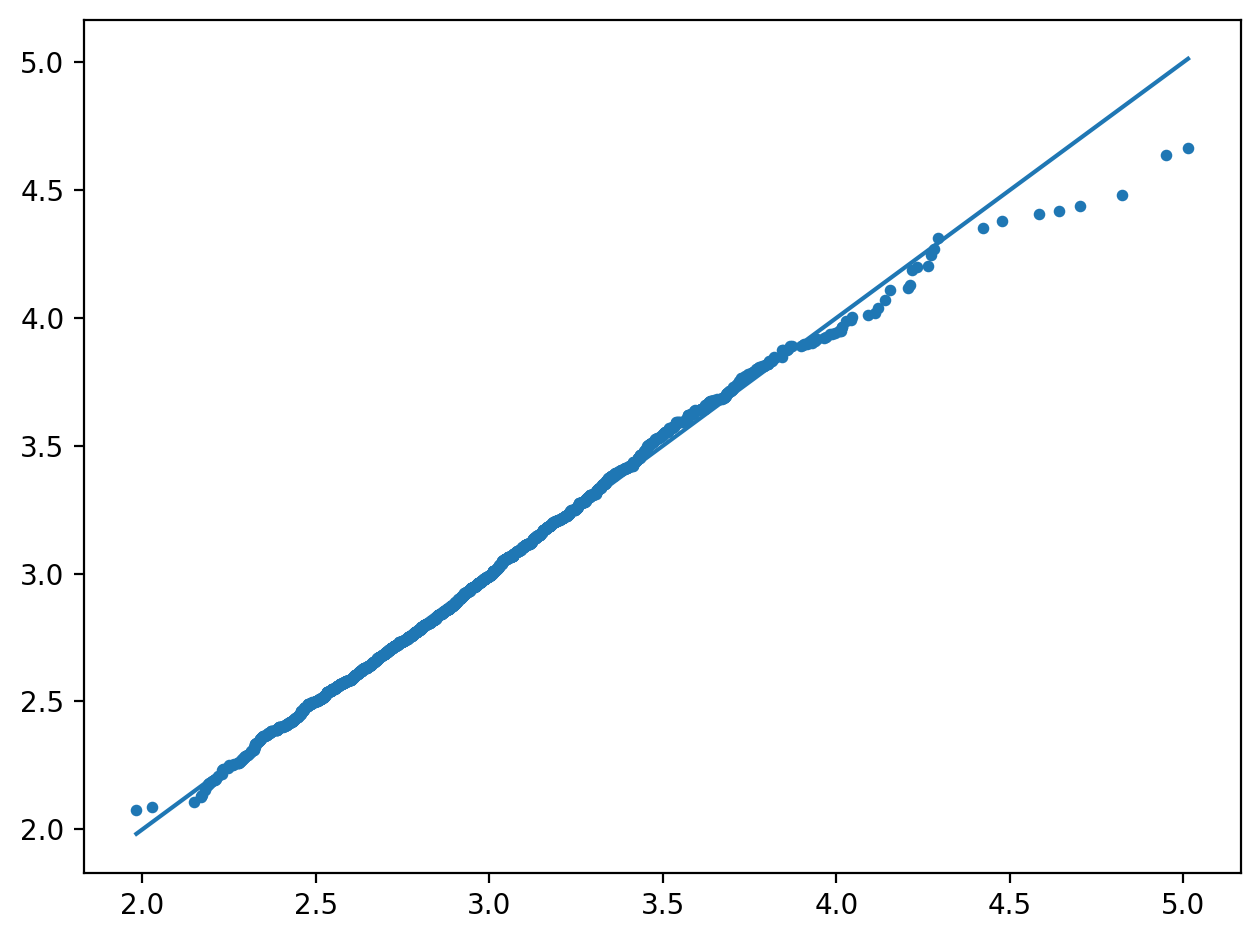}
\end{subfigure}

\vspace{2mm}

\begin{subfigure}[t]{0.48\linewidth}
  \includegraphics[width=\linewidth]{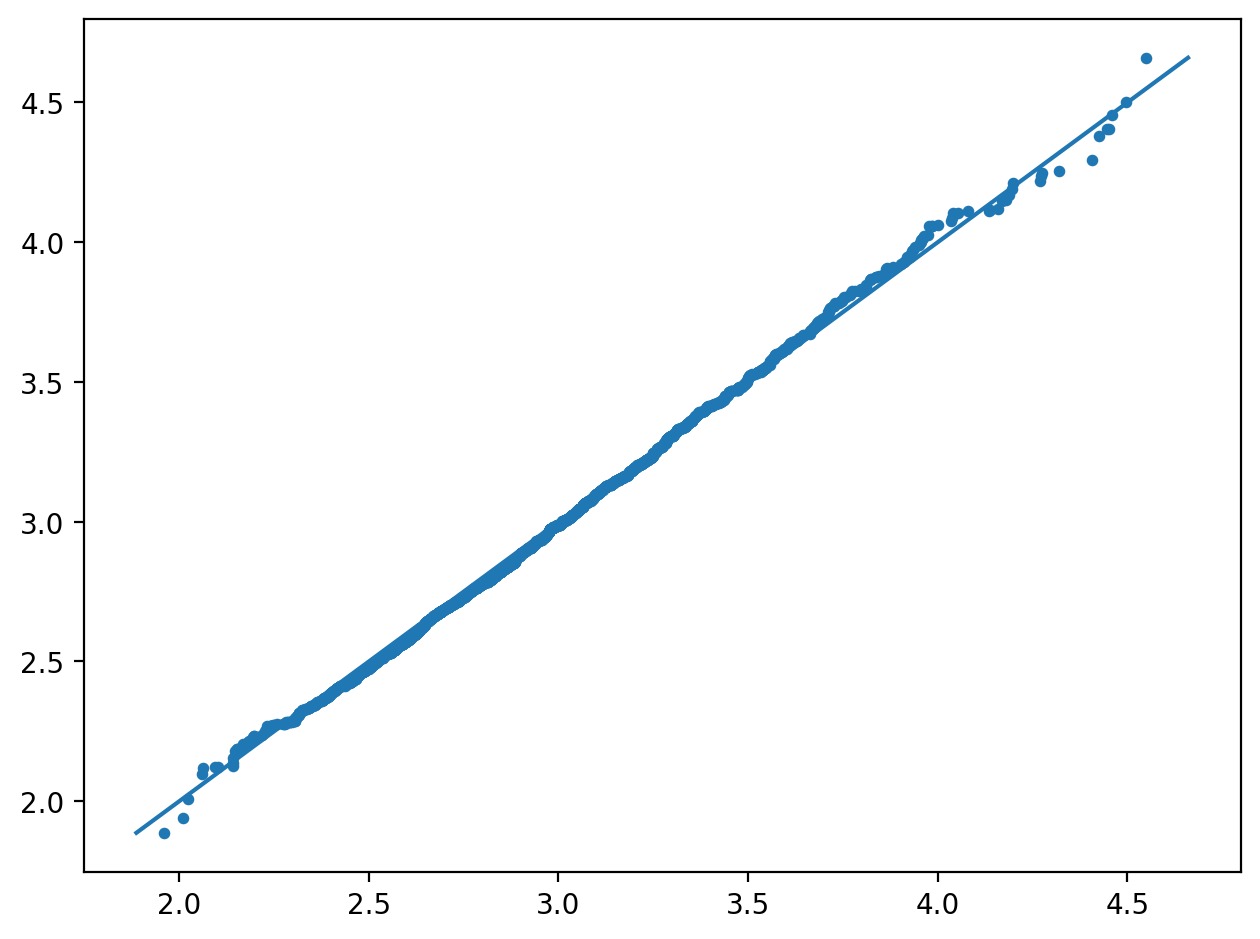}
\end{subfigure}\hfill
\begin{subfigure}[t]{0.48\linewidth}
  \includegraphics[width=\linewidth]{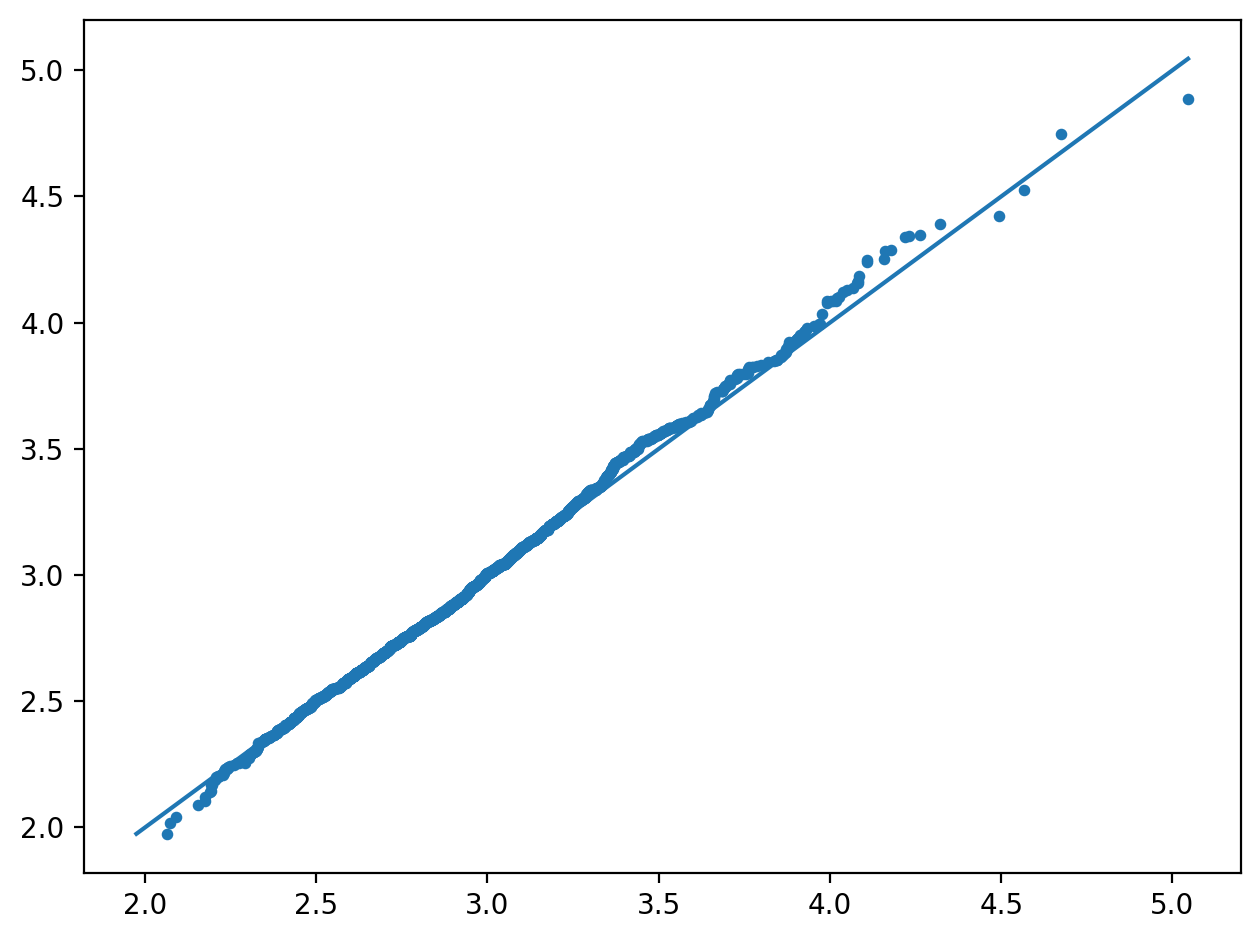}
\end{subfigure}

\caption*{\footnotesize $t = 12$}
\end{minipage}
\caption{
QQ plots for the random geometric graph network.
Left panels correspond to $t=8$ and right panels correspond to $t=12$.
Rows correspond to $n=\{100,200,400\}$ and columns correspond to
$\rho=\{0.6,0.8\}$ with $p=200$ and $R=2000$.
The horizontal axis reports the empirical distributions of $T_n$,
and the vertical axis reports the corresponding Gaussian distributions of $T_z$.
}
\label{fig:RGG_QQ_nu4_nu8}
\end{figure}

\begin{figure}[htbp]
  \centering

  \begin{subfigure}[t]{0.32\linewidth}
    \centering
    \includegraphics[width=\linewidth]{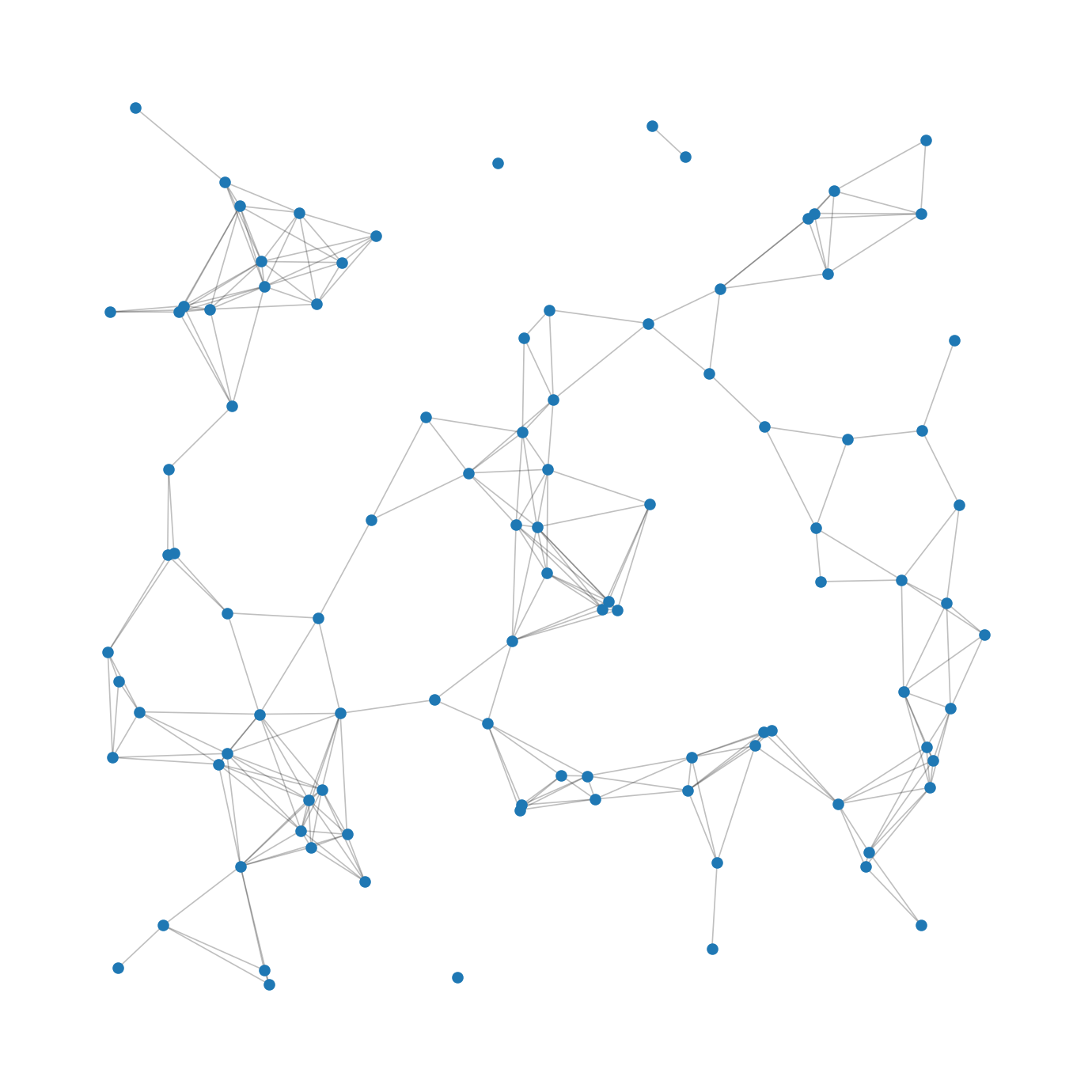}
    \caption{n=100}
  \end{subfigure}\hfill
  \begin{subfigure}[t]{0.32\linewidth}
    \centering
    \includegraphics[width=\linewidth]{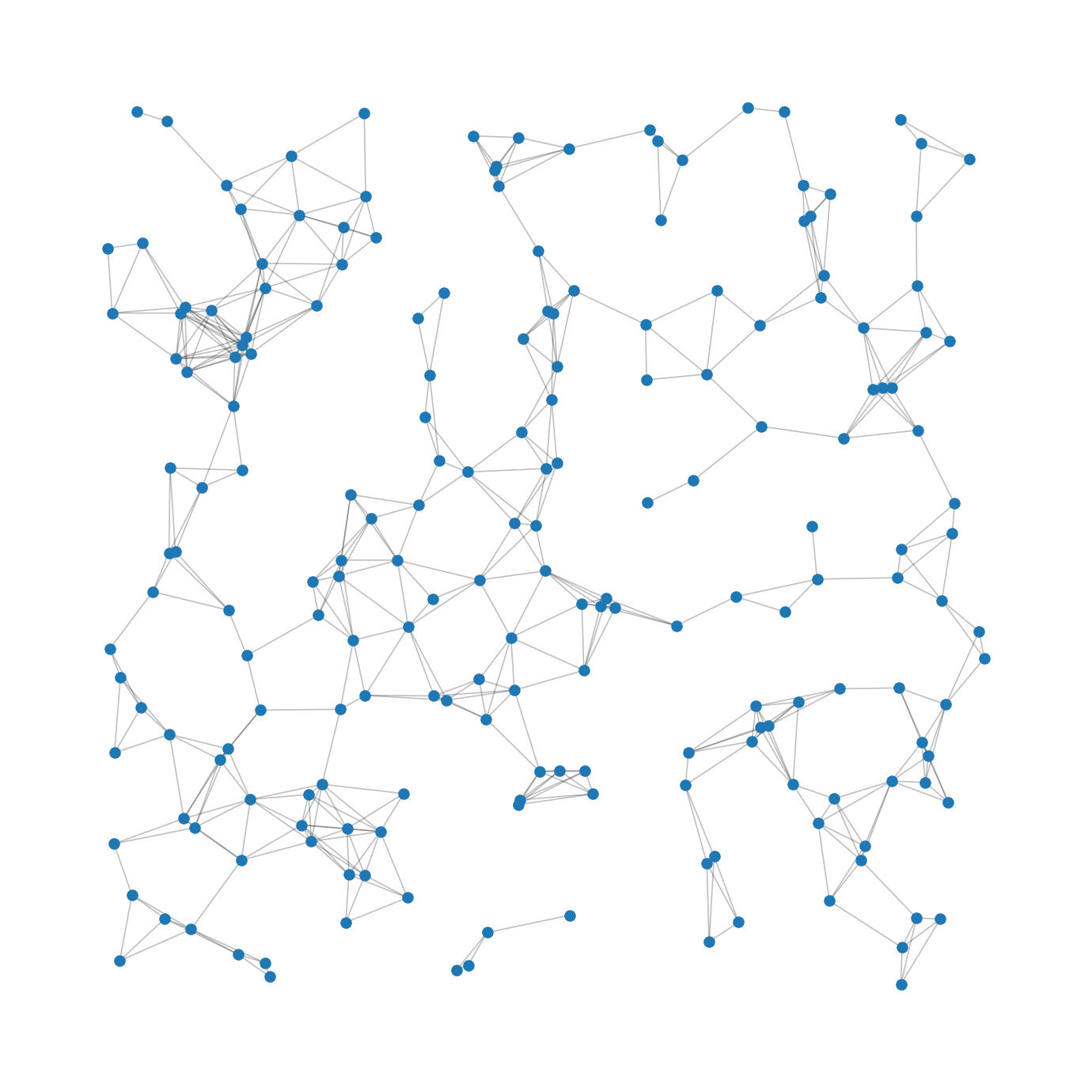}
    \caption{n=200}
  \end{subfigure}\hfill
  \begin{subfigure}[t]{0.32\linewidth}
    \centering
    \includegraphics[width=\linewidth]{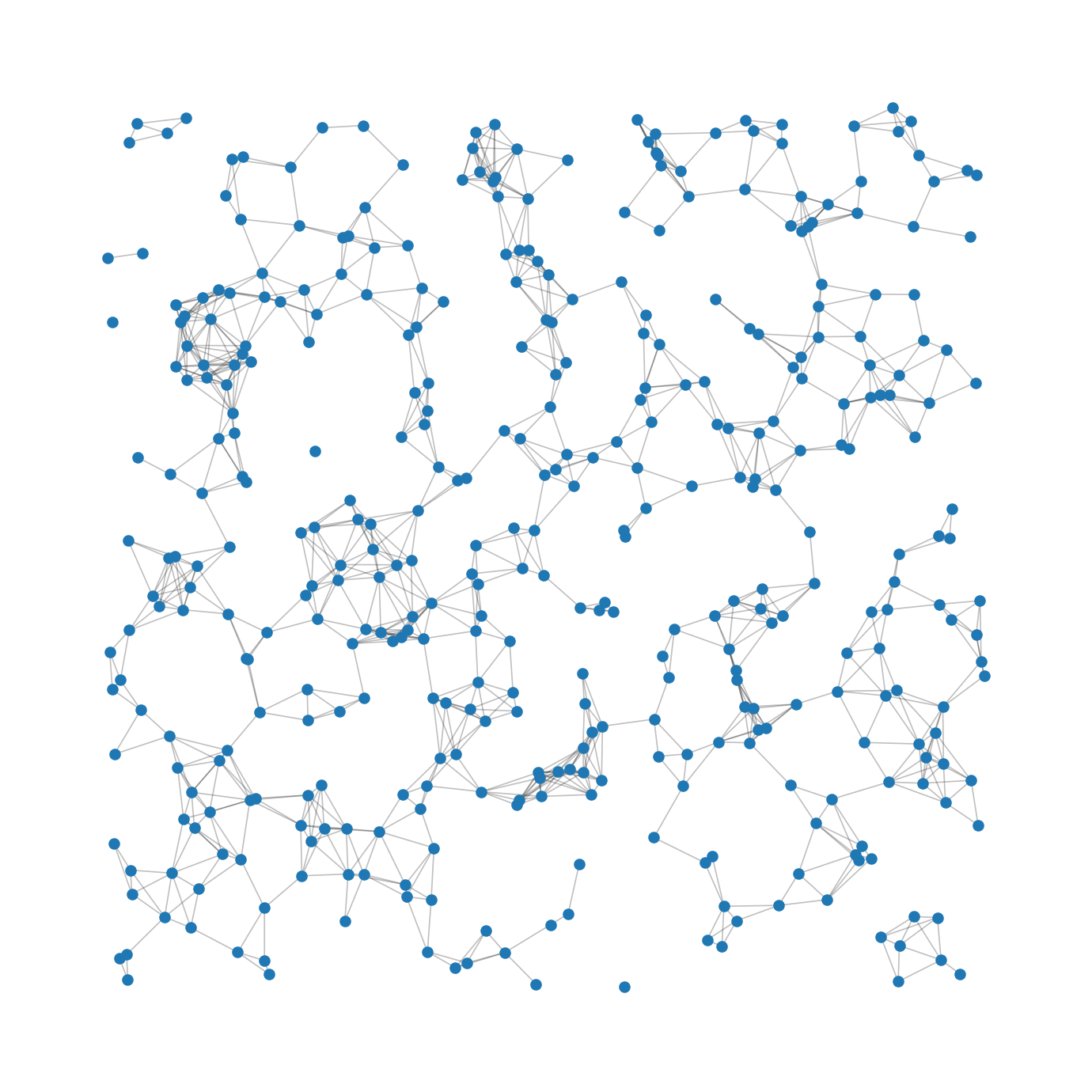}
    \caption{n=400}
  \end{subfigure}

  \caption{
Random geometric graphs with $\bar k=6$.
Panels (a)–(c) correspond to increasing number of points.
}
  \label{fig:rgg_1x3}
\end{figure}

The results are broadly consistent with our theoretical predictions.
First, the Gaussian approximation becomes less accurate when network dependence gets stronger, as reflected by larger values of $\gamma$. 
Second, network topology matters. 
The ring network generally produces a closer Gaussian approximation than the random geometric graph because the ring has a more homogeneous and sparser local structure. 
In contrast, the random geometric graph introduces degree heterogeneity and more irregular shell sizes, which slow down the rate of finite-sample approximation precision.

The improvement from increasing $n$ is visible, but not dramatic, especially for
the random geometric graph. 
This is consistent with the fact that approximation quality depends not only on sample size, but also on the interaction between dependence decay and network expansion. 
When the graph contains more local neighbors or more irregular neighborhoods, the effective dependence accumulated over network shells can remain non-negligible in moderate samples.

Tail behavior also plays an important role. 
Designs with $t=8$ display larger deviations from the Gaussian benchmark than designs with $t=12$, reflecting the effect of heavier-tailed shocks on high-dimensional maxima. 
Overall, the approximation performs better under weaker dependence, sparser and more regular
network structure, and lighter-tailed innovations in accordance with our theoretical predictions.

\section{Simulation Studies II: Statistical Inference}\label{sec:HAC}
In this section, we investigate the finite-sample performance of the proposed method of inference based on the network HAC covariance estimator. 
We evaluate the coverage accuracy of the HAC-based inference under different levels of network dependence and network density.
We focus on the more complicated random geometric graph network, and generate the data according to the DGP described in Section \ref{subsec:sim_dgp} with dimension $p=400$ and degrees of freedom $\nu=12$.  
For covariance estimation, we employ the feasible network HAC estimator
\(\widetilde{\Sigma}_n(b_n)\), which does not rely on knowledge of the expected value $\mu_0$.

To compute the HAC estimator, we follow the implementation strategy of
\cite{KMS2021}. 
Specifically, we employ the Parzen kernel
\[
w(x) =
\begin{cases}
1 - 6x^2 + 6x^3, & 0 \le x \le \tfrac{1}{2}, \\
2(1-x)^3, & \tfrac{1}{2} < x \le 1,\\
0, & \text{otherwise},
\end{cases}
\]
and choose the bandwidth according to
\[
b_n
=
c\frac{\log n}{\log(\mathrm{Avg.Deg.})}.
\]
Following the bandwidth specification suggested by \cite{KMS2021}, we allow the
constant \(c\) to vary across a range of values. 

We conduct sets of Monte Carlo simulations across a variety of dependence and network configurations, with \(\rho\in\{0.2,0.3,0.4,0.5\}\) and \(\bar k\in\{2,3,4,5\}\). 
For the bandwidth constant, we consider \(c\in\{1.2,1.3,1.4,1.5,1.6\}\). 
Among the values examined, \(c=1.5\) delivers the most stable coverage performance and is therefore used in the results reported in Table \ref{tab:rgg_cov_c15}.

\begin{table}[tbp]
\centering
\caption{Average network statistics and simulated coverage probabilities of the
95\% HAC-based confidence bands for bandwidth constant \(c=1.5\).}
\label{tab:rgg_cov_c15}
\begin{tabular}{cccccccc}
\toprule
\(\bar k\) & \(n\) & Diam. & Avg.Deg. &
\multicolumn{4}{c}{Simulated coverage} \\
\cmidrule(lr){5-8}
 &  &  &  & \(\rho=0.2\) & \(\rho=0.3\) & \(\rho=0.4\) & \(\rho=0.5\) \\
\midrule
2.0 & 200 & 7  & 1.75   & 0.975 & 0.973 & 0.978 & 0.975 \\
2.0 & 400 & 16 & 1.84   & 0.965 & 0.960 & 0.961 & 0.963 \\
2.0 & 800 & 10 & 1.9275 & 0.961 & 0.960 & 0.953 & 0.952 \\
\addlinespace
3.0 & 200 & 17 & 2.77   & 0.977 & 0.976 & 0.976 & 0.971 \\
3.0 & 400 & 18 & 2.90   & 0.966 & 0.966 & 0.962 & 0.956 \\
3.0 & 800 & 20 & 2.85   & 0.960 & 0.955 & 0.956 & 0.948 \\
\addlinespace
4.0 & 200 & 14 & 3.64   & 0.977 & 0.976 & 0.974 & 0.968 \\
4.0 & 400 & 42 & 3.845  & 0.966 & 0.964 & 0.954 & 0.944 \\
4.0 & 800 & 39 & 3.8075 & 0.957 & 0.954 & 0.950 & 0.942 \\
\addlinespace
5.0 & 200 & 13 & 4.48   & 0.979 & 0.980 & 0.973 & 0.964 \\
5.0 & 400 & 40 & 4.715  & 0.965 & 0.962 & 0.953 & 0.940 \\
5.0 & 800 & 65 & 4.84   & 0.964 & 0.956 & 0.946 & 0.932 \\
\bottomrule
\end{tabular}
\end{table}

This table reports the simulated coverage probabilities of
the proposed HAC-based confidence band for a range of dependence strengths
and network densities. Overall, the procedure delivers accurate finite-sample
inference, with coverage probabilities concentrated around the nominal
\(95\%\) level across all configurations considered.

Several patterns emerge from the table. First, the empirical coverage remains
remarkably stable over different network densities. Even when the expected
degree increases from \(\bar k=2\) to \(\bar k=5\), the coverage distortion is
generally small. 

Second, stronger dependence leads to a mild deterioration in coverage for most fixed network designs. For most combinations of
\(\bar k\)
 and \(n\), the coverage probability tends to
decrease as \(\rho\) increases from $0.2$ to $0.5$, reflecting the greater
difficulty of estimating the long-run covariance matrix under stronger network
dependence. Nevertheless, the resulting coverage remains close to the nominal
level in most cases.

Third, the finite-sample performance improves with network size. For
$n=800$, coverage probabilities are typically very close to \(95\%\), whereas
for $n=200$ the procedure exhibits slight over-coverage in some designs. This
pattern is consistent with the asymptotic theory and suggests that the proposed
HAC estimator provides a reliable approximation even in moderately sized
networks.

\section{Application: Heterogeneous Spillover Effects}\label{sec:application}
In this section, we revisit the network experiments studied by \cite{paluck2016changing}, \cite{AronowSamii2017}, and \cite{Leung2022}, where they examine the effects of an anti-conflict intervention on adolescent social norms related to antagonistic behaviors such as harassment, rumor spreading, social exclusion, and bullying. 

Focusing on average effects, \citet{Leung2022} found statistically insignificant spillover effects.
Since average effects aggregate potentially heterogeneous spillover effects across individuals, this finding may mask meaningful spillover effects present for particular subpopulations.
Motivated by this possibility, we investigate whether treatment spillovers exhibit systematic heterogeneity and whether significant conditional spillover effects emerge after conditioning on measures of network position and the social environment.



\subsection{Heterogeneous Spillover Effects: Estimation and Inference}
A researcher observes $(Y_i,D_i,W_i)_{i=1}^n$, where $Y_i$ is an outcome, $D_i$ is a binary treatment indicator, and $W_i$ is a covariate.
The outcome $Y_i$ indicates whether student $i$ reports wearing a wristband disseminated as
part of the program as a reward to students observed engaging in conflict-mitigating behavior, $D_i$ indicates whether the student $i$ is offered treatment, and $W_i$ captures measures of network position or the social environment.

In network settings, the outcome of unit $i$ is generally determined not only by its own treatment status but also by the treatment assignments of other units, that is, by the entire treatment vector $D=(D_i)_{i=1}^n$.
Following \citet{Leung2022}, we adopt the exposure mapping framework to provide a parsimonious characterization of both direct treatment and spillover exposures.
Let
\[
T_i=T(i,D,A)
\]
denote the exposure status of unit \(i\), taking values in the exposure space \(\mathcal{T}\subseteq\mathbb{R}\), where \(A\) denotes the observed network.

We can then index the potential outcomes by the exposure state, i.e., \(Y_i(t)\) for \(t \in \mathcal{T}\).
Our estimand of interest is the conditional average exposure effect
\begin{align*}
\tau(w;t,t')
=&
\mu(w,t)-\mu(w,t'),
\qquad\text{where}\\
\mu(w,t)
:=&
\mathbb E[Y_i(t)\mid W_i=w].
\end{align*}

\paragraph{Estimation.}
As in \citet{Leung2022}, we employ the inverse propensity score weighting (IPW) approach.
However, since our estimand $\tau(\,\cdot\,;t,t')$ is a nonparametric functional parameter rather than a scalar, we augment the conventional IPW estimator with kernel smoothing:
\begin{align*}
\hat\tau(w;t,t') =& \frac{ \sum_{i=1}^{n} K_{h}(W_i-w) Z_i(t,t') }{ \sum_{i=1}^{n} K_{h}(W_i-w) }, 
\qquad\text{where}
\\
Z_i(t,t') =& Y_i \left( \frac{\mathbf 1_i(t)}{\pi_i(t)} - \frac{\mathbf 1_i(t')}{\pi_i(t')} \right)
\qquad\text{and}\qquad
K_{h}(u)=h^{-1}K(u/h)
\end{align*}
with $K$ denoting the Gaussian kernel and $h$ denoting the bandwidth parameter.

To conduct statistical inference accounting for the first-order estimation effects of both the numerator and denominator in $\hat\tau(w,t,t')$, define the feasible score process
\[
\hat\psi_i(w;t,t')
=
\frac{
K_{h}(W_i-w)
}{
\hat f_W(w)
}
\left\{
Z_i(t,t')
-
\hat\tau(w;t,t')
\right\},
\]
where
$
\hat f_W(w)
=
\frac1n
\sum_{i=1}^{n}
K_{h}(W_i-w).
$
The centering term \(\hat\tau(w;t,t')\) accounts for the
first-order contribution of estimating the denominator of the
kernel ratio estimator.
We employ undersmoothing $h$ so that the smoothing bias is asymptotically
negligible relative to the stochastic error uniformly over the evaluation
points.
Let
\begin{align*}
\hat\psi_i(t,t')
&=
\Bigl(
\hat\psi_i(w_1;t,t'),
\dots,
\hat\psi_i(w_p;t,t')
\Bigr)^\top,
\\
\hat\tau(t,t')
&=
\Bigl(
\hat\tau(w_1;t,t'),
\dots,
\hat\tau(w_p;t,t')
\Bigr)^\top.
\end{align*}
The vector \(\hat\tau(t,t')\) summarizes the estimated conditional
spillover effects over the support of \(W_i\), while
\(\hat\psi_i(t,t')\) is used below for network HAC covariance
estimation and simultaneous inference.
\paragraph{Inference.}
Following the network HAC inference procedure developed in
Section~\ref{sec:inference}, we conduct simultaneous inference for the
vector of conditional spillover effects
\[
\hat\tau(t,t')
=
\Bigl(
\hat\tau(w_1;t,t'),
\dots,
\hat\tau(w_p;t,t')
\Bigr)^\top.
\]
Let
\[
\hat\psi_i(t,t')
=
\Bigl(
\hat\psi_i(w_1;t,t'),
\dots,
\hat\psi_i(w_p;t,t')
\Bigr)^\top
\]
denote the feasible score vector defined above, where each coordinate
accounts for the first-order effect of estimating the denominator of the
kernel ratio estimator. The covariance matrix of the first-order
stochastic component of
\[
\sqrt n
\Bigl(
\hat\tau(t,t')
-
\tau_{h}(t,t')
\Bigr)
\]
is estimated by the network HAC estimator
\[
\widetilde\Sigma_n(t,t')
=
\frac1n
\sum_{i=1}^{n}
\sum_{j=1}^{n}
\hat\psi_i(t,t')
\hat\psi_j(t,t')^\top
1\{\ell_A(i,j)\le b_n\},
\]
where
\[
\tau_{h}(t,t')
=
\Bigl(
\tau_{h}(w_1;t,t'),
\dots,
\tau_{h}(w_p;t,t')
\Bigr)^\top
\]
denotes the corresponding kernel-smoothed population target. We employ
the truncated kernel and the bandwidth \(b_n=2\), as recommended by
\citet{Leung2022}.

For each $\ell=1,\ldots,p$, let
$
\tilde\sigma_{\ell\ell}^2(t,t')
=
\widetilde\Sigma_{n,\ell\ell}(t,t')
$
denote the estimated variance of the \(\ell\)-th coordinate of
$
\sqrt n
\Bigl(
\hat\tau(t,t')
-
\tau(t,t')
\Bigr).
$
To obtain the simultaneous critical value, we simulate
\[
Z^*
\sim
N\!\left(
0,
\widetilde\Sigma_n(t,t')
\right)
\]
conditional on the data and compute
\[
c_{1-\alpha}
=
q_{1-\alpha}
\left(
\max_{1\le \ell\le p}
\left|
\frac{Z_\ell^*}
{\tilde\sigma_{\ell\ell}(t,t')}
\right|
\right),
\]
where \(q_{1-\alpha}(\cdot)\) denotes the conditional \((1-\alpha)\)-quantile.
The resulting \(1-\alpha\) confidence band for the conditional spillover function is
\[
\hat\tau(w_\ell;t,t')
\pm
\frac{
c_{1-\alpha}
\tilde\sigma_{\ell\ell}(t,t')
}{
\sqrt n
},
\qquad
\ell=1,\ldots,p.
\]

Under the linear-in-means model and complex contagion model considered in \cite{Leung2022}, network dependence decays exponentially with graph distance. Consequently, the score process underlying the kernel IPW estimator satisfies the dependence conditions required by the Gaussian approximation and network HAC inference results developed in Section~\ref{sec:inference}. We therefore apply the proposed simultaneous inference procedure to the conditional spillover estimator.

\subsection{Spillover Effects conditional on Socio-Demographic Homophily}

We explore heterogeneity in spillover effects along socio-demographic dimensions. For each student, we construct homophily measures based on the similarity between the student and his network neighbors in terms of race, gender, and grade. 

For race, students may belong to multiple racial categories and some observations contain missing or non-applicable entries. 
We define racial homophily as the proportion of observed neighbors whose race profile exactly matches that of student \(i\). Analogous measures are constructed for gender and grade homophily after excluding observations with missing values.

Figure~\ref{fig:homophily_dist} reports the distributions of the three homophily measures. Gender and grade homophily exhibit substantial mass points at one, indicating that many students are surrounded entirely by peers of the same gender or grade. Race homophily displays greater variation across individuals.
Because the three homophily measures are highly concentrated and partially overlapping, we summarize their common variation using the first principal component \(W_i\).

\begin{figure}[htbp]
\centering

\begin{subfigure}{0.32\textwidth}
    \centering
    \includegraphics[width=\linewidth]{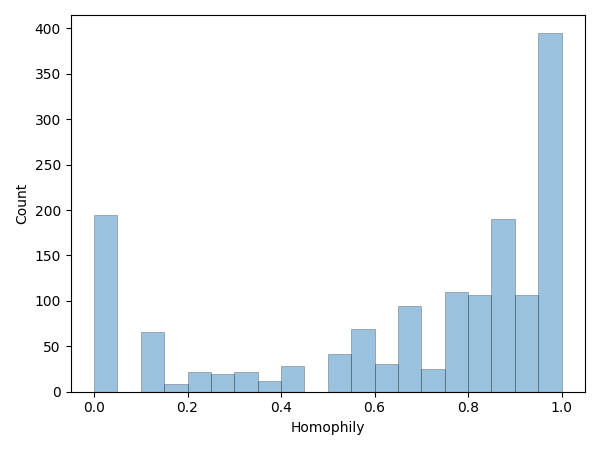}
    \caption{Race homophily}
\end{subfigure}
\hfill
\begin{subfigure}{0.32\textwidth}
    \centering
    \includegraphics[width=\linewidth]{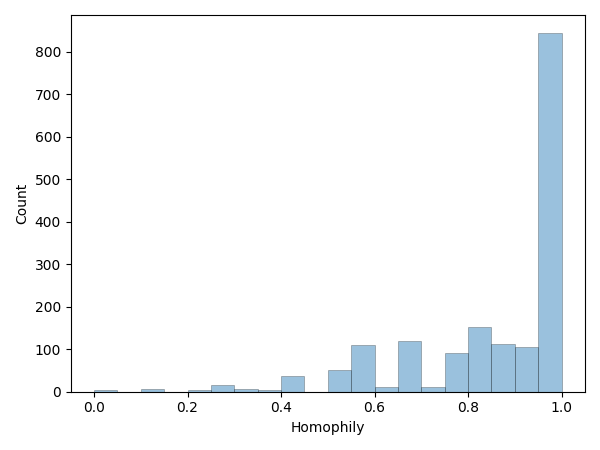}
    \caption{Gender homophily}
\end{subfigure}
\hfill
\begin{subfigure}{0.32\textwidth}
    \centering
    \includegraphics[width=\linewidth]{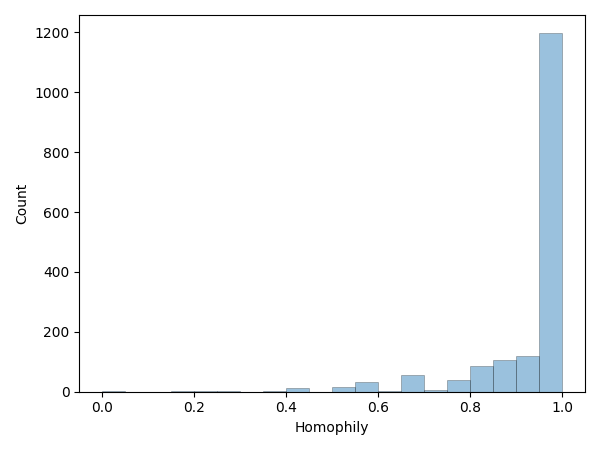}
    \caption{Grade homophily}
\end{subfigure}

\caption{Distribution of homophily measures.}
\label{fig:homophily_dist}
\end{figure}


To avoid unstable kernel estimates in regions with sparse support, we restrict the evaluation range of the homophily index to \(w \in [-1,1]\). All reported spillover curves and confidence bands are constructed over this range. We then conduct inference on the heterogeneous spillover effect conditional on the social homophily index \(W_i=w\). Figure~\ref{fig:cate_homo} reports the estimated spillover curve together with the 95\% confidence band.

\begin{figure}
    \centering
    \includegraphics[width=1\linewidth]{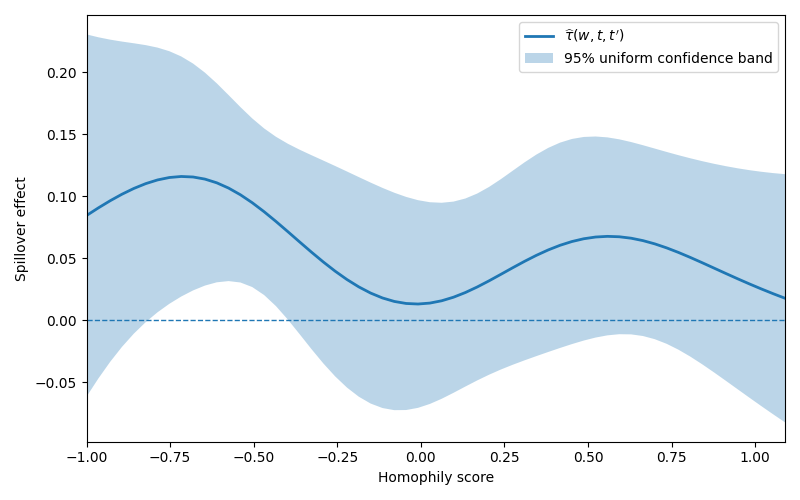}
    \caption{Spillover effects conditional on social homophily score.}
    \label{fig:cate_homo}
\end{figure}

The estimated spillover effect varies across the support of the homophily index, suggesting the presence of spillover effect heterogeneity. The curve is positive over most of the
reported range and reaches a local peak around \(W_i \approx -0.7\), but it does not
display a monotonic relationship with the homophily index. Although the point estimates
are positive for most evaluation points, the confidence band contains zero
over much of the support.

The confidence band excludes zero over a narrow region centered around \(W_i\approx -0.6\), providing evidence of positive spillover effects for students with these socio-demographic characteristics. Outside this region, the data do not provide sufficient evidence to distinguish the spillover effect from zero at the 95\% confidence level.

Overall, the results suggest that spillover effects are not constant across students and may depend on the socio-demographic composition of their friendship networks. However, the evidence for positive spillovers is concentrated in a localized region of the homophily index rather than being present uniformly across the population.

\section{Summary and Discussions}\label{sec:summary}
In this paper, we developed a Gaussian approximation theory for sums of high-dimensional network-dependent random vectors. Under suitable conditions on the decay of network dependence and the underlying network topology, we established Gaussian approximation results for the maximum norm and characterized admissible dimensionality under both finite-moment and sub-Weibull regimes. We also proposed a high-dimensional network HAC estimator and established its asymptotic convergence properties in high dimensions. Together, these theoretical results provide a foundation for asymptotically valid simultaneous inference for high-dimensional random vectors in the presence of network dependence.

Our simulation studies indicate that the proposed Gaussian approximation and inference procedure perform well in finite samples. We further illustrated the practical usefulness of the methodology through an empirical application examining how spillover effects vary with a continuous measure of network homophily, where the proposed procedure was used to construct confidence bands for the heterogeneous conditional spillover-effect function.

Gaussian approximation for high-dimensional vectors has proven useful in a
wide range of empirical applications. Beyond inference on the conditional
average treatment effect (CATE) function, as illustrated in our empirical
application, the proposed method is applicable to a broad class of functional
treatment-effect parameters, including continuous treatment effect functions, among others.

Beyond treatment-effect settings, high-dimensional inference methods have also
proven useful for inference on structural parameters that are partially
identified by conditional moment equalities or inequalities. Our results are
therefore expected to facilitate applications to inference in structural
network models as well.

\bibliographystyle{apalike}   
\bibliography{reference}
\clearpage
\appendix
\section*{Appendix}
\section{Proofs for the Main Theorems in Section \ref{sec:GA}}{\label{sec:proof}}
This section presents proofs for the two main theorems presented in section \ref{sec:GA}.

The proof strategy is as follows. 
We first localize the process by replacing each observation $X_{n,i}$ with its $m$-localized version, which depends only on shocks within graph distance $m$. 
We then partition the network into clusters and remove observations near cluster boundaries through an $m$-buffer construction. 
After localization, observations belonging to the interiors of different clusters depend on disjoint sets of shocks and are therefore independent across clusters. 
This allows us to apply the high-dimensional Gaussian approximation result for independent observations \citep{CCK2013} to the obtained independent cluster sums. 
Finally, the proof is completed by controlling the truncation error, the buffer contribution, and the covariance approximation error.

This section is organized as follows.
Sections \ref{sec:truncation}--\ref{sec:gaussian_to_gaussian} present auxiliary lemmas.
For ease of reading, we relegate proofs of these auxiliary lemmas to Appendix \ref{sec:appendix}.
Then, Theorems \ref{th:GA_q} and \ref{th:GA_sub} are proven in Section \ref{sec:proof_last}.

\subsection{Locality Truncation}\label{sec:truncation}
Recall the neighborhood $\mathcal N_n(i,m):=\{u:\,d_n(i,u)\le m\}$ and the $m$-localized $X_{n,i}^{(m)}$ defined in \eqref{eq:localize}.
Define the localized sum
\(
 T_{X,m}=\sum_{i\in V_n}X_{n,i}^{(m)}.
\)
To measure the truncation error under the $L_\infty$ norm, define
\begin{equation}{\label{eq:omega}}
\omega_{n,i,u,q}
:= \big\|\,|X_{n,i}-X_{n,i}^{*(u)}|_\infty\,\big\|_q
\qquad\text{and}\qquad
\Omega_{n,q}(m)=\sup_{i\in V_n}\sum_{u:d_n(i,u)\ge m} \omega_{n,i,u,q}.
\end{equation}
\begin{lemma}[Truncation Error under $L_{\infty}$-Dependence]\label{lem:trunc_omega}
For all $m\ge 0$,
\[
\sup_{ i\in V_n}
\big\|\,|X_{n,i}-X_{n,i}^{(m)}|_\infty\big\|_q
\;\le\;
\Omega_{n,q}(m)
\]
holds, and consequently,
\[
\Big\|\,|T_X-T_{X,m}|_\infty\Big\|_q
\;\le\;
 n\,\Omega_{n,q}(m),
\]
where $T_X:=\sum_{i\in V_n} X_{n,i}$ and $T_{X,m}:=\sum_{i\in V_n} X_{n,i}^{(m)}$.
Therefore, we have
\[f_1(y):=\mathbb P(|T_X-T_{X,m}|_\infty\ge y)\le y^{-q}n^q\Omega_{n,q}^q(m)\le y^{-q}n^{q}p^q\Psi_{n,q}^q(m).\]
\end{lemma}

\noindent
See Appendix \ref{sec:lem:trunc_omega} for a proof of this lemma.

\subsection{Buffer Control}
For $m\ge 0$, let the $m$-buffer within cluster $\mathcal C_g$ be defined by
\[
\mathcal B_g(m)
:=
\{i\in\mathcal C_g : d_n(i,\partial\mathcal C_g)\le m\},
\qquad\text{and let }\qquad
\mathcal B(m):=\bigcup_{g=1}^{G_n}\mathcal B_g(m).
\]
Let $\mathcal{C}_g^\circ\subset \mathcal{C}_g$ be the $m$-interior (after discarding the $m$-buffer $\mathcal{B}_g(m)$) defined by
\[
\mathcal{C}_g^\circ=\{i\in \mathcal{C}_g, i \notin \mathcal{B}_g(m)\}.
\] 
For $l_g:=|\mathcal{C}_g|$, we have $l_g \asymp l_n$ for all $g$ by Assumption \ref{ass:clusters}, and $n=\sum_{g=1}^{G_n} l_g$.
We then write
\begin{equation}{\label{eq:MG}}
M_g:=|\mathcal{C}_g^\circ|,\qquad n^\circ:=\sum_{g=1}^{G_n} M_g,\qquad\text{and\qquad}
\partial_n:=1-\frac{n^\circ}{n}=\frac{|\mathcal{B}(m)|}{n}.
\end{equation}
Define the full and localized sums of original and truncated random variables
\begin{equation}{\label{eq:tym}}
Y_g:=\sum_{i\in \mathcal{C}_g^\circ} X_{n,i},
\quad
Y_{g,m}:=\sum_{i\in \mathcal{C}_g^\circ} X_{n,i}^{(m)},
\quad
T_Y:=\sum_{g=1}^{G_n} Y_g,
\quad\text{and}\quad
T_{Y,m}:=\sum_{g=1}^{G_n} Y_{g,m}.
\end{equation}

The next lemma shows the bound for $f_2(y):=\mathbb P (|T_{X,m}-T_{Y,m}|_\infty>y) $ under two regimes.
  
\begin{lemma}{\label{lem:buftail}}${}$\\

\noindent
(i) Under Assumption \ref{ass:decay rate}, we have 

\[
f_2^*(y):=\mathbb P(|T_{X,m}-T_{Y,m}|_\infty\ge y)\le C_q\,\frac{pU_B(m)^{q/2}\Psi_{n,q}^q(0)}{y^{q}};
\]
\noindent
(ii) Under Assumption \ref{ass:subweibull}, we have
\[
 f_2^\diamond(y):=\mathbb P(|T_{X,m}-T_{Y,m}|_\infty\ge y)
\le
2p\exp\!\left(
-C_\beta\left(\frac{y}{\sqrt{U_B(m)}\,\Phi_{n,\psi_{\nu}}}
\right)^\beta\right),
\]
where $\beta=2/(1+2\nu)$,
\[U_B(m):=|\mathcal U_B(m)|,
\qquad\text{and}\qquad
\mathcal U_B(m)
:=\bigcup_{i\in\mathcal B(m)} \mathcal N_n(i,m)
\subset \bigcup_{g=1}^{G_n}\{u:\ d_n(u,\partial \mathcal{C}_g)\le 2m\}.
\]
\end{lemma}
\noindent
See Appendix \ref{sec:lem:buftail} for a proof of this lemma.

\subsection{Gaussian Approximation in Independent Clusters}
Recall the independent shocks $\varepsilon_n:=\{\varepsilon_{n,u}\}_{u\in V_n}$, and that
$X_{n,ij}$ is measurable with respect to $\varepsilon_n$ for each $i\in  V_n$ and for each coordinate $1\le j\le p$.
We obtain independent clusters $\mathcal{C}_1^\circ,\dots,\mathcal{C}_{G_n}^\circ$ (after truncation) such that the collections $\{X_{n,i}^{(m)}:i\in \mathcal{C}_g^\circ\}$ are independent across $g$. 
This follows because any two nodes belonging to distinct cluster interiors are separated by more than $2m$, so their $m$-neighborhoods are disjoint.

Let $Z_{g} \sim N(0,M_g B_g)$ and $Z_{g,m} \sim N(0,M_g\tilde B_g)$ be independent draws across $g$, where $M_g$ is interior cluster size defined in \eqref{eq:MG}, and the covariance matrices, $B_g$ and $\tilde B_g$, are given by 
\[
B_g=(b_{g,ij})^p_{i,j=1}=\mathrm{Cov}(Y_g/\sqrt{M_g})
\qquad\text{and}\qquad
\tilde B_g =(\tilde b_{g,ij})^p_{i,j=1}=\mathrm{Cov}(Y_{g,m}/\sqrt{M_g}),
\]
respectively.
Write $T_{Z}=\sum_{g=1}^{G_n} Z_g$ and $T_{Z,m}=\sum_{g=1}^{G_n} Z_{g,m}$.
\begin{lemma}
\label{lem:cluster-ga}
Let $D=\mathrm{diag}(d_{11},\dots,d_{pp})$ be a diagonal matrix.
Assume that there exist constants $c>0$ and $0<c_1<c_2$ such that
\[
c<\min_{1\le j\le p} d_{jj},
\qquad
c_1 \le \tilde b_{g,(jj)}/d_{jj} \le c_2,
\quad 1\le j\le p, \quad 1\le g\le G_n.
\]
Assume also that $\Psi_{n,q}(0)<\infty$ for some $q\ge4$.
Then, for all $\lambda\in(0,1)$,
\[h(\lambda,u_m(\lambda)):=
\sup_{t\in\mathbb R}
\Big|
\mathbb P\!\Big(
\big|D^{-1/2}T_{Y,m}/\sqrt n\big|_\infty \le t
\Big)
-
\mathbb P\!\Big(
\big|D^{-1/2}T_{Z,m}/\sqrt n\big|_\infty \le t
\Big)
\Big|
\]
\[
\le
C\Bigg[
G_n^{-1/8}
\big(\Psi^{3/4}_{n,3}(0)\vee \Psi^{1/2}_{n,4}(0)\big)
\big(\log(pG_n/\lambda)\big)^{7/8}
+
G_n^{-1/2}
\big(\log(pG_n/\lambda)\big)^{3/2}
u_m(\lambda)
+
\lambda
\Bigg],
\]
where $C$ depends only on $c,c_1,c_2$, $q$, and $u_m(\lambda):= u_{Y,m}(\lambda)\vee u_{Z,m}(\lambda)$. 

Here $ u_{Y,m}(\lambda)$ is defined as the infimum over all numbers $u > 0$ such that 
\[  \Pr(\frac{\sqrt{G_n}}{\sqrt{n}}d_{jj}^{-1/2}Y_{gj,m}\le u, 1 \le g \le G_n, 1\le j\le p)\ge 1-\lambda \]
and $ u_{Z,m}(\lambda)$ is defined as the infimum over all numbers $u > 0$ such that 
   \[ \Pr(\frac{\sqrt{G_n}}{\sqrt{n}}d_{jj}^{-1/2}Z_{gj,m}\le u, 1 \le g \le G_n, 1\le j\le p)\ge 1-\lambda.
\]
We can bound $u_{Y,m}(\lambda)$ as follows.
\subparagraph{(i) Finite-$q$ moment regime:}
If $\Psi_{n,q}(0)<\infty$ for some $q\ge4$, then
\[
u_{Y,m}(\lambda)
\lesssim
\Psi_{n,q}(0)
\left(
\frac{G_np}{\lambda}
\right)^{1/q}.
\]
\subparagraph{(ii) Sub-Weibull regime:}
If $\Phi_{n,\psi_\nu}<\infty$, then
\[
u_{Y,m}(\lambda)\lesssim \Phi_{n,\psi_{\nu}}\left(\log(pG_n/\lambda)\right)^{1/\beta}
\]
In addition, we can bound $u_{Z,m}(\lambda)$ as follows.
\[u_{Z,m}(\lambda)\le \sqrt{\frac{G_n\times M_g}{n}}\times \sqrt{\log(pG_n/\lambda)}\lesssim \sqrt{\log(pG_n/\lambda)}.  \]
\end{lemma}
\noindent
See Appendix \ref{sec:lem:cluster-ga} for a proof of this lemma.

\subsection{Gaussian to Gaussian Comparison}\label{sec:gaussian_to_gaussian}
\begin{lemma}{\label{lem:gtg}}
     Let
$D=(d_{ij})_{i,j=1}^p$ be a diagonal matrix such that there exist some constants
$0<C_1<C_2$ such that
\[
C_1 \le \sigma_{jj}/d_{jj} \le C_2
\qquad\text{for all }1\le j\le p,
\]
where $\Sigma_n=(\sigma_{ij})_{i,j=1}^p$ is the true covariance matrix of the network sum $\frac{1}{\sqrt{n}}T_X$, and $Z \sim N(0,\Sigma_n)$.
Then, we have\small
\[
\sup_{t\in\mathbb{R}}
\Bigl|
\mathbb{P}\bigl(\,|D^{-1/2}T_{Z,m}/\sqrt{n}|_\infty \le t\,\bigr)
-
\mathbb{P}\bigl(\,|D^{-1/2}Z|_\infty \le t\,\bigr)
\Bigr|
\;\lesssim\;\\
\pi\!\left(
\max_{1\le j\le p} d_{jj}^{-1}\,
(\Psi_{n,2}(m)\Psi_{n,2}(0)\sqrt{l_{n}}
+\partial_n)
\right),
\]\normalsize
where
\[
\pi(x)=x^{1/3}\Bigl(1\vee \log(p/x)\Bigr)^{2/3}
\quad\text{for }x>0.
\]
\end{lemma}
\noindent
See Appendix \ref{sec:lem:gtg} for a proof of this lemma.

\subsection{Proofs of Theorems \ref{th:GA_q} and \ref{th:GA_sub}}{\label{sec:proof_last}}
\begin{proof}[Proofs of Theorems \ref{th:GA_q} and \ref{th:GA_sub}]
First, apply Lemma \ref{lem:trunc_omega} to obtain 
$$
f_1(y)=\mathbb P(|T_X-T_{X,m}|_\infty\ge y)\le y^{-q}n^q\Omega_{n,q}^q(m)\le y^{-q}n^{q}p\Psi^q_{n,q}(m).
$$
Next, apply lemma \ref{lem:buftail} to obtain
\[
f_2^*(y)=\mathbb P(|T_{X,m}-T_{Y,m}|_\infty\ge y)\le C_q\,\frac{pU_B(m)^{q/2}\Psi_{n,q}^q(0)}{y^{q}},
\]
under Assumption \ref{ass:decay rate} (a finite $q$-th moment), and
 \[f_2^\diamond(y)=\mathbb P(|T_{X,m}-T_{Y,m}|_\infty\ge y)
\le
2p\exp\!\left(
-C_\beta\left(\frac{y}{\sqrt{U_B(m)}\,\Phi_{n,\psi_{\nu}}}
\right)^\beta\right).\]
under Assumption \ref{ass:subweibull} (sub-Weibull tails).
We have the geometric bound for $U_B(m)$:
\[
U_B(m)
\;\le\;
\sum_{g=1}^{G_n} \left(|\partial \mathcal{C}_g|\times\,N_n(2m)\right)=\eta_n(G_n)\times N_n(2m)\lesssim \eta_n(G_n)m^d
\]
under Assumption \ref{ass:volume}.

Recall the network covariance matrix $\Sigma_n$. 
Let $\Sigma_0 = \mathrm{diag}(\Sigma_n)$ be the diagonal matrix of $\Sigma_n$, 
and $D_0 = \mathrm{diag}(\sigma_{11}^{1/2}, \ldots, \sigma_{pp}^{1/2}) 
= \Sigma_0^{1/2}$. 
Consider the following normalized version of Gaussian Approximation:
\begin{equation}
\rho_n 
:= \sup_{u \ge 0} 
\left| 
\mathbb{P}\!\left( 
\,\bigl| D_0^{-1} T_X/\sqrt{n} \bigr|_\infty \le u 
\right)
-
\mathbb{P}\!\left(
\bigl| D_0^{-1} Z \bigr|_\infty \le u
\right)
\right|
\end{equation}

Under Assumptions \ref{ass:clusters} and \ref{ass:variance lb} and Condition (ii) in Theorems \ref{th:GA_q} and \ref{th:GA_sub}, the conditions of Lemmas \ref{lem:cluster-ga} and \ref{lem:gtg} are satisfied for our truncated cluster sum $T_{Y,m}$ -- see Appendix \ref{sec:valid} for details of this argument.
Therefore, for every $\lambda \in (0,1)$,
\[\sup_{t\in\mathbb R}
\Big|
\mathbb P\!\Big(
\big|D_0^{-1}T_{Y,m}/\sqrt n\big|_\infty \le t
\Big)
-
\mathbb P\!\Big(
\big|D_0^{-1}Z\big|_\infty \le t
\Big)
\Big|
\le
h(\lambda,u_m(\lambda))+\pi(\Psi_{n,2}(m)\Psi_{n,2}(0)\sqrt{l_{n}}
+\partial_n).
\]
Observe the Gaussian vector $D_0^{-1}Z$ has marginal variance $1$. For every $\eta>0$, we have 
\[
\sup_{t\in\mathbb R}
\mathbb P\!\left(
\bigl|
|D_0^{-1}Z|_\infty - t
\bigr|
\le \eta
\right)
\lesssim
\eta\sqrt{\log p}.
\]
By the triangle inequality, for every $\eta>0$, we have
\begin{align*}
&\sup_{t\in\mathbb R}
\Big|
\mathbb P\!\left(
\bigl|D_0^{-1}T_X/\sqrt n\bigr|_\infty > t
\right)
-
\mathbb P\!\left(
\bigl|D_0^{-1}T_{Y,m}/\sqrt n\bigr|_\infty > t
\right)
\Big|
\\
\le&
\mathbb P\!\left(
\bigl|D_0^{-1}(T_X-T_{Y,m})/\sqrt n\bigr|_\infty > \eta
\right)
+
\sup_{t\in\mathbb R}
\mathbb P\!\left(
\bigl|
|D_0^{-1}T_{Y,m}/\sqrt n|_\infty - t
\bigr|
\le \eta
\right).
\end{align*}
Therefore, for every $\eta>0$ and $\lambda\in(0,1)$,
\begin{equation}{\label{eq:rho}}
\rho_n
\lesssim
h(\lambda,u_m(\lambda))
+
\pi(\Psi_{n,2}(m)\Psi_{n,2}(0)\sqrt{l_{n}}
+\partial_n)+
f_1(\sqrt{n}\eta )+f_2(\sqrt{n}\eta )
+
\eta\sqrt{\log p}.
\end{equation}
To satisfy \(\rho_n \to 0\), it is sufficient that, for some $\eta>0$ and $\lambda\in(0,1)$, we  can choose $m$ such that every term on the
right-hand side of \eqref{eq:rho} converges to zero.
We now branch into the cases under consideration by Theorems \ref{th:GA_q} and \ref{th:GA_sub}.


\bigskip\noindent
\textit{Proof of Theorem \ref{th:GA_q}:}
Assumptions \ref{ass:decay rate} and \ref{ass:volume} yield
\[
\Psi_{n,q}(0)\asymp1
\qquad\text{and}\qquad
\Psi_{n,q}(m)\lesssim(1+m)^d\rho^m.
\]
See Remark \ref{rem:local}.
Choose  \(m\asymp\log(np)\) with a sufficiently
large constant. 
Then,  we can choose $\eta \asymp p^{1/2q} (\frac{U_B(m)}{n})^{1/4}(\log p)^{-1/4}. $ 
With \(m\asymp\log(np)\),  Condition (iii) in the statement of the theorem implies
\[U_B(m)\lesssim
 \eta_n(G_n)
m^{d}
\ll
np^{-2/q}(\log p)^{-1}.
\]
Therefore, we have \[
p^{1/q} \sqrt{\frac{U_B(m)}{n}} \ll \eta \ll (\log p)^{-1/2},
\]
and hence,
\[
f_2^*(\sqrt n\eta)\to0.
\]
Also, with the above choice of $\eta$, we have
$
\eta\sqrt{\log p}\to0
$
and
$
n^{1/2}p\eta^{-1}\Psi_{n,q}(m)\to0,
$
so that
\[
f_1(\sqrt n\eta)\to0.
\]

Next, since \(m \asymp \log(np)\) with sufficiently large constant, we have
\[
\Psi_{n,2}(m)\Psi_{n,2}(0)\sqrt{l_{n}}\ll \Psi_{n,2}(m)\Psi_{n,2}(0)n^{1/2} \ll (\log p)^{-2}.
\]
Moreover,  Condition (iii) in the statement of the theorem gives
\[
\partial_n\lesssim \frac{\eta_n(G_n)m^{d}}{n}
\ll
p^{-2/q}(\log p)^{-1}\ll
(\log p)^{-2}.
\]
Therefore,
\[
\pi\!\left(
\Psi_{n,2}(m)\Psi_{n,2}(0)\sqrt{l_n}
+\partial_n
\right)\to0.
\]

Finally, Condition (i) in the statement of the theorem gives 
\[
G_n
\gg
\max\{
{\log(np)}^{7},
p^{2/(q-2)}
{\log(np)}^{3q/(q-2)}\}.
\]
As $\Psi^{3/4}_{n,3}(0)\asymp \Psi^{1/2}_{n,4}(0)=O(1)$, $u_m(\lambda):= u_{Y,m}(\lambda)\vee u_{Z,m}(\lambda)\lesssim
\Psi_{n,q}(0)
\left(
\frac{G_np}{\lambda}
\right)^{1/q}.$
We, therefore, have 
\begin{align*}
h(\lambda,u_m(\lambda))&=G_n^{-1/8}
\big(\Psi^{3/4}_{n,3}(0)\vee \Psi^{1/2}_{n,4}(0)\big)
\big(\log(pG_n/\lambda)\big)^{7/8}
+
G_n^{-1/2}
\big(\log(pG_n/\lambda)\big)^{3/2}
u_m(\lambda)
+
\lambda
\\
&\to0.
\end{align*}
This completes a proof of Theorem  \ref{th:GA_q}.

\bigskip\noindent
\textit{Proof of Theorem \ref{th:GA_sub}:}
Assumption \ref{ass:decay rate}, \ref{ass:volume}, and \ref{ass:subweibull} yield $\Phi_{n,\psi_\nu}\asymp1$. 
We choose  $m\asymp\log (np)$. 
Then, we can choose $\eta\asymp (\frac{U_B(m)}{n})^{1/4}(\log p)^{-\frac{1}{4}+\frac{1}{2\beta}}$, where $\beta=2/(1+2\nu)$.
Condition (iii) in the statement of the theorem implies
\[
U_B(m)\lesssim\eta_n(G_n)m^{d}
\ll
n(\log p)^{-1-2/\beta}.
\]
Therefore, we have \[
\sqrt{\frac{U_B(m)}{n}}(\log p)^{1/\beta}
\ll
\eta
\ll
(\log p)^{-1/2},
\]
and hence, 
\[
f_2^\diamond(\sqrt n\eta)\to0.
\]
Also, with the above choice of $\eta$, we have
$
\eta\sqrt{\log p}\to0
$
and
$
n^{1/2}p\eta^{-1}\Psi_{n,q}(m)\to0,
$
so that
\[
f_1(\sqrt n\eta)\to0.
\]
Moreover, similarly to the case of the finite $q$-th moment regime, by the choice
\(m\asymp\log(np)\) with sufficiently large constant,
\[
\Psi_{n,2}(m)\Psi_{n,2}(0)\sqrt{l_{n}}\ll \Psi_{n,2}(m)\Psi_{n,2}(0)n^{1/2} \ll (\log p)^{-2},
\]
and Condition (iii) in the statement of the theorem gives
\[
\partial_n\lesssim \frac{\eta_n(G_n)m^{d}}{n}
\ll
(\log p)^{-2-2\nu}\ll
(\log p)^{-2}.
\]
Therefore,
\[
\pi\!\left(
\Psi_{n,2}(m)\Psi_{n,2}(0)\sqrt{l_n}
+\partial_n
\right)\to0.
\]

Finally, Condition (i) in the statement of the theorem gives
\[ G_n\gg (\log(np))^{\max\{7,\,4+2\nu\}}.
\]
As $\Psi^{3/4}_{n,3}(0)\asymp \Psi^{1/2}_{n,4}(0)=O(1)$, $u_m(\lambda):= u_{Y,m}(\lambda)\vee u_{Z,m}(\lambda)\lesssim
\Phi_{n,\psi_\nu}
\left(
\log(pG_n/\lambda)
\right)^{(1+2\nu)/2}.$
We have \[
h(\lambda,u_m(\lambda))=G_n^{-1/8}
\big(\Psi^{3/4}_{n,3}(0)\vee \Psi^{1/2}_{n,4}(0)\big)
\big(\log(pG_n/\lambda)\big)^{7/8}
+
G_n^{-1/2}
\big(\log(pG_n/\lambda)\big)^{3/2}
u_m(\lambda)
+
\lambda\to0.
\]
 This completes a proof of Theorem \ref{th:GA_sub}.
\end{proof}

\section{Proof of Lemmas Introduced in Appendix \ref{sec:proof}}\label{sec:appendix}
\subsection{Auxiliary Lemmas}
\begin{lemma}[Moment Bound for Shock-Generated Sums]
\label{lem:shock-moment}
Let
$
S_j=\sum_{i\in I}X_{n,ij}^{(m)},
$
where \(S_j\) is measurable with respect to the shock field
\(\{\varepsilon_{n,u}:u\in\mathcal U\}\).
Assume
$
E(S_j)=0.
$
Then, for every $q\ge2$,
\[
\|S_j\|_q
\le
Cq^{1/2}
\left[
\sum_{u\in\mathcal U}
\left(
\sum_{i\in I}\delta_{n,i,u,q,j}
\right)^2
\right]^{1/2}.
\]
In particular, if
$
\sum_{i\in I}\delta_{n,i,u,q,j}
\le
\Delta_{n,0,q,j}
$
for all $u\in\mathcal U$,
then
\[
\|S_j\|_q
\le
Cq^{1/2}
|\mathcal U|^{1/2}
\Delta_{n,0,q,j}.
\]

Under Assumption~\ref{ass:subweibull}, therefore, we have
\[
\|S_j\|_q
\le
Cq^{1/\beta}
|\mathcal U|^{1/2}
\Phi_{n,\psi_\nu},
\qquad
\beta=\frac{2}{1+2\nu},
\]
and hence
\[
\|S_j\|_{\psi_{1/\beta}}
\le
C|\mathcal U|^{1/2}\Phi_{n,\psi_\nu}.
\]
Consequently,
\[
P(|S_j|>t)
\le
2\exp\left[
-c
\left(
\frac{t}
{|\mathcal U|^{1/2}\Phi_{n,\psi_\nu}}
\right)^\beta
\right].
\]
\end{lemma}
\begin{proof}
Fix an arbitrary ordering of the shock index set
\[
\mathcal U=\{u_1,\ldots,u_{|\mathcal U|}\},
\]
and define
\[
\mathcal F_r
=
\sigma(\varepsilon_{n,u_1},\ldots,\varepsilon_{n,u_r})
\qquad\text{ for each }
r=0,\ldots,|\mathcal U|.
\]
Let $S_j^{*(u_r)}$ denote the coupled version of $S_j$ obtained by replacing
$\varepsilon_{n,u_r}$ with an independent copy while leaving all other innovations
unchanged. Define
\[
D_r
=
E\!\left[
S_j-S_j^{*(u_r)}
\mid
\mathcal F_r
\right].
\]
Then, $\{D_r,\mathcal F_r\}$ is a martingale difference sequence and
\[
S_j=\sum_{r=1}^{|\mathcal U|}D_r.
\]
By the contraction property of conditional expectation, we have
\[
\|D_r\|_q
\le
\|S_j-S_j^{*(u_r)}\|_q.
\]
Moreover,
\[
S_j-S_j^{*(u_r)}
=
\sum_{i\in I}
\left(
X_{n,ij}^{(m)}-X_{n,ij}^{(m),*(u_r)}
\right),
\]
so Minkowski's inequality gives
\[
\|D_r\|_q
\le
\sum_{i\in I}
\|X_{n,ij}^{(m)}-X_{n,ij}^{(m),*(u_r)}\|_q
\le
\sum_{i\in I}\delta_{n,i,u_r,q,j}.
\]
Therefore, by Burkholder's inequality,
\[
\|S_j\|_q
=
\left\|
\sum_{r=1}^{|\mathcal U|}D_r
\right\|_q
\le
Cq^{1/2}
\left(
\sum_{r=1}^{|\mathcal U|}
\|D_r\|_q^2
\right)^{1/2},
\]
and hence
\[
\|S_j\|_q
\le
Cq^{1/2}
\left[
\sum_{u\in\mathcal U}
\left(
\sum_{i\in I}\delta_{n,i,u,q,j}
\right)^2
\right]^{1/2}.
\]
The simplified bound follows immediately if
\[
\sum_{i\in I}\delta_{n,i,u,q,j}
\le
\Delta_{n,0,q,j}
\quad
\text{uniformly over }u\in\mathcal U.
\]
Under Assumption~\ref{ass:subweibull},
\[
\Delta_{n,0,q,j}
\le
Cq^\nu\Phi_{n,\psi_\nu}.
\]
Thus,
\[
\|S_j\|_q
\le
Cq^{1/2+\nu}
|\mathcal U|^{1/2}
\Phi_{n,\psi_\nu}
=
Cq^{1/\beta}
|\mathcal U|^{1/2}
\Phi_{n,\psi_\nu},
\qquad
\beta=\frac{2}{1+2\nu}.
\]
The Orlicz norm bound follows from the definition of the $\psi_{1/\beta}$ norm,
and the final probability inequality follows from the standard tail bound for
$\psi_{1/\beta}$ random variables.
\end{proof}

\begin{lemma}[Covariance Bounds under Network Dependence]
\label{lem:covariance}
For each $u\in V_n$, let
$X_{n,i}^{*(u)}$ denote the coupled version of $X_{n,i}$ obtained by replacing
$\varepsilon_u$ with an independent copy while leaving all other innovations
unchanged.

\begin{enumerate}[label=(\roman*),leftmargin=2.2em]

\item For every $i,i'\in V_n$ and $j,k\le p$,
\[
\left|
\mathrm{Cov}(X_{n,ij},X_{n,i'k})
\right|
\le
C
\sum_{u\in V_n}
\delta_{n,i,u,2,j}\,
\delta_{n,i',u,2,k}.
\]

\item Consequently,
\[
\sup_{i\in V_n}
\sum_{i'\in V_n}
\left|
\mathrm{Cov}(X_{n,ij},X_{n,i'k})
\right|
\le
C\Psi_{n,2}^2(0),
\qquad
j,k=1,\ldots,p.
\]
\end{enumerate}
\end{lemma}

\begin{proof}
Part (i) follows directly from the Efron--Stein (replace-one-shock)
covariance inequality.

For part (ii), summing the bound in part (i) over $i'$ gives
\[
\sum_{i'\in V_n}
|\mathrm{Cov}(X_{n,ij},X_{n,i'k})|
\le
C
\sum_{u\in V_n}
\delta_{n,i,u,2,j}
\sum_{i'\in V_n}
\delta_{n,i',u,2,k}.
\]
Using
\[
\sup_{u}
\sum_{i'}
\delta_{n,i',u,2,k}
\le
\Psi_{n,2}(0)
\qquad\text{and}\qquad
\sup_i
\sum_u
\delta_{n,i,u,2,j}
\le
\Psi_{n,2}(0)
\]
yields
\[
\sup_i
\sum_{i'}
|\mathrm{Cov}(X_{n,ij},X_{n,i'k})|
\le
C\Psi_{n,2}^2(0),
\]
as claimed.
\end{proof}

\subsection{Proof of Lemma \ref{lem:trunc_omega}}\label{sec:lem:trunc_omega}
Fix $i\in V_n$.
Recall
\[
\mathcal F_{i,m}
=
\sigma\{\varepsilon_u:d_n(i,u)\le m\}
\qquad\text{and}\qquad
X_{n,i}^{(m)}
=
E(X_{n,i}\mid\mathcal F_{i,m}).
\]
Let
\[
U_{i,>m}
=
\{u:d_n(i,u)>m\}
=
\{u_1,\ldots,u_{n_m}\},
\]
where the ordering is arbitrary.
Define
\[
\mathcal G_0=\mathcal F_{i,m},\qquad
\mathcal G_\ell
=
\sigma(\mathcal F_{i,m},
\varepsilon_{u_1},\ldots,\varepsilon_{u_\ell}),
\quad
1\le\ell\le n_m .
\]
By Doob's martingale decomposition, we have
\[
X_{n,i}-X_{n,i}^{(m)}
=
\sum_{\ell=1}^{n_m}
\Big(
E(X_{n,i}\mid\mathcal G_\ell)
-
E(X_{n,i}\mid\mathcal G_{\ell-1})
\Big).
\]

Let $X_{n,i}^{*(u_\ell)}$ denote the coupled variable obtained by replacing
$\varepsilon_{u_\ell}$ with an independent copy while leaving all other
innovations unchanged. Standard coupling arguments imply
\[
E(X_{n,i}\mid\mathcal G_\ell)
-
E(X_{n,i}\mid\mathcal G_{\ell-1})
=
E\!\left(
X_{n,i}-X_{n,i}^{*(u_\ell)}
\mid
\mathcal G_\ell
\right).
\]
Hence,
\[
|X_{n,i}-X_{n,i}^{(m)}|_\infty
\le
\sum_{\ell=1}^{n_m}
E\!\left(
|X_{n,i}-X_{n,i}^{*(u_\ell)}|_\infty
\mid
\mathcal G_\ell
\right)
\quad\text{a.s.}
\]
Applying Minkowski's inequality and the $L^q$ contraction property of conditional
expectation yields
\[
\bigl\|
|X_{n,i}-X_{n,i}^{(m)}|_\infty
\bigr\|_q
\le
\sum_{u:d_n(i,u)>m}
\bigl\|
|X_{n,i}-X_{n,i}^{*(u)}|_\infty
\bigr\|_q
\le
\Omega_{n,q}(m),
\]
where the last inequality follows from the definition of
$\Omega_{n,q}(m)$ given in \eqref{eq:omega}.

Finally,
\[
\bigl\|
|T_X-T_{X,m}|_\infty
\bigr\|_q
\le
\sum_{i\in V_n}
\bigl\|
|X_{n,i}-X_{n,i}^{(m)}|_\infty
\bigr\|_q
\le
n\,\Omega_{n,q}(m),
\]
and the convergence in probability follows immediately from Markov's inequality.
By the definition of $\Psi_{n,q}(m)$ given in \eqref{eq:psi},
\[
\Omega_{n,q}(m)
\le
p\Psi_{n,q}(m).
\]
This completes a proof.
\qed

\subsection{Proof of Lemma \ref{lem:buftail}}\label{sec:lem:buftail}
Fix $j\le p$ and write
\[
S_j^{buf}(m)
=
\sum_{i\in\mathcal B(m)}X_{n,ij}^{(m)}.
\]
Since $S_j^{buf}(m)$ is generated by the shock set
$\mathcal U_B(m)$. By Lemma~\ref{lem:shock-moment}, we have \[
\|S_j^{buf}(m)\|_q
\le
Cq^{1/2}
\left[
\sum_{u\in\mathcal U_B(m)}
\left(
\sum_{i\in\mathcal B(m)}\delta_{n,i,u,q,j}
\right)^2
\right]^{1/2}.
\]
Hence, by Markov's inequality,
\[
\mathbb P\left(
\left|S_j^{\mathrm{buf}}(m)\right|>t
\right)
\le
\frac{C_q}{t^q}
\left[
U_B(m)
\right]^{q/2}\Psi_{n,q}^q(0),
\]
 Therefore, a union bound over
\(j=1,\ldots,p\) yields
\[
\mathbb P\left(
\left|S^{\mathrm{buf}}(m)\right|_\infty>t
\right)
\le
p\frac{C_q}{t^q}
\left[
U_B(m)
\right]^{q/2}\Psi_{n,q}^q(0).
\]

It remains to consider the sub-Weibull case. By Lemma~\ref{lem:shock-moment},
with $\mathcal U=\mathcal U_B(m)$ and $I=\mathcal B(m)$,
\[
\|S_j^{buf}(m)\|_q
\le
Cq^{1/\beta}
U_B(m)^{1/2}
\Phi_{n,\psi_\nu},
\qquad
\beta=\frac{2}{1+2\nu}.
\]
Equivalently,
\[
\|S_j^{buf}(m)\|_{\psi_{1/\beta}}
\le
C
U_B(m)^{1/2}
\Phi_{n,\psi_\nu}.
\]
Hence, for every $t>0$,
\[
P\left(|S_j^{buf}(m)|>t\right)
\le
2\exp\left[
-c
\left(
\frac{t}
{U_B(m)^{1/2}\Phi_{n,\psi_\nu}}
\right)^\beta
\right].
\]
Applying a union bound over $j=1,\ldots,p$ gives
\[
P\left(|S^{buf}(m)|_\infty>t\right)
\le
2p\exp\left[
-c
\left(
\frac{t}
{U_B(m)^{1/2}\Phi_{n,\psi_\nu}}
\right)^\beta
\right].
\]
This completes a proof.
\qed

\subsection{Proof of Lemma \ref{lem:cluster-ga}}\label{sec:lem:cluster-ga}
Fix $j \le p$. By the definition in \eqref{eq:tym},
We have the block sums 
\(
Y_{g,m}
:=
\sum_{i\in \mathcal C_g^\circ}X_{n,i}^{(m)}.
\)
To be consistent with the normalization in
\cite{CCK2013}, define
\[
\tilde X_g
:=
\sqrt{\frac{G_n}{n}}\,Y_{g,m}.
\]
Then
\[
\frac{T_{Y,m}}{\sqrt n}
=
\frac{1}{\sqrt{G_n}}
\sum_{g=1}^{G_n}\tilde X_g,
\]
which is the normalized sum of $G_n$ independent random vectors. Hence, it
suffices to verify the moment conditions required by \cite{CCK2013}.

 For $k>0$, define
\[
M_k
:=
\max_{1\le j\le p}
\left(
\frac{1}{G_n}
\sum_{g=1}^{G_n}
E|\tilde X_{g,j}|^k
\right)^{1/k}.
\]
By Lemma~\ref{lem:shock-moment}, applied with
$I=\mathcal C_g^\circ$ and $\mathcal U=\mathcal C_g$,
\[
\|\tilde X_{g,j}\|_k
=
\sqrt{\frac{G_n}{n}}\,
\|Y_{g,j,m}\|_k
\le
C_k
\sqrt{\frac{G_n}{n}}\,
l_g^{1/2}
\Delta_{n,0,k,j}.
\]
Therefore,
\[
E|\tilde X_{g,j}|^k
\le
C_k^k
\left(\frac{G_n}{n}\right)^{k/2}
l_g^{k/2}
\Delta_{n,0,k,j}^k.
\]
Taking the maximum over $j$ and averaging over $g$ gives
\[
M_k
\le
C_k
\left[
\left(\frac{G_n}{n}\right)^{k/2}
\frac{1}{G_n}
\sum_{g=1}^{G_n}l_g^{k/2}
\right]^{1/k}
\Psi_{n,k}(0).
\]
In particular, under Assumption~\ref{ass:clusters},
\[
M_3\lesssim \Psi_{n,3}(0),
\qquad
M_4\lesssim \Psi_{n,4}(0),
\]
because $l_g\asymp l_n$ uniformly in $g$ and $G_nl_n\asymp n$.
The desired Gaussian approximation for
$G_n^{-1/2}\sum_{g=1}^{G_n}\tilde X_g$
then follows from the high-dimensional Gaussian approximation theorem of
\cite{CCK2013} for independent, non-identically distributed summands.

It remains to bound the maximal block sizes entering the CCK approximation
error.
We branch into the finite-$q$-th-moment regime and the sub-Weibull regime.

\paragraph{Finite-$q$-th-Moment Regime:}

From the proof above, \[\|Y_{g,j,m}\|_q
\le
C_q
l_g^{1/2}
\Delta_{n,0,q,j}.
\]
By Markov's inequality, for all $j$, we have 
\[
\mathbb P\left(
\left|Y_{g,j,m}\right|>t
\right)
\le
\frac{C_q}{t^q}
l_g^{q/2}\Psi_{n,q}^q(0),
\]
Since
\[
\tilde X_{g,j}
=
\sqrt{\frac{G_n}{n}}Y_{g,j,m},
\]
and $l_g\asymp l_n$, it follows that
\[
P\left(
|d_{jj}^{-1/2}\tilde X_{g,j}|>t
\right)
\le
C_q
\frac{\Psi_{n,q}^q(0)}{t^q}.
\]
Applying a union bound over
$g=1,\ldots,G_n$ and
$j=1,\ldots,p$ gives
\[
P\left(
\max_{g\le G_n}\max_{j\le p}
|d_{jj}^{-1/2}\tilde X_{g,j}|>t
\right)
\le
C_q
\frac{G_np\Psi_{n,q}^q(0)}{t^q}.
\]
Consequently,
\[
u_{Y,m}(\lambda)
\le
C
\Psi_{n,q}(0)
\left(
\frac{G_np}{\lambda}
\right)^{1/q}.
\]

\paragraph{Sub-Weibull Regime:}

By Lemma~\ref{lem:shock-moment},
\[
\|\tilde X_{g,j}\|_q
\le
Cq^{1/\beta}\Phi_{n,\psi_\nu}.
\]
Hence,
\[
P(|d_{jj}^{-1/2}\tilde X_{g,j}|>t)
\le
2\exp\left[
-c
\left(
\frac{t}{\Phi_{n,\psi_\nu}}
\right)^\beta
\right].
\]
Applying a union bound over
$g\le G_n$ and
$j\le p$,
\[
P\left(
\max_{g\le G_n}\max_{j\le p}
|d_{jj}^{-1/2}\tilde X_{g,j}|>t
\right)
\le
2G_np
\exp\left[
-c
\left(
\frac{t}{\Phi_{n,\psi_\nu}}
\right)^\beta
\right].
\]
Therefore,
\[
u_{Y,m}(\lambda)
\le
C
\Phi_{n,\psi_\nu}
\{\log(G_np/\lambda)\}^{1/\beta}.
\]

For the Gaussian term, recall that
\[
\tilde B_g = (\tilde b_{g,ij})_{i,j=1}^p
= \mathrm{Cov}\!\left(\frac{Y_{g,m}}{\sqrt{M_g}}\right)=
\mathrm{Cov}\!\left(\frac{Z_{g,m}}{\sqrt{M_g}}\right) .
\]
Then the CCK normalized Gaussian summand is
\[
\tilde Z_g := \sqrt{\frac{G_n}{n}}\,Z_{g,m},
\qquad
\mathrm{Var}\!\left(\tilde Z_{g,j}\right)
=O (\tilde b_{g,jj}).
\]
Since
\[
c_1 \le \tilde b_{g,jj}/d_{jj} \le c_2
\qquad \text{uniformly in } g,j,
\]
 we obtain
\[
\mathrm{Var}\!\left(d_{jj}^{-1/2}\tilde Z_{g,j}\right) \le C.
\]
Hence, for all $t>0$,
\[
P\left(
|d_{jj}^{-1/2}\tilde Z_{g,j}| > t
\right)
\le
2\exp(-ct^2).
\]
Applying a union bound over $g \le G_n$ and $j \le p$ gives
\[
P\left(
\max_{g \le G_n}\max_{j \le p}
|d_{jj}^{-1/2}\tilde Z_{g,j}| > t
\right)
\le
2G_np \exp(-ct^2).
\]
Therefore,
\[
u_{Z,m}(\lambda)
\le
C\sqrt{\log(G_np/\lambda)}.
\]
This completes a proof.
\qed

\subsection{Proof for Lemma \ref{lem:gtg}}\label{sec:lem:gtg}
Recall
\[
B_g:=Var\!\left(Y_g/\sqrt{M_g}\right)
\qquad\text{and}\qquad
\tilde B_g:=Var\!\left(Y_{g,m}/\sqrt{M_g}\right),
\]
and define the size-weighted averages
\[
B^\circ:=\sum_{g=1}^{G_n}\frac{M_g}{n}B_g
\qquad\text{and}\qquad
\tilde B^\circ:=\sum_{g=1}^{G_n}\frac{M_g}{n}\tilde B_g.
\]
Recall $\Sigma_n=Var(T_X/\sqrt n)$. Let $D=\mathrm{diag}(d_{11},\dots,d_{pp})$, and define
\[
\Sigma^{Z}_n:=Var\!\left(D^{-1/2}\frac{T_Z}{\sqrt n}\right)
\qquad\text{and}\qquad
\Sigma^{Z,m}_n:=Var\!\left(D^{-1/2}\frac{T_{Z,m}}{\sqrt n}\right).\]
We get
\[
\Sigma^{Z}_n= D^{-1/2}B^\circ D^{-1/2}
\qquad\text{and}\qquad
\Sigma^{Z,m}_n= D^{-1/2}\tilde B^\circ D^{-1/2}.
\]
Hence, letting $V:=D^{-1/2}\Sigma_n  D^{-1/2}$, we have
\[
|\Sigma^{Z,m}_n-D^{-1/2}\Sigma_n D^{-1/2}|_\infty
\le
\max_{1\le j\le p}d_{jj}^{-1}\Big(|\tilde B^\circ-B^\circ|_\infty+|B^\circ-\Sigma_n|_\infty\Big).
\]

\subsubsection{ From $\tilde B^\circ$ to $B^\circ$.} 
Fix $g$ and $(j,k)$. Write \[ \tilde b_{g,jk}-b_{g,jk} = \mathrm{Cov}\!\left( \frac{Y_{g,m,j}}{\sqrt{M_g}}, \frac{Y_{g,m,k}}{\sqrt{M_g}} \right) - \mathrm{Cov}\!\left( \frac{Y_{g,j}}{\sqrt{M_g}}, \frac{Y_{g,k}}{\sqrt{M_g}} \right). \] By adding and subtracting cross terms and applying Cauchy--Schwarz, we have \[ |\tilde b_{g,jk}-b_{g,jk}| \le \Big\| \frac{Y_{g,m,j}-Y_{g,j}}{\sqrt{M_g}} \Big\|_2 \Big\| \frac{Y_{g,m,k}}{\sqrt{M_g}} \Big\|_2 + \Big\| \frac{Y_{g,j}}{\sqrt{M_g}} \Big\|_2 \Big\| \frac{Y_{g,m,k}-Y_{g,k}}{\sqrt{M_g}} \Big\|_2. \] Observe that \[ Y_{g,m,j}-Y_{g,j} = \sum_{i\in\mathcal C_g^\circ} \left( X_{n,ij}^{(m)}-X_{n,ij} \right). \] By the same shock-replacement argument used in the proof of
Lemma~\ref{lem:trunc_omega}, for all $j\le p$,
\[
\left\|
Y_{g,m,j}-Y_{g,j}
\right\|_2
\le
CM_g\Psi_{n,2}(m).
\]
By Lemma~\ref{lem:covariance}(ii), for the unlocalized block sum,
\[
\mathrm{Var}(Y_{g,j})
=
\sum_{i\in\mathcal C_g^\circ}
\sum_{i'\in\mathcal C_g^\circ}
\mathrm{Cov}(X_{n,ij},X_{n,i'j}) 
\le
\sum_{i\in\mathcal C_g^\circ}
\sum_{i'\in V_n}
\big|\mathrm{Cov}(X_{n,ij},X_{n,i'j})\big| 
\le
C M_g\Psi_{n,2}^2(0).
\]
Hence,
\[
\Big\|
\frac{Y_{g,j}}{\sqrt{M_g}}
\Big\|_2
\le
C\Psi_{n,2}(0).
\]
The same bound holds for the localized block sum. Indeed, since
\(X_{n,i}^{(m)}=E(X_{n,i}\mid\mathcal F_{i,m})\), the \(L^2\) coupling effect of
\(X_{n,i}^{(m)}\) is no larger than that of \(X_{n,i}\). Hence the covariance
summability bound in Lemma~\ref{lem:covariance}(ii) also applies with
\(X_{n,i}\) replaced by \(X_{n,i}^{(m)}\). Therefore, for all $j\le p$,
\[
\Big\|
\frac{Y_{g,j}}{\sqrt{M_g}}
\Big\|_2
\vee
\Big\|
\frac{Y_{g,m,j}}{\sqrt{M_g}}
\Big\|_2
\le
C\Psi_{n,2}(0).
\]
Consequently,
\[
|\tilde b_{g,jk}-b_{g,jk}|
\le
C\sqrt{M_g}\Psi_{n,2}(m)\Psi_{n,2}(0).
\]
Since
\[
\tilde B^\circ-B^\circ
=
\sum_{g=1}^{G_n}
\frac{M_g}{n}(\tilde B_g-B_g),
\qquad
\sum_{g=1}^{G_n}\frac{M_g}{n}\le 1,
\qquad\text{and}\qquad
M_g\le l_n,
\]
we obtain
\[
|\tilde B^\circ-B^\circ|_\infty
\le
C\sqrt{l_n}\Psi_{n,2}(m)\Psi_{n,2}(0).
\]

\subsubsection{From $B^\circ$ to $\Sigma_n$}
Recall \[
\Sigma_n
:=
\frac1n \mathrm{Var}\!\Big(\sum_{i\in V_n}X_{n,i}\Big)
\qquad\text{and}\qquad
B^\circ
:=
\frac1{n}\sum_{g=1}^{G_n} \mathrm{Var}(Y_g).
\]
Define
\[
\Sigma_n^\circ
:=
\frac1{n} \mathrm{Var}\!\Big(\sum_{i\in V_n^\circ}X_{n,i}\Big).
\]

\paragraph{Step 1. Covariance Decomposition.}

Since
\[
\sum_{i\in V_n^\circ}X_{n,i}
=
\sum_{g=1}^{G_n}Y_g,
\]
we have
\[
B^\circ-\Sigma_n
=
(B^\circ-\Sigma_n^\circ)
+
(\Sigma_n^\circ-\Sigma_n).
\]
The two terms correspond to the cross-cluster covariance and the boundary contribution, respectively.

\paragraph{Step 2. Cross-Cluster Covariance.}

By the bilinearity of covariance,
\[
B^\circ-\Sigma_n^\circ
=
-\frac1n
\sum_{g\neq h}
\mathrm{Cov}(Y_g,Y_h).
\]
Since the interior clusters satisfy
\[
d_n(\mathcal C_g^\circ,\mathcal C_h^\circ)\ge2m+1,
\]
we obtain
\[
\begin{aligned}
|B^\circ-\Sigma_n^\circ|_\infty
&\le
\frac1n
\sum_{g\neq h}
|\mathrm{Cov}(Y_g,Y_h)|_\infty\\
&\le
\frac1n
\sum_{i\in V_n^\circ}
\sum_{i':\,d_n(i,i')\ge2m+1}
\max_{j,k}
|\mathrm{Cov}(X_{n,ij},X_{n,i'k})|.
\end{aligned}
\]
By Lemma~\ref{lem:covariance}(i),
\[
\begin{aligned}
&\sum_{i':\,d_n(i,i')\ge2m+1}
|\mathrm{Cov}(X_{n,ij},X_{n,i'k})|  \\
&\le
C
\sum_{u\in V_n}
\delta_{n,i,u,2,j}
\sum_{i':\,d_n(i,i')\ge2m+1}
\delta_{n,i',u,2,k}.
\end{aligned}
\]
Since
\[
d_n(i,i')\ge2m+1
\quad\Longrightarrow\quad
d_n(i',u)>m
\ \text{or}\
d_n(i,u)>m,
\]
we have
\[
\sum_{i':\,d_n(i,i')\ge2m+1}
\delta_{n,i',u,2,k}
\le
\Delta_{n,m,2,k}, \qquad \sum_{u\in V_n} 
\delta_{n,i,u,2,j}\le \Delta_{n,0,2,k}
\]
or
\[
\sum_{i':\,d_n(i,i')\ge2m+1}
\delta_{n,i',u,2,k}
\le
\Delta_{n,0,2,k}, \qquad \sum_{u\in V_n} 
\delta_{n,i,u,2,j}\le \Delta_{n,m,2,k}
\]
Hence,
\[
\sum_{i':\,d_n(i,i')\ge2m+1}
|\mathrm{Cov}(X_{n,ij},X_{n,i'k})|
\le
C\,
\Psi_{n,2}(m)\Psi_{n,2}(0).
\]
Substituting this bound into the previous one display yields
\[
|B^\circ-\Sigma_n^\circ|_\infty
\lesssim
\Psi_{n,2}(m)\Psi_{n,2}(0).
\]

\paragraph{Step 3. Boundary Contribution.}

Write
\[
\sum_{i\in V_n}X_{n,i}
=
S^\circ+S^\partial,
\]
where
\[
S^\circ
=
\sum_{i\in V_n^\circ}X_{n,i}
\qquad\text{and}\qquad
S^\partial
=
\sum_{i\in V_n\setminus V_n^\circ}X_{n,i}.
\]
Then,
\[
\Sigma_n-\Sigma_n^\circ
=
\frac1n\mathrm{Var}(S^\partial)
+
\frac2n\mathrm{Cov}(S^\circ,S^\partial).
\]
By Lemma~\ref{lem:covariance}(ii),
\[
\begin{aligned}
|\Sigma_n-\Sigma_n^\circ|_\infty
&\le
\frac2n
\sum_{i\in V_n\setminus V_n^\circ}
\sum_{i'\in V_n}
\max_{j,k}
|\mathrm{Cov}(X_{n,ij},X_{n,i'k})|\\
&\lesssim
\partial_n \Psi_{n,2}^2(0)\\
&\lesssim
\partial_n.
\end{aligned}
\]

\subsubsection{Final Step.}
Combining the previous bounds,
\[
|B^\circ-\Sigma_n|_\infty
\lesssim
\sqrt{l_n}\Psi_{n,2}(m)\Psi_{n,2}(0)
+
\partial_n.
\]
Therefore,
\[
\left|
\Sigma_n^{Z,m}
-
D^{-1/2}\Sigma_nD^{-1/2}
\right|_\infty
\lesssim
\max_{1\le j\le p} d_{jj}^{-1}\,
(\sqrt{l_n}\Psi_{n,2}(m)\Psi_{n,2}(0)
+
\partial_n).
\]
Finally, Lemma~3.1 of \cite{CCK2013} implies that the Kolmogorov distance between the corresponding Gaussian distributions is bounded by
\[
\pi\!\left(\max_{1\le j\le p} d_{jj}^{-1}\,
(
\sqrt{l_n}\Psi_{n,2}(m)\Psi_{n,2}(0)
+
\partial_n)
\right).
\]
This completes a proof.
\qed

\subsection{Verification of the Conditions in Lemma~\ref{lem:cluster-ga} }{\label{sec:valid}}
To apply Lemma~\ref{lem:cluster-ga} in Appendix \ref{sec:proof}, it remains to verify its variance
normalization condition
\[
c<\min_{1\le j\le p} d_{jj},
\qquad
c_1 \le \tilde b_{g,(jj)}/d_{jj} \le c_2,
\quad 1\le j\le p, \quad 1\le g\le G_n.
\]
where
\[
(\tilde b_{g,jk})_{j,k=1}^p=\widetilde B_g
=
\mathrm{Cov}\!\left(Y_{g,m}/\sqrt{M_g}\right).
\]
Let
\[
S_{g,j}
:=
\sum_{i\in\mathcal C_g}X_{n,ij}.
\]
By Assumption~\ref{ass:variance lb} and the covariance summability bound in Lemma \ref{lem:covariance},
\[
c_1\sigma_{jj}
\le
\mathrm{Var}(l_g^{-1/2}S_{g,j})
\le
c_2\sigma_{jj},
\qquad
g=1,\ldots,G_n,\quad
j=1,\ldots,p.
\]
where $\mathrm{diag}(\sigma_{jj})_{j=1}^p=\Sigma_0$.
Moreover,
\[
S_{g,j}-Y_{g,m,j}
=
\sum_{i\in\mathcal C_g}
(X_{n,ij}-X_{n,ij}^{(m)})
+
\sum_{i\in\mathcal B_g(m)}
X_{n,ij}^{(m)}.
\]
By the proof for Lemma~\ref{lem:trunc_omega},
\[
\sup_{i\in V_n}
\|X_{n,ij}-X_{n,ij}^{(m)}\|_2
\le
\Psi_{n,2}(m),
\]
and therefore
\[
\left\|
l_g^{-1/2}
\sum_{i\in\mathcal C_g}
(X_{n,ij}-X_{n,ij}^{(m)})
\right\|_2
\le
l_g^{1/2}\Psi_{n,2}(m).
\]
For the buffer term, Lemma~\ref{lem:covariance}(ii) yields
\[
\left\|
l_g^{-1/2}
\sum_{i\in\mathcal B_g(m)}
X_{n,ij}^{(m)}
\right\|_2^2
\lesssim
\frac{|\partial\mathcal C_g|\,N_n(2m)}
{l_g}
\Psi_{n,2}^2(0).
\]
Since
\[
m=\lceil C_m\log(np)\rceil
\]
with $C_m$ sufficiently large,
\[
l_g^{1/2}\Psi_{n,2}(m)=o(1),
\]
and the condition (ii) in theorem \ref{th:GA_q} and \ref{th:GA_sub} implies
\[
\max_{g\le G_n}
\frac{|\partial\mathcal C_g|\,N_n(2m)}
{l_g}
=o(1),
\]
we obtain
\[
\max_{g\le G_n}
\max_{j\le p}
\left\|
l_g^{-1/2}S_{g,j}
-
l_g^{-1/2}Y_{g,m,j}
\right\|_2
=o(1).
\]
Also, by
 \[
\max_{g\le G_n}
\frac{|\partial\mathcal C_g|\,N_n(2m)}
{l_g}
=o(1),
\]
we can get \(
M_g\asymp l_g.
\) 
Then it follows that
\[
Var(M_g^{-1/2}Y_{g,m,j})
\asymp
Var(l_g^{-1/2}Y_{g,m,j})
\asymp
Var(l_g^{-1/2}S_{g,j})
\asymp
\sigma_{jj},
\]
uniformly over \(g=1,\ldots,G_n\) and \(j=1,\ldots,p\). Hence,
\[
0<c_1'
\le
\frac{\tilde b_{g,jj}}
{\sigma_{jj}}
\le
c_2'
<\infty,
\qquad
g=1,\ldots,G_n,\quad
j=1,\ldots,p.
\]
Therefore, the variance normalization condition required in
Lemma~\ref{lem:cluster-ga} is satisfied.

\section{Proof of the Theorems in Section \ref{sec:estimation-cov}}
\subsection{Quadratic-Form Coupling Lemma}
\begin{lemma}[Quadratic-Form Coupling Bound]
\label{lem:hac-coupling}

Let
\[
S_{ab}
=
\sum_{i\in V_n}
\sum_{i'\in V_n}
w_{ii'}\,
X_{n,ia}X_{n,i'b},
\]
where the weights satisfy
\[
|w_{ii'}|\le1
\qquad\text{and}\qquad
w_{ii'}=0
\quad\text{whenever }d_n(i,i')>b_n.
\]
For each shock $u\in V_n$, let
$S_{ab}^{*(u)}$
denote the coupled version obtained by replacing
$\varepsilon_{n,u}$
with an independent copy while leaving all other shocks unchanged.
Then for every $r\ge2$,
\[
\|S_{ab}-S_{ab}^{*(u)}\|_r
\le
C
N_n(b_n)
\Psi_{n,2r}^2(0).
\]
\end{lemma}

\begin{proof}
Using
\[
XY-X'Y'
=
(X-X')Y
+
X'(Y-Y'),
\]
we obtain
\[
S_{ab}-S_{ab}^{*(u)}
=
\sum_{i,i'}
w_{ii'}
\Big[
(X_{n,ia}-X_{n,ia}^{*(u)})X_{n,i'b}
+
X_{n,ia}^{*(u)}
(X_{n,i'b}-X_{n,i'b}^{*(u)})
\Big].
\]
Applying the triangle inequality and Hölder's inequality with
exponents $(2r,2r)$ gives
\[
\|S_{ab}-S_{ab}^{*(u)}\|_r
\le
\sum_{i,i'}
|w_{ii'}|
\Big(
\|X_{n,ia}-X_{n,ia}^{*(u)}\|_{2r}
\|X_{n,i'b}\|_{2r}
+
\|X_{n,ia}^{*(u)}\|_{2r}
\|X_{n,i'b}-X_{n,i'b}^{*(u)}\|_{2r}
\Big).
\]
Since
\[
X_{n,i}^{*(u)}
\stackrel{d}{=}
X_{n,i},
\]
and
\[
\|X_{n,ij}\|_{2r}
\le
\Psi_{n,2r}(0),
\]
we obtain
\[
\|S_{ab}-S_{ab}^{*(u)}\|_r
\le
2\Psi_{n,2r}(0)
\sum_{i,i'}
|w_{ii'}|
\,
\delta_{n,i,u,2r,a}.
\]
Because
\[
w_{ii'}=0
\quad\text{if }d_n(i,i')>b_n,
\]
each row contains at most
\(N_n(b_n)\)
nonzero entries. Hence
\[
\sum_{i'}
|w_{ii'}|
\le
N_n(b_n),
\]
which yields
\[
\|S_{ab}-S_{ab}^{*(u)}\|_r
\le
2
N_n(b_n)
\Psi_{n,2r}(0)
\sum_{i}
\delta_{n,i,u,2r,a}.
\]

Finally,
\[
\sum_i
\delta_{n,i,u,2r,a}
\le
\Psi_{n,2r}(0),
\]
by the definition of
\(\Psi_{n,2r}(0)\).
Therefore,
\[
\|S_{ab}-S_{ab}^{*(u)}\|_r
\le
C
N_n(b_n)
\Psi_{n,2r}^2(0).
\]
This completes a proof.
\end{proof}

\subsection{Proof for Theorem \ref{lem:bias}: Bias of the HAC estimator }
Recall
\[
\Gamma_n(s)
=
\frac1n
\sum_{i\in V_n}
\sum_{j\in \mathcal{N}_n^\partial(i,s)}
\mathbb{E}[X_{n,i}X_{n,j}^\top].
\]
Then,
\[
\mathbb{E}(\widehat\Sigma_n(b_n))
=
\sum_{s\ge0}w(s/b_n)\Gamma_n(s)
\qquad\text{and}
\Sigma_n=\sum_{s\ge0}\Gamma_n(s).
\]
Hence,
\[
\mathbb E(\widehat\Sigma_n(b_n))-\Sigma_n
=
\sum_{s\ge1}
\{w(s/b_n)-1\}\Gamma_n(s),
\]
so that
\[
\big|
\mathbb E(\widehat\Sigma_n(b_n))-\Sigma_n
\big|_\infty
\le
\sum_{s\ge1}
|w(s/b_n)-1|
\,|\Gamma_n(s)|_\infty .
\]
It therefore remains to bound the shell covariance
\(|\Gamma_n(s)|_\infty\).

\begin{lemma}[Shell Covariance Bound]
\label{lem:control-Afd}
Under Assumptions \ref{ass:decay rate} and ~\ref{ass:volume},
\[
|\Gamma_n(s)|_\infty
\le
CA_{n,s}^{\mathrm{FD}},
\]
where
\[
A_{n,s}^{\mathrm{FD}}
:=
\max_{a,b\le p}
\frac1n
\sum_{i\in V_n}
\sum_{j\in\mathcal N_n^\partial(i,s)}
\sum_{u\in V_n}
\delta_{n,i,u,2,a}
\delta_{n,j,u,2,b}.
\]
Moreover,
\[
A_{n,s}^{\mathrm{FD}}
\le
C(1+s)^{2d+1}\rho^s,
\]
and consequently,
\[
\sum_{s\ge1}
A_{n,s}^{\mathrm{FD}}
<\infty.
\]
\end{lemma}

\begin{proof}
By Lemma~\ref{lem:covariance}(i), for every $a,b\le p$,
\[
\left|
E\!\left[X_{n,i,a}X_{n,i',b}\right]
\right|
\le
C
\sum_{u\in V_n}
\delta_{n,i,u,2,a}
\delta_{n,i',u,2,b}.
\]
Therefore,
\[
|\Gamma_n(s)|_\infty
\le
C A_{n,s}^{\mathrm{FD}}.
\]
By Assumption~\ref{ass:decay rate},
\[
A_{n,s}^{\mathrm{FD}}
\le
C_2^2
\frac1n
\sum_{i\in V_n}
\sum_{i'\in \mathcal{N}_n^\partial(i,s)}
\sum_{u\in V_n}
\rho^{d_n(i,u)}\rho^{d_n(i',u)}.
\]
Rearranging the sums gives
\[
A_{n,s}^{\mathrm{FD}}
\le
C_2^2
\frac1n
\sum_{u\in V_n}
\sum_{i\in V_n}
\rho^{d_n(i,u)}
\sum_{i'\in \mathcal{N}_n^\partial(i,s)}
\rho^{d_n(i',u)}.
\]
Fix \(u\) and \(i\), and write \(r=d_n(i,u)\). If \(i'\in \mathcal{N}_n^\partial(i,s)\), then
by the triangle inequality,
\[
d_n(i',u)\ge |s-r|.
\]
Therefore, by Assumption~\ref{ass:volume},
\[
\sum_{i'\in \mathcal{N}_n^\partial(i,s)}
\rho^{d_n(i',u)}
\le
N_n^\partial(i,s)\rho^{|s-r|}
\le
C_N(1+s)^{d}\rho^{|s-r|}.
\]
Hence
\[
A_{n,s}^{\mathrm{FD}}
\le
C
(1+s)^{d}
\frac1n
\sum_{u\in V_n}
\sum_{i\in V_n}
\rho^{d_n(i,u)}\rho^{|s-d_n(i,u)|}.
\]
Grouping \(i\)'s according to \(r=d_n(i,u)\), we get
\[
\sum_{i\in V_n}
\rho^{d_n(i,u)}\rho^{|s-d_n(i,u)|}
=
\sum_{r\ge0}
\sum_{i\in \mathcal{N}_n^\partial(u,r)}
\rho^r\rho^{|s-r|}.
\]
By polynomial shell growth  Assumption~\ref{ass:volume},
\[
\sum_{r\ge0}
\sum_{i\in \mathcal{N}_n^\partial(u,r)}
\rho^r\rho^{|s-r|}
\le
C
\sum_{r\ge0}
(1+r)^{d}\rho^r\rho^{|s-r|}.
\]
The last sum satisfies
\[
\sum_{r\ge0}
(1+r)^{d}\rho^r\rho^{|s-r|}
\le
C(1+s)^{d+1}\rho^s.
\]
Thus,
\[
A_{n,s}^{\mathrm{FD}}
\le
C(1+s)^{d}(1+s)^{d+1}\rho^s
=
C(1+s)^{2d+1}\rho^s.
\]
Since a polynomial times a geometrically decaying sequence is summable,
\[
\sum_{s\ge1}A_{n,s}^{\mathrm{FD}}<\infty.
\]

By Lemma~\ref{lem:control-Afd},
\[
\big|
E(\widehat\Sigma_n(b_n))-\Sigma_n
\big|_\infty
\le
C
\sum_{s\ge1}
|w(s/b_n)-1|
A_{n,s}^{\mathrm{FD}}.
\]
Since \(w(x)=0\) for \(x>1\), we decompose
\[
\sum_{s\ge1}
|w(s/b_n)-1|
A_{n,s}^{\mathrm{FD}}
=
I_{n,1}+I_{n,2},
\]
where
\[
I_{n,1}
=
\sum_{1\le s\le b_n}
|w(s/b_n)-1|
A_{n,s}^{\mathrm{FD}}
\]
and
\[
I_{n,2}
=
\sum_{s>b_n}
A_{n,s}^{\mathrm{FD}}.
\]

\paragraph{Step 1. Kernel Approximation Error.}

By Assumption~\ref{ass:kernel},
\[
|w(s/b_n)-1|
\le
C
\left(\frac{s}{b_n}\right)^\kappa,
\qquad
1\le s\le b_n.
\]
Hence,
\[
\begin{aligned}
I_{n,1}
&\le
Cb_n^{-\kappa}
\sum_{s=1}^{b_n}
s^\kappa
A_{n,s}^{\mathrm{FD}}   \\
&\le
Cb_n^{-\kappa}
\sum_{s=1}^{b_n}
(1+s)^{\kappa+2d+1}\rho^s.
\end{aligned}
\]
Since
\[
\sum_{s\ge1}
(1+s)^{\kappa+2d+1}\rho^s
<\infty,
\]
we obtain
\[
I_{n,1}
=
O(b_n^{-\kappa}).
\]

\paragraph{Step 2. Tail Truncation Error.}

Let \(r=2d+1\). By Lemma~\ref{lem:control-Afd},
\[
I_{n,2}
\le
C
\sum_{s>b_n}
(1+s)^r\rho^s.
\]
Writing \(s=b_n+k\) gives
\[
\begin{aligned}
I_{n,2}
&\le
C\rho^{b_n}
\sum_{k\ge1}
(1+b_n+k)^r\rho^k   \\
&\le
C(1+b_n)^r\rho^{b_n}
\sum_{k\ge1}
(1+k)^r\rho^k,
\end{aligned}
\]
where we use
\[
(1+b_n+k)^r
\le
C(1+b_n)^r(1+k)^r.
\]
Since
\[
\sum_{k\ge1}(1+k)^r\rho^k<\infty,
\]
it follows that
\[
I_{n,2}
=
O(b_n^r\rho^{b_n}).
\]

Combining the bounds for \(I_{n,1}\) and \(I_{n,2}\),
\[
\big|
E(\widehat\Sigma_n(b_n))-\Sigma_n
\big|_\infty
=
O(b_n^{-\kappa})
+
O(b_n^r\rho^{b_n})
=
o(1),
\]
since \(b_n\to\infty\).
\end{proof}

\subsection{Proofs of Theorems \ref{lem:hac-doob-variance_q} and \ref{lem:hac-doob-variance_subwei} }
Fix a pair of coordinates $(a,b)\in[p]^2$ and write the covariance matrix in the form \[ \widehat\Sigma_{n,ab}(b_n) = \frac1n \sum_{i\in V_n} \sum_{i'\in V_n} w_{ii'}(b_n)X_{n,ia}X_{n,i'b}, \qquad w_{ii'}(b_n) = w\!\left(\frac{d_n(i,i')}{b_n}\right). \] Define \[
Q_{ab}:=n\Big(\widehat\Sigma_{n,ab}(b_n)-\mathbb E\widehat\Sigma_{n,ab}(b_n)\Big).
\] 
Then
\[
Q_{ab}=\sum_{i\in V_n}\sum_{i'\in V_n} w_{ii'}(b_n)\Big(X_{n,ia}X_{n,i'b}-\mathbb E[X_{n,ia}X_{n,i'b}]\Big)=S_{ab}-\mathbb E(S_{ab}),
\]
where $S_{ab}$ is defined in Lemma \ref{lem:hac-coupling}.

Fix an arbitrary ordering of the shock field
\(
\varepsilon_n=\{\varepsilon_{n,1},\ldots,\varepsilon_{n,n}\}
\)
and let
\[
\mathcal F_k
=
\sigma(\varepsilon_{n,1},\ldots,\varepsilon_{n,k}).
\]
Define
\[
M_k
=
E(S_{ab}\mid\mathcal F_k),
\qquad
D_k
=
M_k-M_{k-1}.
\]
Since
\[
M_n=S_{ab},
\qquad
M_0=E(S_{ab}),
\]
we have
\[
Q_{ab}
=
S_{ab}-E(S_{ab})
=
M_n-M_0
=
\sum_{k=1}^nD_k.
\]
Moreover, by conditional Jensen's inequality,
\[
\|D_k\|_r
\le
\|S_{ab}-S_{ab}^{*(k)}\|_r.
\]
Applying Lemma~\ref{lem:hac-coupling},
\[
\|D_k\|_r
\le
CN_n(b_n)\Psi_{n,2r}^2(0),
\qquad
r\ge1.
\]

\paragraph{(i) Finite-$q$-th-Moment Regime:}
By Burkholder’s inequality, 
\[
\|Q_{ab}\|_{q/2}
\le
C_q
\left(
\sum_{k=1}^n
\|D_k\|_{q/2}^2
\right)^{1/2} 
\le
C_q
n^{1/2}
N_n(b_n)
\Psi_{n,q}^2(0).
\]
Therefore, by Markov's inequality, for every \(x>0\),
\[
\mathbb P\left(
|Q_{ab}|>x
\right)
\le
C_q
\frac{
n^{q/4}
N_n(b_n)^{q/2}
\Psi_{n,q}^{q}(0)
}{
x^{q/2}
}.
\]
Applying a union bound over \((a,b)\in[p]^2\), we obtain
\[
\mathbb P\left(
n\left|
\widehat\Sigma_n(b_n)
-
\mathbb E\widehat\Sigma_n(b_n)
\right|_\infty
>x
\right)
\le
C_q
\frac{
p^2
n^{q/4}
N_n(b_n)^{q/2}
\Psi_{n,q}^{q}(0)
}{
x^{q/2}
}.
\]

\paragraph{(ii) Sub-Weibull Regime:}
By Burkholder’s inequality, 
\[
\|Q_{ab}\|_{q}
\le\sqrt {q-1}
\left\|
\left(
\sum_{k=1}^n
\mathbb E[D_k^2\mid \mathcal F_{k-1}]
\right)^{1/2}
\right\|_{q}\le \sqrt {q-1}
\left(
\sum_{k=1}^n
\|D_k\|_{q}^2
\right)^{1/2} 
\le \sqrt q\sqrt n N_n(b_n)\Psi^2_{n,2q}(0).
\]
Also,
\[
\Psi_{n,2q}(0)\le\Phi_{n,\psi_{\nu}}\times (2q)^\nu
\]
Therefore \[
\frac{\|Q_{ab}\|_{q}}{q^{1+2v}}\le \sqrt n N_n(b_n) \Phi_{n,\psi_{\nu}}^2
\]
Therefore we get the tail probability
   \[
\mathbb P\!\left(
|Q_{ab}|
\ge x
\right)
\le
p\exp\!\left(-c\,\frac{x^\gamma}{\left(\sqrt n N_n(m)\Phi_{n,\psi_{\nu}}^2\right)^\gamma}\right),
\]
Furthermore, for infinite norm of covariance matrix
    \[
\mathbb P\!\left(
n\bigl|\widehat\Sigma_n(m)-\mathbb E\widehat\Sigma_n(m)\bigr|_{\infty}
\ge x
\right)
\le
p^2\exp\!\left(-c\,\frac{x^\gamma}{\left(\sqrt n N_n(m)\Phi_{n,\psi_{\nu}}^2\right)^\gamma}\right).
\]
where $\gamma=\frac{1}{1+2\nu}$.

\subsection{Proof of Corollary ~\ref{cor:sigmafea}: Feasible HAC Estimator}

Recall the feasible HAC estimator when we do not know the expectation 
\[
\widetilde\Sigma_n(b_n)=\sum_{s\ge 0}w_n(s)\widetilde\Gamma_n(s),
\]
where
\[
\widetilde\Gamma_n(s)
=
\frac{1}{n}
\sum_{i\in V_n}
\sum_{i'\in \mathcal{N}_n^\partial(i,s)}
(X_{n,i}-\bar X_n)(X_{n,i'}-\bar X_n)^\top .
\]
Let
\[
Y_{n,i}=X_{n,i}-\mu_0,\qquad 
\bar Y_n=\bar X_n-\mu_0 .
\]
Then
\[
X_{n,i}-\bar X_n
=
Y_{n,i}-\bar Y_n .
\]
We can rewrite the feasible HAC estimator as 
\[
\widetilde\Sigma_{n}(b_n)=\frac1n\sum_{i\in V_n}\sum_{i'\in V_n} w_{ii'}(b_n)\,(Y_{n,i}-\bar Y_n)(Y_{n,i'}-\bar Y_n)^\top,
\]
Expanding the product yields
\[
\begin{aligned}
\widetilde\Sigma_n-\widehat\Sigma_n
=&
-\frac1n
\sum_{i,i'}
w_{ii'}(b_n)
Y_{n,i}\bar Y_n^\top
-\frac1n
\sum_{i,i'}
w_{ii'}(b_n)
\bar Y_nY_{n,i'}^\top  \\
&+
\frac1n
\sum_{i,i'}
w_{ii'}(b_n)
\bar Y_n\bar Y_n^\top .
\end{aligned}
\]

For the first term, define the row-sum weights
\[
r_{n,i}:=\sum_{i'\in V_n}w_{ii'}(b_n).
\]
Then
\[
\frac1n
\sum_{i,i'\in V_n}
w_{ii'}(b_n)
Y_{n,i}\bar Y_n^\top
=
\left(
\frac1n
\sum_{i\in V_n}
r_{n,i}Y_{n,i}
\right)
\bar Y_n^\top .
\]
Moreover,
\[
|r_{n,i}|
\le
\sum_{i'\in V_n}|w_{ii'}(b_n)|
\le
N_n(b_n).
\]
Since \(r_{n,i}\) depends only on the network and bandwidth, it is deterministic
conditional on the network. 

We use the same Doob martingale argument as above. We define $S_{r,j}=\sum_{i \in V_n}r_{n,i}Y_{n,ij}$. Let $\{\mathcal F_{n,u}\}_{u\in V_n}$ be the filtration generated by the shock field on all the nodes $\varepsilon_n$. We define the Doob  martingale differences \[
D_{n,u,r,j}
:=\mathbb E\!\left[ S_{r,j}-S_{r,j}^{*(u)}\mid \mathcal F_{n,u}\right],
\qquad u\in V_n.
\]
And $S_{r,j}=\sum_{u\in V_n}D_{n,u,r,j}$.

Since \(r_{n,i}\) is deterministic conditional on the network, we have
\[
S_{r,j}-S_{r,j}^{*(u)}
=
\sum_{i\in V_n}r_{n,i}
\left(Y_{n,ij}-Y_{n,ij}^{*(u)}\right).
\]
Hence,
\[
\begin{aligned}
\|D_{n,u,r,j}\|_q
&=
\left\|
\mathbb E\left[S_{r,j}-S_{r,j}^{*(u)}\mid \mathcal F_{n,u}\right]
\right\|_q  \\
&\le
\left\|
S_{r,j}-S_{r,j}^{*(u)}
\right\|_q  \\
&\le
\sum_{i\in V_n}
|r_{n,i}|
\left\|
Y_{n,ij}-Y_{n,ij}^{*(u)}
\right\|_q .
\end{aligned}
\]
Since
\[
|r_{n,i}|\le N_n(b_n),
\]
we obtain
\[
\|D_{n,u,r,j}\|_q
\le
N_n(b_n)
\sum_{i\in V_n}
\left\|
Y_{n,ij}-Y_{n,ij}^{*(u)}
\right\|_q .
\]
Therefore,
\[
\|D_{n,u,r,j}\|_q
\le
N_n(b_n)
\sum_{i\in V_n}
\delta_{n,i,u,q,j}  
\le
N_n(b_n)
\Psi_{n,q}(0)
\]

Now, since \(\{D_{n,u,r,j},\mathcal F_{n,u}\}_{u\in V_n}\) is a martingale
difference array, Burkholder's inequality gives
\[
\begin{aligned}
\|S_{r,j}\|_q
&=
\left\|
\sum_{u\in V_n}D_{n,u,r,j}
\right\|_q  \\
&\le
Cq^{1/2}
\left(
\sum_{u\in V_n}
\|D_{n,u,r,j}\|_q^2
\right)^{1/2}  \\
&\le
Cq^{1/2}
\sqrt n\,N_n(b_n)\Psi_{n,q}(0).
\end{aligned}
\]

\paragraph{(i) Finite-$q$-th-Moment Regime:} 
Under fixed $q$,
\[
\left\|
\frac1n S_{r,j}
\right\|_q
\le
C_q
N_n(b_n)\Psi_{n,q}(0)n^{-1/2}.
\]
Taking the maximum over \(j\le p\), we use
\[
\max_{1\le j\le p}|x_j|
\le
\left(\sum_{j=1}^p |x_j|^q\right)^{1/q}.
\]
Thus,
\[
\left\|
\left|
\frac1n\sum_{i\in V_n}r_{n,i}Y_{n,i}
\right|_\infty
\right\|_q
\le
\left(
\sum_{j=1}^p
\left\|
\frac1n S_{r,j}
\right\|_q^q
\right)^{1/q} 
\le
C_q
p^{1/q}
N_n(b_n)\Psi_{n,q}(0)n^{-1/2}.
\]
By Markov's inequality,
\[
\left|
\frac1n\sum_{i\in V_n}r_{n,i}Y_{n,i}
\right|_\infty
=
O_p\left(
N_n(b_n)\Psi_{n,q}(0)p^{1/q}n^{-1/2}
\right).
\]
Similarly, \[|\bar Y_n|_\infty=O_p(\Psi_{n,q}(0)p^{1/q}n^{-1/2})\]
Combining the three terms yields
\[
|\widetilde\Sigma_n-\widehat\Sigma_n|_\infty
=
O_p\left(
N_n(b_n)\Psi_{n,q}^2(0) p^{2/q}n^{-1}
\right).
\]

\paragraph{(ii) Sub-Weibull Regime:}
Burkholder's inequality gives 
\[
\|S_{r,j}\|_q
\le
\sqrt{q}
\sqrt n\,N_n(b_n)\Psi_{n,q}(0).
\]
Since
\[
\Psi_{n,q}(0)
\le
q^\nu\Phi_{n,\psi_\nu},
\]
we have
\[
\|S_{r,j}\|_q
\le
Cq^{1/2+\nu}
\sqrt n\,N_n(b_n)\Phi_{n,\psi_\nu}
=
Cq^{1/\beta}
\sqrt n\,N_n(b_n)\Phi_{n,\psi_\nu},
\]
where
\(
\beta=\frac{2}{1+2\nu}.
\)
Hence,
\[
\left\|
\frac1nS_{r,j}
\right\|_{\psi_{1/\beta}}
\le
Cn^{-1/2}
N_n(b_n)
\Phi_{n,\psi_\nu}.
\]
Applying the maximal inequality for sub-Weibull random variables yields
\[
\left|
\frac1n
\sum_{i\in V_n}
r_{n,i}Y_{n,i}
\right|_\infty
=
O_p\!\left(
N_n(b_n)
\Phi_{n,\psi_\nu}
(\log p)^{1/\beta}
n^{-1/2}
\right).
\]
Similarly, \[|\bar Y_n|_\infty=O_p(\Phi_{n,\psi_\nu}\log p^{1/\beta}n^{-1/2}).\]
Combining the three terms yields
\[
|\widetilde\Sigma_n-\widehat\Sigma_n|_\infty
=
O_p\left(
N_n(b_n)\Phi_{n,\psi_\nu}^2 \log p^{1/\gamma}n^{-1}
\right),
\]
where $\gamma=\frac{1}{1+2\nu}$.
\qed
\end{document}